\documentclass[10pt, a4paper ]{article}
\usepackage{setspace}

\usepackage{authblk}
\usepackage[numbers,sort&compress, super]{natbib}
\usepackage{mathtools}
\usepackage[a4paper, total={6.5in, 9in}]{geometry}
\usepackage[justification=justified]{caption}
\usepackage{amsmath}
\usepackage{amsthm, mathtools}
\usepackage{multicol, blindtext}
\usepackage{amssymb}
\usepackage[shortlabels]{enumitem}
\usepackage{qcircuit}
\usepackage[english]{babel}
\usepackage{algorithm}
\usepackage{dsfont}
\usepackage{color} 
\usepackage[framemethod=TikZ]{mdframed}
\usepackage[framemethod=TikZ]{mdframed}
\usepackage{algorithm}
\usepackage{algpseudocode}
\usepackage{braket}
\usepackage{etoolbox}
\apptocmd{\sloppy}{\hbadness 10000\relax}{}{}
\usepackage[labeled]{multibib}
\newcites{S}{Supplementary Material References}
\usepackage{url}

\usepackage{tikz}
\usepackage{graphicx}
\usetikzlibrary{positioning}
\usepackage[breaklinks=true,colorlinks=true,linkcolor=blue,urlcolor=blue,citecolor=blue]{hyperref}

\usepackage{soul}

\usepackage[utf8]{inputenc}
\usepackage{pgfplots}
\pgfplotsset{compat=newest}
\usepgfplotslibrary{groupplots}
\usepgfplotslibrary{dateplot}

\usepackage{graphicx}
\usepackage{subcaption}

\usepackage{tocloft}

\usepackage{etoolbox}

\newcommand{\defref}[1]{Definition \ref{#1}}
\newcommand{\secref}[1]{Section \ref{#1}}
\newcommand{\appref}[1]{Supplementary Material Section \ref{#1}}
\newcommand{\theref}[1]{Theorem \ref{#1}}
\newcommand{\figref}[1]{Fig.~\ref{#1}}
\renewcommand{\eqref}[1]{equation (\ref{#1})}

\newcommand{\suppfigref}[1]{Supplementary Figure~\ref{#1}}

\newcommand{\comment}[1]{\footnote{}}

\newcommand{\tr}{\textsf{tr}}

\newcommand{\BQP}{\textsf{BQP}}
\newcommand{\BPP}{\textsf{BPP}}
\newcommand{\IBM}{\textsf{QCIBM}}
\newcommand{\QCBM}{\textsf{QCBM}}

\DeclareMathOperator*{\argmin}{arg\,min}
\DeclareMathOperator{\MMD}{\textsf{MMD}}
\DeclareMathOperator{\SD}{\textsf{SD}}
\DeclareMathOperator{\SH}{\textsf{SHD}}
\DeclareMathOperator{\W}{\textsf{W}}
\DeclareMathOperator{\OT}{\textsf{OT}}
\DeclareMathOperator{\KL}{\textsf{KL}}

\DeclareMathOperator{\TV}{\textsf{TV}}
\DeclareMathOperator{\BosonSampling}{\textsf{BosonSampling}}

\DeclareMathOperator{\QAOA}{\textsf{QAOA}}
\DeclareMathOperator{\PH}{\textsf{PH}}
\DeclareMathOperator{\NP}{\textsf{NP}}

\DeclareMathOperator{\EVAL}{\textsf{EVAL}}
\DeclareMathOperator{\GEN}{\textsf{GEN}}

\DeclareMathOperator{\C}{\textsf{C}}

\DeclareMathOperator{\IQP}{\textsf{IQP}}
\DeclareMathOperator{\PP}{\textsf{PP}}

\DeclareMathOperator{\PQAOA}{\textsf{PostQAOA}}
\DeclareMathOperator{\PIQP}{\textsf{PostIQP}}

\DeclareMathOperator{\PBQP}{\textsf{PostBQP}}
\DeclareMathOperator{\PBPP}{\textsf{PostBPP}}

\DeclareMathOperator{\Pee}{\textsf{P}}

\newtheorem{theorem}{Theorem}[section]
\newtheorem{corollary}{Corollary}[theorem]
\newtheorem{definition}{Definition}[section]

\newcommand{\revneg}{\mathrel{\reflectbox{\rotatebox[origin=c]{180}{$\neg$}}}}

\DeclareUnicodeCharacter{2212}{-}

\definecolor{ForestGreen}{RGB}{34, 139, 34}
\definecolor{DeepSkyBlue}{RGB}{0, 191,255}
\definecolor{Lavender}{RGB}{230, 230, 250}
\newcommand\crule[3][black]{\textcolor{#1}{\rule{#2}{#3}}}

\newcommand{\red}[1]{{#1}}

\setstcolor{red}

\newcommand{\threeqqvm}{%
    \begin{tikzpicture}[transform canvas={scale=0.09}]
\draw[fill=gray] (0,0) circle (20pt);
\draw[fill=gray] (4,0) circle (20pt);
\draw[fill=gray] (2,3) circle (20pt);
\draw[gray, thick] (0,0) -- (4,0) -- (0,0) -- (2,3) -- (4,0) -- (2,3);
\end{tikzpicture}
}

\newcommand{\fourqqvm}{%
\begin{tikzpicture}[transform canvas={scale=0.08}]
\draw[fill=gray] (0,0) circle (20pt);
\draw[fill=gray] (4,0) circle (20pt);
\draw[fill=gray] (0,4) circle (20pt);
\draw[fill=gray] (4,4) circle (20pt);
\draw[gray, thick] (0,0) -- (4,0) -- (4,4) -- (0,4) -- (0,0) -- (4,4) -- (4,0) -- (0,4) ;
\end{tikzpicture}
}

\title{The Born Supremacy: Quantum Advantage and Training of an Ising Born Machine}

\author[1, *]{Brian Coyle}
\author[1]{Daniel Mills}
\author[1, 2]{Vincent Danos}
\author[1, 3]{Elham Kashefi}
\affil[1]{School of Informatics, 10 Crichton Street, University of Edinburgh, United Kingdom.}
\affil[2] {CNRS, \'{E}cole Normale Sup\'{e}rieure, Paris, 75005 Paris, France.}
\affil[3] {CNRS, LIP6, Sorbonne Universit\'{e}, 4 place Jussieu, 75005 Paris, France}
\affil[ ]{*Corresponding author: brian.coyle@ed.ac.uk}

\date{}

\begin{document}


\maketitle


\begin{abstract}
    The search for an application of near-term quantum devices is widespread. Quantum machine learning is touted as a potential utilisation of such devices, particularly those out of reach of the simulation capabilities of classical computers. In this work, we study such an application in generative modelling, focusing on a class of quantum circuits known as Born machines. Specifically, we define a subset of this class based on Ising Hamiltonians, and show that the circuits encountered during gradient-based training cannot be efficiently sampled from classically, up to multiplicative error in the worst case. Our gradient-based training methods use cost functions known as the Sinkhorn divergence and the Stein discrepancy which have not previously been used in the gradient based training of quantum circuits, and we also introduce quantum kernels to generative modelling. We show that these methods outperform the previous standard method, which used maximum mean discrepancy ($\MMD$) as a cost function, and achieve this with minimal overhead. Finally, we discuss the ability of the model to learn hard distributions and provide formal definitions for `quantum learning supremacy'. We also exemplify the work of this paper by using generative modelling to perform quantum circuit compilation.
\end{abstract}

\begin{multicols}{2}

\section*{Introduction} \label{sec:intro}

As quantum devices with $\sim 80-200$ qubits, but without fault tolerance, begin to be built, we near the dawn of the Noisy Intermediate Scale Quantum (NISQ) \cite{preskill_quantum_2018} technology era. Because of the low number of qubits, the limited connectivity between them, \red{and the low circuit depth permitted by low decoherence times,} these devices cannot perform many of the most famous algorithms thought to demonstrate exponential speedups over classical algorithms \cite{shor_polynomial-time_1997, harrow_quantum_2009}.

In spite of this, NISQ devices could provide efficient solutions to other problems which cannot be solved in polynomial time by classical means. Showing this to be true is referred to as a demonstration of \textit{quantum computational supremacy}\cite{bremner_classical_2011, gao_quantum_2017, bremner_average-case_2016,  aaronson_computational_2013, farhi_quantum_2016, boixo_characterizing_2018}\red{, with the first such experimental realisation occurring recently \cite{arute_quantum_2019}}.

Proposals for demonstrations of quantum computational supremacy on NISQ technology typically involve sampling from the output distribution of random quantum circuits. While a realisation of such an advantage is of great theoretical importance, generating random samples is not obviously independently interesting. We incorporate this sampling into a \textit{useful} application, keeping the provable quantum advantage, but in a context with more practical applicability. 

Specifically, we explore generative modelling in quantum machine learning (QML) which is the task of \textit{generalising} from a finite set of samples, $\{\mathbf{y}\}^M$, drawn from a data set. By learning the underlying probability distribution from which these samples are drawn, $\pi(\mathbf{y})$, a model should be able to generate new samples from said distribution. 

Generative models range from simple Na{\"i}ve Bayes\cite{maron_automatic_1961} models, to complicated neural networks, like generative adversarial networks (GANs)\cite{goodfellow_generative_2014}. The intrinsic randomness inherent in quantum mechanics allows for the definition of a new class of generative models which are without a classical analogue. Known as \textit{Born machines}\cite{cheng_information_2017, liu_differentiable_2018, benedetti_generative_2019}, they have the ability to produce statistics according to Born's measurement rule. Specifically, for a state $\ket{\psi}$, a measurement produces a sample $\mathbf{x} \sim p(\mathbf{x}) = |\braket{\mathbf{x}|\psi}|^2 \label{bornrule}$. There are several variants, including Bayesian approaches\cite{du_expressive_2018}, adversarial training methods \cite{zeng_learning_2019}, and adaptations to continuous distributions\cite{romero_variational_2019}. 

Quantum circuit Born machines ($\QCBM$) are a subclass of parametrized quantum circuits (PQC) and are widely applicable (see ref. \cite{benedetti_parameterized_2019} for a review). PQCs consist of a quantum circuit which carries parameters that are updated during a training process (typically a classical optimisation routine). The circuit is kept as shallow as possible so as to be suitable for NISQ devices. 

We ask in this work if it is possible to have a machine learning application for a PQC, which comes with a \textit{provable} superior performance over all classical alternatives on near term devices? Such provable guarantees are even more relevant given recent work in QML algorithm `dequantisations' \cite{tang_quantum-inspired_2018, tang_quantum-inspired_2018-1, andoni_solving_2018, chia_quantum-inspired_2018, gilyen_quantum-inspired_2018}. 

We take the first steps in answering this question in several ways. We define a subclass of $\QCBM$ which we call \emph{Ising Born machines} ($\IBM$). We improve the training of the model over previous methods, which use the maximum mean discrepancy\cite{liu_differentiable_2018} ($\MMD$) with a classical kernel, by introducing quantum kernels into the $\MMD$, as well as by using entirely new cost functions: the Stein discrepancy and the Sinkhorn divergence. To do so we derive their corresponding gradients in the quantum setting. 

We show that these novel methods outperform the $\MMD$-with-classical-kernel by achieving a closer fit to the data as measured by the total variation ($\TV$) distance. We derive \red{forms of the Sinkhorn divergence, which can either be efficient to compute, or result in an upper bound on $\TV$. We observe numerically that the Stein discrepancy provides an upper bound to $\TV$}. Next, we show that sampling from this model can not be simulated efficiently by any classical randomised algorithm, up to multiplicative error in the worst case, subject to common assumptions in complexity theory (namely the non-collapse of the polynomial hierarchy). Furthermore, this holds for many circuit families encountered during training. 

We define a framework in which a provable advantage could be demonstrated, which we refer to as quantum learning supremacy (QLS), and based on distribution learning theory\cite{kearns_learnability_1994}. Based on our classical sampling hardness results, \red{we conjecture that} the $\IBM$ \red{may be} a good candidate for a quantum model which could demonstrate this notion of learning supremacy, \red{however we leave the further investigation of QLS and its potential to be achieved by such models to future work}. Finally, we provide a novel utilisation of such generative models in quantum circuit compilation.

\section*{Results}

The main results of this work are new efficient gradient-based training methods, and results on the hardness of simulating the model we introduce using classical computers. Firstly we define the model used, and discuss its connection to previously studied quantum circuit families. We then discuss the efficient training of the model, firstly recalling a previously known gradient-based training method, which uses the $\MMD$ cost function, and then moving onto our new training methods which use the Stein discrepancy and the Sinkhorn divergence. We then discuss the Sinkhorn divergence complexity in detail, and further argue, using its connection to \red{the} total variation distance, why it should be used. We then prove the hardness results mentioned above, namely that many circuits encountered during gradient based training are hard to classically simulate, before finally discussing the potential use of quantum generative models in learning distributions which are intractable to classical models. In addition we provide a framework to study these advantages.

\subsection*{Ising Born Machine}

Here we define the model we use for distribution learning. A generic quantum circuit Born machine consists of a parameterised quantum circuit, which produces samples by measuring the resulting quantum state, and a classical optimisation loop used to learn a data distribution. The circuits we study have the following structure:

\begin{small}
\begin{align}
    \Qcircuit @C=0.6em @R=0.8em {
    \lstick{\ket{0}}    & \gate{H}  & \multigate{3}{U_z(\boldsymbol\alpha)} & \gate{U^1_f(\Gamma_1, \Delta_1, \Sigma_1)}  & \meter &\cw & \rstick{x_1} \\
    \lstick{\ket{0}}    & \gate{H}  & \ghost{U_z(\boldsymbol\alpha)}        & \gate{U^2_f(\Gamma_2, \Delta_2, \Sigma_2)}  & \meter &\cw & \rstick{x_2} \\
    \cdots              &           &                                       & \cdots                                    & \cdots &    &  \\
    \lstick{\ket{0}}    & \gate{H}  & \ghost{U_z(\boldsymbol\alpha)}        & \gate{U^n_f(\Gamma_n, \Delta_n, \Sigma_n)}  & \meter &\cw & \rstick{x_n} 
    }
\end{align}
\end{small}
where: $x_i \in \{0, 1\}$; the unitaries are defined by \eqref{diagonalunitary} and \eqref{finalmeasurementgate}; $S_j$ indicates the subset of qubits on which each operator, $j$, is applied; and a boldface parameter indicates a set of parameters, $\boldsymbol\alpha = \{\alpha_j\}$.
\begin{equation}
  U_z(\boldsymbol\alpha) \coloneqq  \prod_j U_z \left( \alpha_j, S_j \right) = \prod_j\exp \left(  \mathrm{i} \alpha_j \bigotimes_{k \in S_j} Z_k \right)\label{diagonalunitary}
 \end{equation}
  \begin{equation}
    U_f \left( \mathbf{\Gamma}, \mathbf{\Delta}, \mathbf{\Sigma} \right) \coloneqq \exp\left( \mathrm{i}\sum\limits_{k=1}^n \Gamma_k X_k + \Delta_k Y_k +\Sigma_k Z_k\right)    \label{finalmeasurementgate}
\end{equation}
The operators, $X_k, Y_k, Z_k$ are the standard Pauli operators acting on qubit $k$. Restricting to the case $|S_j| \leq 2$ (since only single and two-qubit gates are required for universal quantum computation) the term in the exponential of \eqref{diagonalunitary} becomes exactly an Ising Hamiltonian:
\begin{equation}
    \mathcal{H} \coloneqq  \mathrm{i}\sum\limits_{i<j} J_{ij}Z_iZ_j + \mathrm{i}\sum\limits_{k=1}^n b_k Z_k \label{iqp_ising_hamiltonian}
\end{equation}

where we are dividing the diagonal unitary parameters,  $\boldsymbol{\alpha} = \{J_{ij}, b_k\}$,  into local terms which act only on qubit $k$, $\{b_k\}$, and coupling terms between two qubits $i$ and $j$, $\{J_{ij}\}$. We call the model a Quantum Circuit \textit{Ising} Born machine ($\IBM$).

A measurement on all qubits in the computational basis results in sample vectors, $\mathbf{x} \in \mathcal{X}^n$, where 
$\mathcal{X} = \{0, 1\}$. These samples are drawn from the distribution, $p_{\boldsymbol\theta}(\mathbf{x})$, parameterised by the set of angles, $\boldsymbol\theta = \{\boldsymbol\alpha, \mathbf{\Gamma}, \mathbf{\Delta}, \mathbf{\Sigma}\}$:
\begin{equation}
    p_{\boldsymbol\theta}(\mathbf{x}) \coloneqq \left|\bra{\mathbf{x}}U_f \left( \mathbf{\Gamma}, \mathbf{\Delta}, \mathbf{\Sigma} \right) U_z(\boldsymbol\alpha)\ket{+}^{\otimes n}\right|^2
\end{equation}
We denote the above model and parameters by $\IBM({\boldsymbol \theta}) \coloneqq \IBM(\boldsymbol\alpha, \mathbf{\Gamma}, \mathbf{\Delta}, \mathbf{\Sigma})$. We choose this structure in order to easily recover two well known circuit classes, namely Instantaneous Quantum Polynomial Time\cite{shepherd_temporally_2009} ($\IQP$) circuits, and the shallowest depth ($p=1$) version of the Quantum Approximate Optimisation Algorithm\cite{farhi_quantum_2014}, ($\QAOA$). 

$\IQP$ circuits are named to reflect the commuting nature of elements in the produce defining the unitary $U_z$,  while $\QAOA$\cite{farhi_quantum_2014} was originally developed as an approximate version of the Quantum Adiabatic Algorithm\cite{farhi_quantum_2000}. Both of these classes of circuits are known to be routes to demonstrate quantum supremacy\cite{bremner_classical_2011, bremner_average-case_2016, farhi_quantum_2016, bremner_achieving_2017}, and we extend this property here by using the results of ref.\cite{fujii_commuting_2017}. These classes can be recovered by setting the parameters of a $\IBM$ as follows:
\begin{multline}
\IQP(\{J_{ij}, b_{k}\}) =
\IBM\Big(\left.\{J_{ij}, b_{k}\},\right.\\
\mathbf{\Gamma} = \left\{\frac{\pi}{2\sqrt{2}}\right\},
\left.\mathbf{0}, \mathbf{\Sigma}=  \left\{\frac{\pi}{2\sqrt{2}}\right\}\right) \end{multline}
\begin{multline}
\QAOA_{p=1}(\{J_{ij}, b_{k}\}, \mathbf{\Gamma}) \\
=
\IBM\left(\{J_{ij}, b_{k}\}, \mathbf{\Gamma} = -\mathbf{\Gamma}, \mathbf{0} , \mathbf{0}\right)
\end{multline}
\red{We denote, for example $\left\{\frac{\pi}{2\sqrt{2}}\right\}$, to be all parameters of the $n$ single qubit gates set to the same value, $\pi/2\sqrt{2}$. We choose the final gate before the computational basis measurement to be in the form of \eqref{finalmeasurementgate}, rather than the more common Euler decomposition of a single qubit gate decomposition found in the literature\cite{liu_differentiable_2018, du_expressive_2018}. This is chosen to make the classical simulation hardness results more apparent in our proofs. 

To recover $\IQP$ circuits, we simply need to generate the final layer of Hadamard gates (up to a global phase) and do so by setting $U_f$ in \eqref{finalmeasurementgate} as follows:
\begin{multline}
    U_f^{\IQP} \left( \left\{\frac{\mathrm{\pi}}{2\sqrt{2}}\right\}, \mathbf{0}, \left\{\frac{\pi}{2\sqrt{2}}\right\} \right)\\
    = \bigotimes_{k=1}^n \mathrm{e}^{\frac{ \mathrm{i}\mathrm{\pi}}{2\sqrt{2}} \left(X_k + Z_k\right)} = \mathrm{i}H^{\otimes n}    
\end{multline}
To recreate depth $1$ $\QAOA$ circuits, we need to set the Pauli $Z$ and $Y$ parameters, $\mathbf{\Delta}, \mathbf{\Sigma} = \mathbf{0}$, since the final gates should be a product of Pauli-$X$ rotations with parameters, $-\mathbf{\Gamma}$.
}

\subsection*{Training the Ising Born Machine}

Here we introduce the alternative training methods which we use for our purposes, and which would be applicable to any generative model. The training procedure is a hybrid of classical and quantum computation, with the only quantum component being the model itself. The remainder of the computation is classical, bringing our scheme into the realm of what is possible for NISQ devices. The procedure can be seen in \figref{fig:ibmtraining}.

\begin{figure*}%
\centering
            \includegraphics[width=\textwidth, height = 0.4\textwidth]{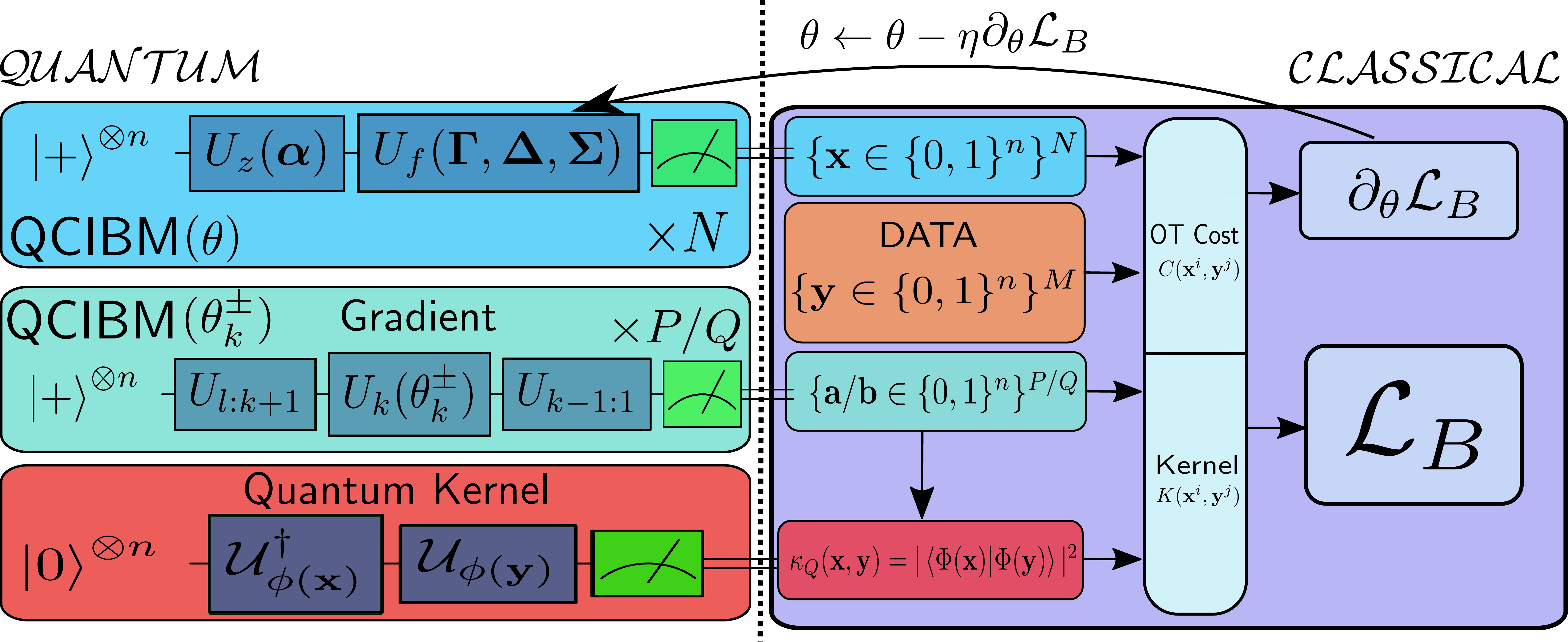}
            \caption{The hybrid training procedures we propose in this work. We have a quantum generator, along with auxiliary circuits used to compute the gradient of the various cost functions with respect to the parameters.
            \red{The training procedure proceeds as follows. First, the $\IBM$ is sampled from $N$ times via measurements. These samples, along with $M$ data samples $\mathbf{y} \sim \pi(\mathbf{y})$, are used to evaluate a cost function, $\mathcal{L}_B$, where $B\in \{\MMD, \SD, \SH\}$ is one of the efficiently computable cost functions. For each updated parameter, $\theta_k$, two parameter shifted circuits are also ran to generate samples, $\mathbf{a}, \mathbf{b} \sim p_{\theta^{\pm}}$ which are used to compute the corresponding gradients, $\partial_{\theta}\mathcal{L}_B$. For all costs functions and gradients, either a kernel (if a quantum kernel is used, the circuit in this figure must be run) is computed for each pair of samples (as is the case for $\MMD$ and $\SD$), or an optimal transport cost function is evaluated (as is the case for $\SH$).
            }}\label{fig:ibmtraining}
\end{figure*}

The optimisation procedures we implement are stochastic gradient descents. Parameters, $\theta_k$, are updated at each epoch of training, $d$, according to the rule $\theta^{d+1}_k \leftarrow \theta_k^{d} - \eta\, \partial_{\theta_k}\mathcal{L}_B$. The parameter $\eta$ is the learning rate, and controls the speed of the descent. The initial proposals to train $\QCBM$s were gradient-free\cite{benedetti_generative_2019, leyton-ortega_robust_2019}, but gradient based methods have also been proposed\cite{liu_differentiable_2018, du_expressive_2018, hamilton_generative_2019}. In this work, we advocate for increasing the classical computational power required in training to achieve better performance, rather than increasing the quantum resources, for example by adding extra ancillae\cite{du_expressive_2018} or adding costly and potentially unstable (quantum) adversaries\cite{zeng_learning_2019, lloyd_quantum_2018, dallaire-demers_quantum_2018}. 

For gradient-based methods, a \textit{cost function} or metric is required, $\mathcal{L}_B\left(p_{\boldsymbol\theta}(\mathbf{x}), \pi(\mathbf{y})\right)$ to compare the Born Machine distribution, $p_{\boldsymbol\theta}(\mathbf{x})$, and the data distribution, $\pi(\mathbf{y})$. Good cost functions will have several characteristics. They should be efficient to compute, measured both by sample and computational complexity. They should also be powerful in the sense that they are sensitive to differences between the two distributions. In this work, we will assess sensitivity by using the total variation ($\TV$) metric as a benchmark:
\begin{equation}
    \TV(p_{\boldsymbol\theta}, \pi) \coloneqq \frac{1}{2}\sum_{\mathbf{x}}|p_{\boldsymbol\theta}(\mathbf{x}) - \pi(\mathbf{x})| \label{total_variation}
\end{equation}
As discussed later, this is a particularly strong metric. The cost functions we use are typically easier to deal with than $\TV$, and we will remark on their relationship to $\TV$. 

One cost function commonly used to train generative models is the Kullback-Leibler ($\KL$) divergence. The $\KL$ divergence is also relatively strong, in the sense that it upper bounds $\TV$ through Pinsker's inequality:
\red{
\begin{equation}
     \TV(p_{\boldsymbol\theta}, \pi) \leq \sqrt{\frac{1}{2} D_{\KL}(p_{\boldsymbol\theta}|| \pi)}\label{eqn:pinskers_inequality}
\end{equation}
where $D_{\KL}(p_{\boldsymbol\theta}|| \pi)$ is the $\KL$ divergence of $\pi$ from $p_{\boldsymbol\theta}$.
}
Unfortunately it is difficult to compute, having a high sample complexity, so neither its gradient, nor the $\KL$ divergence itself, can be evaluated efficiently when training parameterised circuits \cite{liu_differentiable_2018}.

The first efficient gradient method to train Born machines was proposed by ref.\cite{liu_differentiable_2018}. 
There the \textit{maximum mean discrepancy} ($\MMD$) is used to define the cost function. We extend this methodology in two ways. The first is an \red{alteration} to the $\MMD$ itself, and the second is by introducing \red{new} cost functions. From the $\MMD$, the following cost function\cite{borgwardt_integrating_2006, gretton_kernel_2007} can be defined:
\begin{equation}
   \mathcal{L}_{\MMD}  \coloneqq \underset{\substack{\mathbf{x} \sim p_{\boldsymbol\theta}\\ \mathbf{y} \sim p_{\boldsymbol\theta}}}{\mathbb{E}}(\kappa(\mathbf{x},\mathbf{y})) + \underset{\substack{\mathbf{x} \sim \pi \\\mathbf{y} \sim \pi }}{\mathbb{E}}(\kappa(\mathbf{x},\mathbf{y})) -\underset{\substack{\mathbf{x} \sim p_{\boldsymbol\theta}\\ \mathbf{y} \sim \pi}}{2\mathbb{E}}(\kappa(\mathbf{x},\mathbf{y})) \label{mmdexact}
\end{equation}
The $\MMD$ has some very favourable properties; it is a metric on the space of probability distributions, and it is relatively easy to compute \red{(due to low sample complexity)}. The function, $\kappa$ in \eqref{mmdexact} is a \textit{kernel function}, a measure of similarity between points in the sample space $ \mathbf{x}\in\mathcal{X}^n$. A popular choice for this function is the Gaussian mixture kernel\cite{liu_differentiable_2018}:
\begin{equation}
    \kappa_G(\mathbf{x}, \mathbf{y}) \coloneqq  \frac{1}{c}\sum_{i=1}^c\exp\left(-\frac{||\mathbf{x}-\mathbf{y}||^2_2}{2\sigma_i}\right) \label{gaussiankernel}
\end{equation}
The parameters, $\sigma_i$, are \emph{bandwidths} which determine the scale at which the samples are compared, and $||\cdot||_2$ is the $\ell_2$ norm.

Recent works\cite{havlicek_supervised_2019, schuld_supervised_2018} on the near term advantage of using quantum computers in QML have explored \textit{quantum} kernels, which can be evaluated on a quantum computer. To gain such an advantage, these kernels should be difficult to compute on a classical device. In particular, we will adopt the following kernel\cite{havlicek_supervised_2019} in which the samples are encoded in a quantum state, $\ket{\phi(\mathbf{x})}$, via a \textit{feature map}, $\phi:\mathbf{x}\rightarrow \ket{\phi({\mathbf{x}})}$. The kernel is the inner product between vectors:
\begin{equation}
    \kappa_Q(\mathbf{x}, \mathbf{y}) \coloneqq  |\braket{\phi(\mathbf{x})|\phi(\mathbf{y})}|^2 \label{quantumkernel}
\end{equation}
The inner product in \eqref{quantumkernel} is evaluated on a quantum computer, and is conjectured to be hard to compute on a classical one\cite{havlicek_supervised_2019}, given only a classical description of the quantum states.  The state $\ket{\phi(\mathbf{x})}$ is produced by acting an encoding unitary on an initial state, $\ket{\phi(\mathbf{x})} = \mathcal{U}_{\phi(\mathbf{x})}\ket{0}^{\otimes n}$. Explicitly, the kernel is then given by:
\begin{equation}
    \kappa_Q(\mathbf{x}, \mathbf{y}) = 
     |\bra{0^{\otimes n}} \mathcal{U}^{\dagger}_{\phi(\mathbf{x})} \mathcal{U}_{\phi(\mathbf{y})} \ket{0^{\otimes n}}|^2 \label{quantumkernel_expanded} 
\end{equation}
which can be calculated by measuring, in the computational basis, the state which results from running the circuit given by $\mathcal{U}_{\phi(\mathbf{y})}$, followed by that of $\mathcal{U}^{\dagger}_{\phi(\mathbf{x})}$. This is seen in \figref{fig:ibmtraining}. The kernel, \eqref{quantumkernel_expanded} is the observed probability of measuring the all-zero outcome, $0^n$. If this outcome is not observed after polynomially many measurements, the value of the kernel for this particular pair of samples $(\mathbf{x}, \mathbf{y})$ is set to zero. Intuitively, this means the feature map has mapped the original points to points with at most exponentially small overlap in the Hilbert space, and therefore will not contribute to the $\MMD$.

It is also necessary to derive an expression for the gradient of the cost function. For the $\MMD$, the gradient with respect to the  $k^{\text{th}}$ parameter\cite{liu_differentiable_2018}, carried by the $k^{\text{th}}$ unitary gate, $U_k(\theta_k)$, is given by: 
\begin{multline}
    \frac{\partial \mathcal{L}_{\MMD}}{\partial \theta_k} =\underset{\substack{\mathbf{a} \sim p_{\theta_k}^-\\ \mathbf{x} \sim p_{\boldsymbol\theta}}}{2\mathbb{E}}(\kappa(\mathbf{a}, \mathbf{x}))- \underset{\substack{\mathbf{b} \sim p^+_{\theta_k}\\ \mathbf{x} \sim p_{\boldsymbol\theta}}}{2\mathbb{E}}(\kappa(\mathbf{b}, \mathbf{x}))  \\
    - \underset{\substack{\mathbf{a} \sim p^-_{\theta_k}\\ \mathbf{y} \sim \pi}}{2\mathbb{E}}(\kappa(\mathbf{a}, \mathbf{y}))  +\underset{\substack{\mathbf{b} \sim p^+_{\theta_k}\\ \mathbf{y} \sim \pi}}{2\mathbb{E}}(\kappa(\mathbf{b}, \mathbf{y}))\label{mmdgradient}
\end{multline}
where $p^{\pm}_{\theta_k}$ are output distributions generated by running the following auxiliary circuits\cite{mitarai_quantum_2018, schuld_evaluating_2019} \textit{for each} 
unitary gate, $U_k(\theta_k)$:
\begin{equation}
\Qcircuit @C=0.25em @R=1em {
\lstick{\ket{0}^{\otimes n}} & \gate{H^{\otimes n}} & \gate{U_{1:k-1}}&\gate{U_k(\theta_k^{\pm})}&\gate{U_{k+1:l}} & \meter &\cw 
 }
\label{circuit:gradientplusminuscircuit}
\end{equation}
where $\theta_k^{\pm}:=\theta_k \pm\pi/2$ \red{and $U_{l:m}\coloneqq U_lU_{l+1}\dots U_{m-1}U_m$ are the unitary gates in the Born Machine}. \red{This gradient occurs because the form of the unitary gates in our case are exponentiated Pauli operators $U_k(\theta_k) = \exp(i\theta_k \Sigma_k)$, with $\Sigma_k^2= \mathds{I}$. With the unitaries in this form, the gradient of the probabilities outputted from the parameterised state, with respect to a parameter $\theta$, is given by\cite{liu_differentiable_2018, schuld_evaluating_2019}:
\begin{equation}
\frac{\partial p_{\boldsymbol\theta}(\mathbf{z})}{\partial \theta_k} = p_{\theta_k}^{-}(\mathbf{z}) - p_{\theta_k}^{+}(\mathbf{z}) \label{mmdprobabilitygradient}
\end{equation}
There is a slight difference between \eqref{mmdprobabilitygradient} and that of ref\cite{liu_differentiable_2018}, due to a different parameterisation of the unitaries above. 
}

The gradients of the cost functions which we introduce next 
will also \red{require} the parameter-shifted circuits in \eqref{circuit:gradientplusminuscircuit}. For more details on kernel methods and the $\MMD$, see Supplementary Material Section II.

\begin{figure*}
    \centering
    \includegraphics[width=2\columnwidth, height=0.5\columnwidth]
    {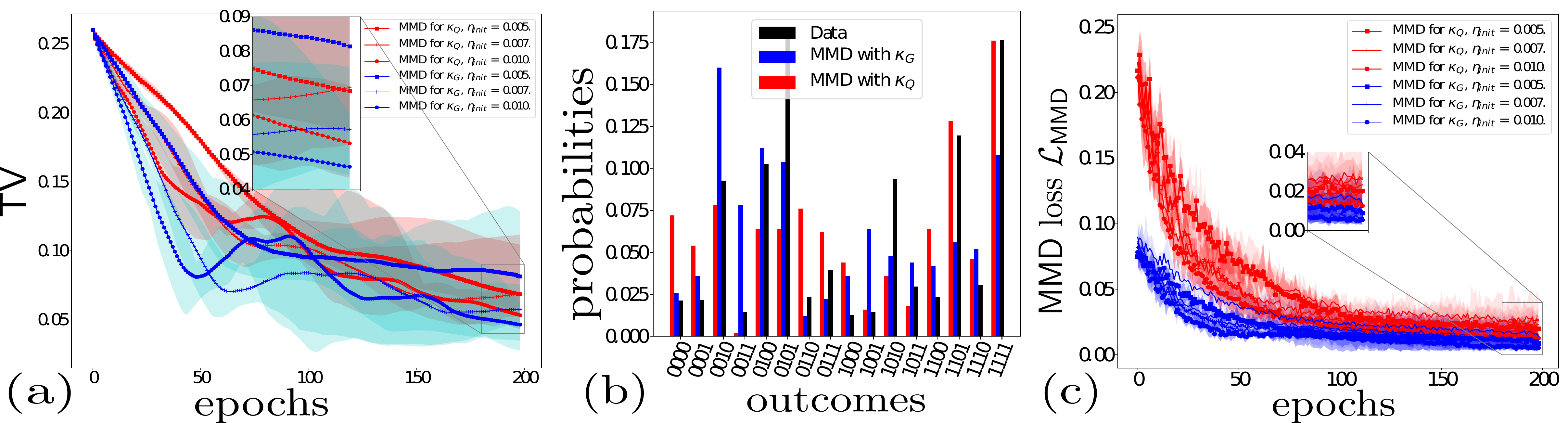}
    \caption{The performance of the quantum kernel $\kappa_Q$ [\crule[red]{0.2cm}{0.2cm}] vs. the Gaussian kernel,  
$\kappa_G$ [\crule[blue]{0.2cm}{0.2cm}] (with $\eta_{\mathsf{init}} = 0.1$) for 4 qubits. During training we sample from the $\IBM$ and the data $500$ times and use a minibatch size of $250$. One epoch is one complete update of all parameters according to gradient descent. Error bars represent maximum, minimum and mean values achieved over 5 independent training runs, with the same initial conditions on the same data samples. (a) $\TV$ difference achieved with both kernel methods during training. No observable or obvious advantage is seen in using the quantum kernel over the Gaussian one, in contrast the Gaussian kernel seems to perform better on average. (b) Final learned probabilities with $\eta_{\mathsf{init}} = 0.01$ using the Adam optimiser. (c) $\MMD$ computed using $400$ samples as training points, $100$ as test points \red{(seen as the thin lines without markers)}, independent of the training data.  }
    \label{fig:QvGkernel4}
\end{figure*}

\subsection*{Stein Discrepancy Training}\label{ssec:steintrainingofibm}
So far, we have only proposed a change of kernel in the $\MMD$ method of training $\IBM$s. We now consider changing the cost function altogether. We endeavour to find costs which are efficient to compute for quantum models, yet stronger than $\MMD$. 

The first cost we propose is called the \textit{Stein discrepancy} ($\SD$). $\SD$ has become popular for goodness-of-fit tests\cite{liu_kernelized_2016}, i.e.\@ testing whether samples come from a particular distribution or not, as opposed to the $\MMD$, which is typically used for kernel two-sample tests\cite{gretton_kernel_2007}. \red{This discrepancy is based on Stein's method\cite{stein_bound_1972}, which is a way to bound distance metrics between probabilities including, for example, the other integral probability metrics (IPM) we utilise in this work. For details on IPMs, see Supplementary Material Section I.} 

We use the discrete version of the Stein discrepancy\cite{yang_goodness--fit_2018} since, in its original form\cite{liu_kernelized_2016}, it only caters for the case where the distributions are supported over a \textit{continuous} space. The discretisation is necessary since the $\IBM$ outputs binary strings and so the standard gradient w.r.t.\@ a sample, $\mathbf{x}$, $\nabla_{\mathbf{x}}$, is undefined. As such, we need to use a discrete `shift' operator, $\Delta_{\mathbf{x}}$, instead, which is an operator defined by $[\Delta_\mathbf{x}f(\mathbf{x})]_i \coloneqq f(\mathbf{x}) - f(\neg_i\mathbf{x})$ for a function $f$,
where $\neg_i$ flips the $i^{th}$ element of the binary vector $\mathbf{x}$. 

Fortunately, the discretisation procedure is relatively straightforward (the necessary definitions and proofs can be found in Supplementary Material Section III).
The discrepancy is derived\cite{liu_kernelized_2016, gorham_measuring_2015} from the (discrete) Stein Identity\cite{yang_goodness--fit_2018}, given by:

\begin{align}
    \underset{\mathbf{x}\sim \pi}{\mathbb{E}}[\mathcal{A}_\pi \phi(\mathbf{x})]&= \underset{\mathbf{x} \sim \pi}{\mathbb{E}}\left[\mathbf{s}_\pi(\mathbf{x})\phi(\mathbf{x}) - \Delta_{\mathbf{x}} \phi(\mathbf{x})\right] = 0 
    \label{complexdiscretesteinidentity}\\
     \mathcal{A}_\pi \phi(\mathbf{x}) &:= \mathbf{s}_\pi(\mathbf{x})\phi(\mathbf{x}) - \Delta_{\mathbf{x}} \phi(\mathbf{x})
\end{align}
where $\underset{\mathbf{x}\sim \pi}{\mathbb{E}}$ denotes the expectation value over the distribution, $\pi$. This holds for any function $\phi: \mathcal{X}^n\rightarrow \mathbb{C}$, and probability mass function $\pi$ on $\mathcal{X}^n$. The function, $s_\pi(\mathbf{x}) = \Delta_\mathbf{x}\log(\pi(\mathbf{x}))$ is the \textit{Stein score} function of the distribution $\pi$, and $\mathcal{A}_\pi$ is a so-called \textit{Stein operator} of $\pi$. Now, the $\SD$ cost function can be written in a kernelised form\cite{liu_kernelized_2016, yang_goodness--fit_2018}, similarly to the $\MMD$:
\begin{equation}
    \mathcal{L}_{\SD}(p_{\boldsymbol\theta}, \pi)\coloneqq \mathbb{E}_{\mathbf{x}, \mathbf{y}\sim p_{\boldsymbol\theta}}\left[\kappa_\pi(\mathbf{x}, \mathbf{y})\right]\label{steindiscrepancybornmachine}
\end{equation}
\begin{multline}
    \kappa_\pi(\mathbf{x}, \mathbf{y}) \coloneqq s_\pi(\mathbf{x})^T\kappa(\mathbf{x}, \mathbf{y})s_\pi(\mathbf{y}) -s_\pi(\mathbf{x})^T\Delta_{\mathbf{y}}^*\kappa(\mathbf{x}, \mathbf{y}) \\
    - \Delta_{\mathbf{x}}^*\kappa(\mathbf{x}, \mathbf{y})^Ts_\pi(\mathbf{y}) + \tr(\Delta_{\mathbf{x}, \mathbf{y}}^*\kappa(\mathbf{x}, \mathbf{y})) \label{weightedkernelbornmachine}
\end{multline}
where $\kappa_\pi$ is the \textit{Stein kernel}, and $\kappa$ is a usual positive semi-definite kernel. $\Delta_{\mathbf{x}}^*$ is a conjugate version of the operator $\Delta_{\mathbf{x}}$, but for our purposes, the behaviour of both $\Delta_{\mathbf{x}}^*$ and $\Delta_{\mathbf{x}}$ are identical. For completeness, we define it in generality in Supplementary Material Section III.

Just as above, the gradient (derived \red{in an identical fashion to the $\MMD$ gradient \eqref{mmdgradient}} as is detailed in Supplementary Material Section III) of $\mathcal{L}_{\SD}$ with respect to the parameter, $\theta_k$, is given by:
\begin{multline}
    \frac{\partial \mathcal{L}_{\SD}}{\partial \theta_k} = \underset{\substack{\mathbf{x} \sim p^-_{\boldsymbol\theta} \\ \mathbf{y}\sim p_{\boldsymbol\theta}}}{\mathbb{E}}[\kappa_\pi(\mathbf{x}, \mathbf{y})] - \underset{\substack{\mathbf{x} \sim p^+_{\boldsymbol\theta} \\ \mathbf{y}\sim p_{\boldsymbol\theta}}}{\mathbb{E}}[\kappa_\pi(\mathbf{x}, \mathbf{y})] \\
    +\underset{\substack{\mathbf{x} \sim p_{\boldsymbol\theta}\\ \mathbf{y}\sim p^-_{\boldsymbol\theta}}}{\mathbb{E}}[\kappa_\pi(\mathbf{x}, \mathbf{y})] - \underset{\substack{\mathbf{x} \sim p_{\boldsymbol\theta} \\ \mathbf{y}\sim p^+_{\boldsymbol\theta}}}{\mathbb{E}}[\kappa_\pi(\mathbf{x}, \mathbf{y})] \label{steingradient}
\end{multline}
We show that almost every term in \eqref{steindiscrepancybornmachine} and \eqref{steingradient} can be computed efficiently, even when the quantum kernel, $\kappa_Q$ from \eqref{quantumkernel}, is used in \eqref{weightedkernelbornmachine}. That is, with the exception of the score function $s_\pi$ with respect to the data distribution. The score contains an explicit dependence on the data distribution, $\pi$. If we are given oracle access to the probabilities, $\pi(\mathbf{y})$, then there is no issue and $\SD$ will be computable. Unfortunately, in any practical application this will not be the case. 

To deal with such a scenario, we give two approaches to approximate the score via samples from $\pi$. The first of these we call the `Identity' method since it inverts Stein's identity\cite{li_gradient_2018} \red{from \eqref{complexdiscretesteinidentity}}. We refer to the second as the `Spectral' method since it uses a spectral decomposition\cite{shi_spectral_2018} of a kernel to approximate the score. \red{The latter approach uses the Nystr\"{o}m method\cite{nystrom_uber_1930}, which is a technique used to approximately solve integral equations}. We will only use the Spectral method in training the $\IBM$ in the numerical results in \figref{fig:MMDvSinkvStein3}, since the Identity method does not give an immediate out-of-sample method to compute the score. Details of these methods can be found in Supplementary Material Section III. 

Notice that even with the difficulty in computing the score, the $\SD$ is still more suitable for training these models than the $\KL$ divergence as the latter requires computing the \textit{circuit} probabilities, $p_{\boldsymbol\theta}(\mathbf{x})$, which is in general intractable, and so could not be computed for \textit{any} dataset. 

\subsection*{Sinkhorn Divergence Training \label{ssec:sinkhorntrainingofibm}}

The second cost function we consider is the so-called \textit{Sinkhorn divergence} ($\SH$). This is a relatively new method to compare probability distributions \cite{ramdas_wasserstein_2015, genevay_learning_2018, feydy_interpolating_2019} , defined by the following:
\begin{multline}
    \mathcal{L}_{\SH}^\epsilon(p_{\boldsymbol\theta}, \pi) \coloneqq
    \OT^c_\epsilon(p_{\boldsymbol\theta}, \pi) \\
    - \frac{1}{2} \OT^c_\epsilon(p_{\boldsymbol\theta}, p_{\boldsymbol\theta}) -\frac{1}{2}\OT^c_\epsilon(\pi, \pi) \label{sinkhorndivergence}
\end{multline}
\begin{multline}
\small
\OT^c_\epsilon(p_{\boldsymbol\theta}, \pi) \coloneqq  \\\small
      \min\limits_{U \in \mathcal{U}(p_{\boldsymbol\theta}, \pi)}\left(\sum\limits_{\substack{(\mathbf{x}, \mathbf{y})\\ \in \mathcal{X}^d\times\mathcal{Y}^d}} c(\mathbf{x}, \mathbf{y})U(\mathbf{x}, \mathbf{y})
    + \epsilon \KL(U|p_{\boldsymbol\theta}\otimes \pi)\right) 
    \label{wassersteinregularised}
\end{multline}
\noindent where $\epsilon \geq 0 $ is a regularisation parameter, $c(x, y)$ is a Lipschitz `cost' function, and $\mathcal{U}(p_{\boldsymbol\theta}, \pi)$ is the set of all \textit{couplings} between $p_{\boldsymbol\theta}$ and $\pi$, i.e.\@ the set of all joint distributions, whose marginals with respect to $\mathbf{x}, \mathbf{y}$ are $p_{\boldsymbol\theta}(\mathbf{x}), \pi(\mathbf{y})$ respectively. 
The above cost function, $\mathcal{L}_{\SH}^\epsilon$, is particularly favourable as a candidate because of its relationship to the theory of \textit{optimal transport}\cite{villani_optimal_2009} ($\OT$),  a method to compare probability distributions. It has become a major tool used to train models in the classical domain, for example with GANs\cite{arjovsky_wasserstein_2017} through a restriction of optimal transport called the Wasserstein metric, which is derived from
$\OT$, when the cost ($c(\mathbf{x}, \mathbf{y})$), is chosen to be a metric on the space of $\mathcal{X}^n$.

We would like to use $\OT$ \textit{itself} to train generative models, due to its metric properties. Unfortunately, $\OT$ has high computational cost and exponential sample complexity\cite{dudley_speed_1969}. For this reason, the Sinkhorn divergence was proposed in refs. \cite{ramdas_wasserstein_2015, genevay_learning_2018, feydy_interpolating_2019} to \textit{interpolate} between $\OT$ and the $\MMD$ as a function of the regularisation parameter, $\epsilon$ in \eqref{wassersteinregularised}. In particular, for the two extreme values of $\epsilon$, we recover\cite{ramdas_wasserstein_2015} both unregularised $\OT$, and the $\MMD$:
\begin{align}
\underline{\epsilon \rightarrow 0}:& \nonumber\\
&\mathcal{L}_{\SH}^0(p_{\boldsymbol\theta}, \pi)\rightarrow \OT_0^c(p_{\boldsymbol\theta}, \pi)\\
\underline{\epsilon \rightarrow \infty}:&\nonumber\\
&\mathcal{L}_{\SH}^\epsilon(p_{\boldsymbol\theta}, \pi) \rightarrow  \MMD(p_{\boldsymbol\theta}, \pi)\\
& \qquad \kappa(\mathbf{x}, \mathbf{y}) = -c(\mathbf{x}, \mathbf{y})\nonumber
\end{align}

As before, we need a gradient of the $\mathcal{L}_{\SH}$ with respect to the parameters, which is given by:
\begin{equation}
    \frac{\partial \mathcal{L}_{\SH}^\epsilon(p_{\boldsymbol\theta}, \pi)}{\partial \theta_k} 
     = \underset{\substack{\mathbf{x} \sim p_{\theta_k^-} }}{\mathbb{E}}[\varphi(\mathbf{x})] 
    -\underset{\substack{\mathbf{x} \sim p_{\theta_k^+}  }}{\mathbb{E}}[\varphi(\mathbf{x})] \label{sinkhorngradient}
\end{equation}

where $\varphi(\mathbf{x})$ is a function which depends on the optimal solutions found to the regularised $\OT$ problem in \eqref{wassersteinregularised}. See Supplementary Material Section IV for more details on the Sinkhorn divergence, and its gradient.

\subsection*{Sinkhorn Complexity}

The sample complexity of the Sinkhorn divergence is of great interest to us as we claim that the $\TV$ and the $\KL$ are not suitable to be directly used as cost functions. This is due to the difficulty of computing the outcome probabilities of quantum circuits efficiently. We now motivate why the $\MMD$ is a weak cost function, and why the Sinkhorn divergence should be used as an alternative. This will depend critically on the regularisation parameter $\epsilon$, which allows a smooth interpolation between the $\OT$ metric and the $\MMD$. 

Firstly, we address the computability of $\mathcal{L}_{\SH}$ and we find, due to the results of ref. \cite{genevay_sample_2018}, a somewhat `optimal' value for $\epsilon$, for which the sample complexity of $\mathcal{L}_{\SH}$ becomes efficient. Specifically, the mean error between $\mathcal{L}_{\SH}$ and its approximation $\hat{\mathcal{L}}_{\SH}$ for $n$ qubits, computed using $M$ samples, scales as:
\begin{multline}
    \mathbb{E}|\mathcal{L}_{\SH}^\epsilon - \hat{\mathcal{L}}_{\SH}^\epsilon| = \mathcal{O}\left[\frac{1}{\sqrt{M}}\left(1+\mathrm{e}^{\left(2\frac{n^2+n}{\epsilon}\right)}\right)\right.\\
    \left.\times\left(1+\frac{1}{\epsilon^{\lfloor n/2\rfloor}}\right)\right]
\end{multline}

We show in Supplementary Material Section IV.1 that by choosing $\epsilon = \mathcal{O}(n^{2})$, we get:
\begin{equation}
     \mathbb{E}|\mathcal{L}_{\SH}^{\mathcal{O}(n^{2})} - \hat{\mathcal{L}}_{\SH}^{ \mathcal{O}(n^{2})}| =~ \mathcal{O}\left(\frac{1}{\sqrt{M}}\right) \label{sinkhorn_expectation_sample_chosen_MAIN}
\end{equation}
\noindent which is the same sample complexity as the $\MMD$\cite{sriperumbudur_integral_2009}, but exponentially better than that of unregularised optimal transport, which scales as $\mathcal{O}\left(1/{M}^{1/n}\right)$\cite{dudley_speed_1969}. 

A similar result can be derived using a concentration bound\cite{genevay_sample_2018}, such that with probability $1-\delta$, 
\begin{equation}
    |\mathcal{L}_{\SH}^{\mathcal{O}(n^{2})}  - \hat{\mathcal{L}}_{\SH}^{\mathcal{O}(n^{2})} |
    = \mathcal{O}\left(\frac{n}{\sqrt{M}}\log(1/\delta)^{1/2}\right) \label{sinkhornborn_samplecomplexity_choosed_MAIN}
\end{equation}

where we have chosen the same scaling for $\epsilon$ as in \eqref{sinkhorn_expectation_sample_chosen_MAIN}. Therefore, we can choose an optimal theoretical value for the regularisation, such that $\mathcal{L}_{\SH}$ is sufficiently far from $\OT$ to be efficiently computable, but perhaps still retains some of its favourable properties. It is likely in practice however, that a much lower value of $\epsilon$ could be chosen without a blow up in sample complexity\cite{genevay_learning_2018,  genevay_sample_2018}. See Supplementary Material Section IV for derivations of the above results.

Secondly, we can relate the $\mathcal{L}_{\SH}$ to unregularised $\OT$ and $\TV$ via a sequence of inequalities. We have mentioned that the $\MMD$ is weak, meaning it provides a \textit{lower} bound on $\TV$ in the following way\cite{sriperumbudur_integral_2009}:
\begin{equation}
    \TV(p_{\boldsymbol\theta}, \pi) \geq \frac{\sqrt{\MMD(p_{\boldsymbol\theta}, \pi)}}{\sqrt{C}}
\end{equation}

\noindent if $C := \sup_{\mathbf{x} \in \mathcal{X}^n} \kappa(\mathbf{x}, \mathbf{x}) < \infty$. 

Note that for the two kernels introduced earlier:

\begin{equation}
\kappa_G(\mathbf{x}, \mathbf{x}) = \frac{1}{c}\sum\limits_{i=1}^c e^{-\frac{1}{2\sigma_i}|\mathbf{x} - \mathbf{x}|^2} =\frac{1}{c}(c) = 1
\end{equation}
\begin{equation}
\kappa_Q(\mathbf{x}, \mathbf{x}) = |\braket{\phi(\mathbf{x})|\phi(\mathbf{x})}|^2 = |\braket{0|0}^{\otimes n}|^2 = 1 
\end{equation}
hence $C = 1$ and the lower bound is immediate. 

In contrast, as is seen from the inequality on a discrete sample space in \eqref{tv_wasserstein_inequality}\cite{gibbs_choosing_2002}, the Wasserstein metric (unregularised $\OT$) provides an \textit{upper} bound on $\TV$, and hence we would expect it to be stronger then the $\MMD$.
\begin{equation}
    d_{min} \TV(p_{\boldsymbol\theta}, \pi) \leq \OT^d_0(p_{\boldsymbol\theta}, \pi) \leq \text{diam}(\mathcal{X})\TV(p_{\boldsymbol\theta}, \pi) \label{tv_wasserstein_inequality}
\end{equation}
where $\text{diam}(\mathcal{X}^n) = \max\{d(\mathbf{x}, \mathbf{y}), \mathbf{x}, \mathbf{y} \in \mathcal{X}^n\}$, $ d_{min} = \min_{x\neq y}d(\mathbf{x}, \mathbf{y})$, and $d(\mathbf{x}, \mathbf{y})$ is the metric on the space, $\mathcal{X}^n$. This arises by choosing $c = d$ and $\epsilon = 0$ in \eqref{wassersteinregularised}. If, for instance, we were to choose $d(\mathbf{x}, \mathbf{y})$ to be the $\ell_1$ metric between the binary vectors of length $n$ (a.k.a.\@ the  Hamming distance), then we get that $ d_{min} = 1, \text{diam}(\mathcal{X}) = n$, and so:
\begin{equation}
     \TV(p_{\boldsymbol\theta}, \pi) \leq \OT^{\ell_1}_0(p_{\boldsymbol\theta}, \pi) \leq n \TV(p_{\boldsymbol\theta}, \pi)
\end{equation}
Finally, we can examine the relationship induced by the regularisation parameter through the following inequality; Theorem 1 in ref.\cite{genevay_sample_2018}:
\begin{multline}
    0 \leq \OT^c_\epsilon(p_{\boldsymbol\theta}, \pi ) -\OT^c_0(p_{\boldsymbol\theta}, \pi ) \leq 2\epsilon\log\left(\frac{\mathrm{e}^2LD}{n\epsilon}\right)\\
    \sim_{\epsilon \rightarrow 0} 2\epsilon\log\left(1/\epsilon\right)
\end{multline}
where the size of the sample space is bounded by $D$, as measured by the metric, and $L$ is the Lipschitz constant of the cost $c$. As detailed in Supplementary Material Section IV.1 we can choose $D = n$, $L = n$:
\begin{equation}
    0 \leq \OT^{\ell_1}_\epsilon(p_{\boldsymbol\theta}, \pi) - \OT^{\ell_1}_0(p_{\boldsymbol\theta}, \pi) \leq 2\epsilon\log\left(\frac{\mathrm{e}^2n}{\epsilon}\right)
\end{equation}
The $\log$ term will be positive as long as $\epsilon \leq n\mathrm{e}^2$, in which case regularised $\OT$ will give an upper bound for the Wasserstein metric, and hence the $\TV$ through \eqref{tv_wasserstein_inequality} so we arrive at:
\begin{equation}
     \TV(p_{\boldsymbol\theta}, \pi) \leq  \OT^{\ell_1}_0(p_{\boldsymbol\theta},\pi) \leq  \OT^{\ell_1}_{\epsilon \leq n \mathrm{e}^2}\label{tv_wasserstein_sinkhorn_inequality}
\end{equation}
Unfortunately, comparing this with \eqref{sinkhorn_expectation_sample_chosen_MAIN} and \eqref{sinkhornborn_samplecomplexity_choosed_MAIN}, we can see that with this scaling of $\epsilon$, the sample complexity would pick up an exponential dependence on the dimension, $n$, so it would not be efficiently computable. \red{We comment further on this point later.}

\subsection*{Numerical Performance}

In Figs. \ref{fig:QvGkernel4}, \ref{fig:MMDvSinkvStein3}, \ref{fig:MMDvSink4_real}, we illustrate the superior performance of our alternative training methods, as measured by the total variation distance. A lower $\TV$ indicates the model is able to learn parameters which fit the true data more closely. $\TV$ was chosen as an objective benchmark for several reasons. Firstly, it is typically the notion of distance which is required by quantum supremacy experiments where one wants to prove hardness of classical simulation. Secondly, we use it in the definitions of quantum learning supremacy. Finally, it is one of the strongest notions of convergence in probability one can ask for, so it follows that a training procedure which can more effectively minimise $\TV$, in an efficient way, should be better for generative modelling.

We train the model on Rigetti's Forest platform\cite{smith_practical_2016} using both a simulator, and real quantum hardware, the {\fontfamily{cmtt}\selectfont Aspen} QPU. \red{\figref{fig:QvGkernel4} illustrates the training of the model using the Gaussian (\eqref{gaussiankernel}) versus the quantum kernel (\eqref{quantumkernel}) for 4 qubits, and we see that the quantum kernel offers no significant advantage versus training with a purely classical one. \figref{fig:QvGkernel4}(a) shows the $\TV$ as trained for $200$ epochs, using both the classical and quantum kernels with various learning rates. \figref{fig:QvGkernel4}(b) shows the learned probabilities outputted after training with each kernel, and \figref{fig:QvGkernel4}(c) shows the difference in the actual $\mathcal{L}_{\MMD}$ itself while training with both methods. Interestingly, the latter behaviour is quite different for both kernels, with the quantum kernel initialising with much higher values of $\mathcal{L}_{\MMD}$, whereas they both minimise $\TV$ in qualitatively the same way. This indicates that hardness of classical simulation (of computing the kernel), does not imply an advantage in learning. 

On the other hand, a noticeable out-performance is observed for the Sinkhorn divergence and the Stein discrepancy relative to training with the $\MMD$ (using a Gaussian kernel), as measured by $\TV$ in \figref{fig:MMDvSinkvStein3}. Furthermore we observed that the gap (highlighted in the inset in \figref{fig:MMDvSinkvStein3}(a)) which separates the Sinkhorn divergence and Stein discrepancy (red and blue lines) from the $\MMD$ (green, yellow and cyan lines) grows as the number of qubits grows. Unfortunately, the Spectral method to approximate the Stein score does not outperform the $\MMD$, despite training successfully. The discrepancy between the true and approximate versions of the Stein score is likely due to the low number of samples used to approximate the score, with the number of samples limited by the computational inefficiency. We leave tuning the hyperparameters of the model in order to get better performance to future work. 

This behaviour is shown to persist on the QPU, \figref{fig:MMDvSink4_real}, where we show training of the model with both the $\MMD$ and $\SH$ relative to $\TV$, (\figref{fig:MMDvSink4_real}(a)), the learned probabilities of both methods on, and off, the QPU (\figref{fig:MMDvSink4_real}(b)) and the behaviour of the cost functions associated to both methods (Figs. \ref{fig:MMDvSink4_real}(c), \ref{fig:MMDvSink4_real}(d)). This reinforces our theoretical argument that the Sinkhorn divergence is able to better minimise $\TV$ to achieve superior results.
}

\red{Given the performance noted above, we would recommend the Sinkhorn divergence as the primary candidate for future training of these models, due to its simplicity and competitive performance. One should also note that we do not attempt to learn these data distributions exactly since we use a shallow fixed circuit structure for training (i.e.\@ a $\QAOA$ circuit) which we do not alter. Better fits to the data could likely be achieved with deeper circuits with more parameters.}

For extra numerical result demonstrating the performance of the learning algorithms, see Supplementary Material Section V, \red{including: a comparison between the quantum and Gaussian kernels for two qubits, similar to \figref{fig:QvGkernel4}, the behaviour of the corresponding cost functions themselves associated to \figref{fig:MMDvSinkvStein3}, the performance of the model for $4$ qubits, similar to \figref{fig:MMDvSinkvStein3},  
and the results using a 3 qubit device, the {\fontfamily{cmtt}\selectfont Aspen-4-3Q-A}. In all cases, the performance was qualitatively similar to that reported in the main text.}

\begin{figure*}
    \centering
    \includegraphics[width=2\columnwidth, height=0.8\columnwidth]{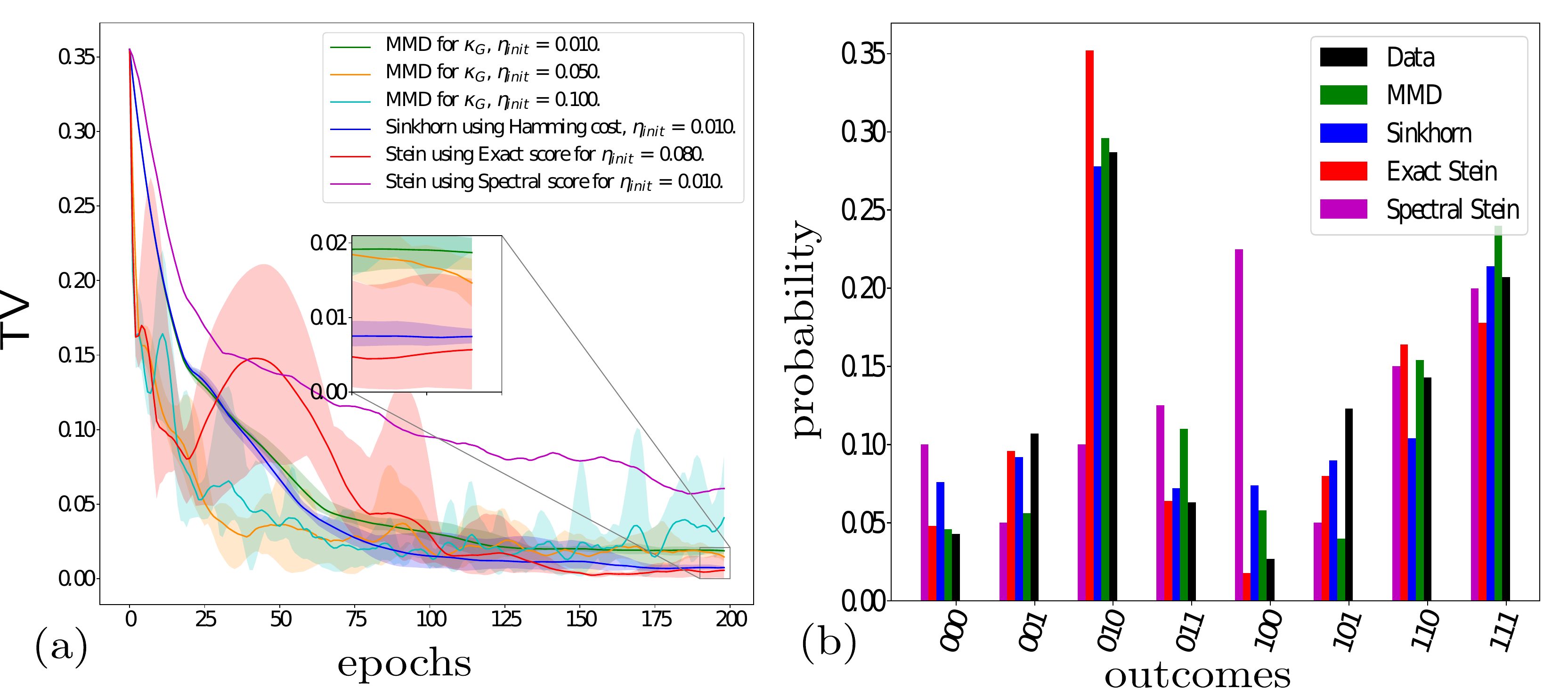}
\caption{$\MMD$ [\crule[cyan]{0.2cm}{0.2cm}, \crule[yellow]{0.2cm}{0.2cm}, \crule[ForestGreen]{0.2cm}{0.2cm}] vs.\@ Sinkhorn [\crule[blue]{0.2cm}{0.2cm}] and Stein training with Exact Score function [\crule[red]{0.2cm}{0.2cm}] and Spectral Score method [\crule[magenta]{0.2cm}{0.2cm}] for 3 qubits with fully connected topology, Rigetti {\fontfamily{cmtt}\selectfont 3q-qvm},
\protect\threeqqvm \ \ \ , trained on the data, \eqref{toydatadistribution}. 500 data points are used for training, with 400 used as a training set, and 100 used as a test set. Plots show mean, maximum and minimum values achieved over 5 independent training runs on the same dataset. (a) $\TV$ difference between training methods, with regularisation parameter $\epsilon = 0.1$ for $\SH$, and 3 eigenvectors for Spectral Stein method. Both Sinkhorn divergence and Stein discrepancy are able to achieve a lower $\TV$ than the $\MMD$. Inset shows region of outperformance on the order of $\sim 0.01$ in $\TV$. We observe that the Spectral score method was not able to minimise $\TV$ as well as the exact Stein discrepancy, potentially indicating the need for better approximation methods. (b) Final learned probabilities of each training method.  See Supplementary Material Section V for behaviour of corresponding cost functions. }
    \label{fig:MMDvSinkvStein3}
\end{figure*}

\begin{figure*}
    \centering
    \includegraphics[width=2\columnwidth, height=0.45\columnwidth]{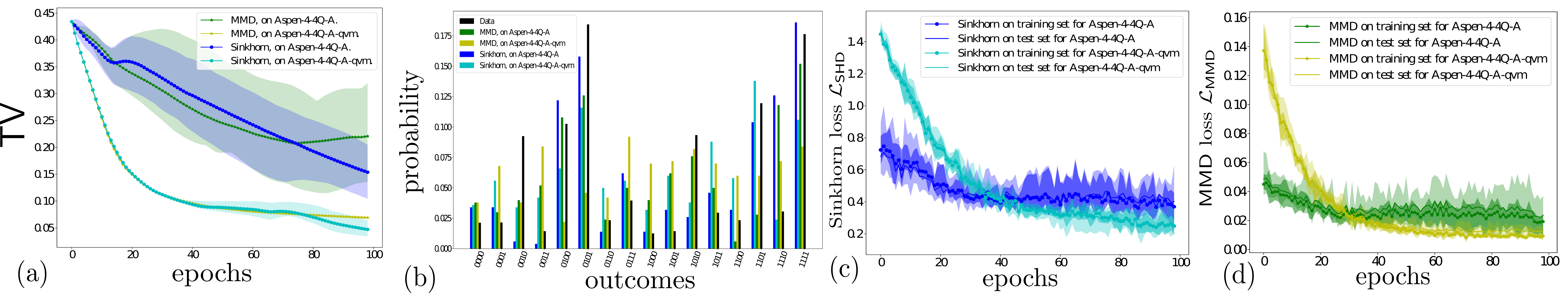}
\caption{$\MMD$ [\crule[ForestGreen]{0.2cm}{0.2cm}, \crule[yellow]{0.2cm}{0.2cm}] vs.\@ Sinkhorn [\crule[blue]{0.2cm}{0.2cm}, \crule[cyan]{0.2cm}{0.2cm}] for 4 qubits comparing performance on the real QPU ({\fontfamily{cmtt}\selectfont Aspen-4-4Q-A}) vs.\@ simulated behaviour on QVM ({\fontfamily{cmtt}\selectfont Aspen-4-4Q-A-qvm}) using 500 samples and a batch size of 250, learning target data [\crule[black]{0.2cm}{0.2cm}] and a initial learning rate for Adam as $\eta_{init} = 0.01$. (a)$\TV$ Difference between training methods with regularisation parameter $\epsilon = 0.08$, (b) Final learned probabilities, \red{[\crule[black]{0.2cm}{0.2cm}] indicates the probabilities of a random instance of the data distribution (see Methods) chosen. The probabilities given by the other bars are those achieved after training the model with either the $\MMD$ or $\SH$ on the simulator or the physical Rigetti chip, on an average run. The probabilities of the model are generated by simulating the entire wavefunction after training.}, (c) $\mathcal{L}^{0.08}_{\SH}$ on QVM  [\crule[cyan]{0.2cm}{0.2cm}] vs.\@ QPU  [\crule[blue]{0.2cm}{0.2cm}]. (d) $\mathcal{L}_{\MMD}$ on QVM  [\crule[yellow]{0.2cm}{0.2cm}] vs.\@ QPU  [\crule[ForestGreen]{0.2cm}{0.2cm}]. \red{In both latter cases, trained model performance on 100 test samples is seen as the thin lines without markers. Again it can be seen that the Sinkhorn divergence outperforms the $\MMD$ both simulated and on chip, with the deviation apparent towards the end of training. Similar behaviour observed after 100 epochs, but not shown due to limited QPU time.}  }
    \label{fig:MMDvSink4_real}
\end{figure*}

\subsection*{Hardness and Quantum Advantage}
It is crucially important, not just for our purposes, but for the design of quantum machine learning algorithms in general, that the algorithm itself is providing some advantage over \textit{any} classical one, for the same task. This is the case for so-called \textit{coherent} algorithms, like the HHL linear equation solver\cite{harrow_quantum_2009}, which is $\BQP$-complete, and therefore unlikely to be fully dequantised, 
However, such a proven advantage for \textit{near term} QML algorithms is yet out of reach. We attempt to address such a question in two steps. 
\begin{enumerate}
    
    \item We show that for a large number of parameter values, $\boldsymbol\theta$, our $\IBM$ circuits are `hard'. That is to say cannot be efficiently simulated classically up to a multiplicative error, in the worst case. We also show that this holds for the auxiliary quantum circuits used for the gradient estimation, and hence the model may \textit{remain} hard during training (although we do not know for sure).
    
    \item We provide formal definitions for quantum learning supremacy, the ability of a quantum model to provably outperform all classical models in a certain task, and a potential pathway to prove such a thing.

\end{enumerate}

The \emph{intuition} behind point 2 is the following. If our $\IBM$ model could learn a target distribution $\pi$, which demonstrates quantum supremacy, by providing a quantum circuit $C$ close enough to $\pi$ (i.e.\@ below a threshold error in total variation), then the model would have demonstrated something which is classically infeasible. Else there would exist an efficient classical algorithm which can get close to $\pi$ which contradicts hardness.

Point 1 does not completely fit that intuition. For one thing, hardness is not known to hold for the required notion of additive error (i.e.\@ total variation distance) but only for multiplicative error. Also, even though the model is more \textit{expressive} than any classical model\cite{du_expressive_2018}, this does not imply that it could actually \textit{learn} a hard distribution. On the other hand, it is easy to see why the converse would be true, if the $\IBM$ could learn a distribution which is hard to sample from classically, the underlying circuit must have, at some point, reached a circuit configuration for which the output distribution is hard to classically sample.

We can address point 1 informally (see Supplementary Material Section VI for the formal statements and proof) in three steps:
\begin{itemize}
    
    \item If the parameters of the model are initialised randomly in $\{\boldsymbol \alpha\} = \{J_{ij}, b_k\}\in\{0,\frac{\pi}{8}, \dots, \frac{7\pi}{8}\}$ and final measurement angles are chosen such that $U_f \left( \mathbf{\Gamma}, \mathbf{\Delta}, \mathbf{\Sigma} \right) = H^{\otimes n}$, then the resulting $\IBM$ circuit class will be hard to simulate up to an additive error of $1/384$ in total variation distance, subject to a conjecture relating to the hardness of computing the Ising partition function \cite{bremner_average-case_2016}.
    
    \item If certain configurations of the parameters are chosen to be either of the form, $(2l+1)\pi/kd$, where $l, d$ are integers, and $k$ is a number which depends on the circuit family, or in the form $2\pi\nu$, where $\nu$ is irrational, then the resulting class of circuits will be hard to sample from classically, up to a multiplicative error, in the worst case.
    
    \item The circuits produced at each epoch as a result of the gradient updates will each result in a hard circuit class as long as the gradient updates are not chosen carelessly. In each epoch, if the update step is constrained in a way that the new value of the parameter $\theta^{d+1}_k= \theta_k^{d} - \eta\partial_{\theta_k}\mathcal{L}_B$ does not become rational, then the updated circuits will also belong to a class which is hard to simulate (a similar result can be shown for the case where the parameters are updated to keep within the form of $(2l+1)\pi/kd$).  This is because the updates can simply be absorbed into the original gates, to give a circuit which has the same form. This holds also for the gradient shifted circuits in \eqref{circuit:gradientplusminuscircuit} since these correspond to circuits whose parameters are updated as follows: $\theta^{d, \pm}_k \leftarrow \theta_k^{d} \pm \pi/2$.
    
\end{itemize}

We now provide definitions to meet the requirements of point 2, adapting definitions from distribution learning theory\cite{kearns_learnability_1994} for this purpose. Specifically, we say that a generative QML algorithm, $\mathcal{A}\in \BQP$ (with a small abuse of notation) has demonstrated \textit{quantum learning supremacy} (QLS) if there exists a class of probability distributions $\mathcal{D}_n$ over $\mathcal X^n$ (bit vectors of length $n$), for which there exists a metric $d$, and a fixed $\epsilon$ such that $\mathcal{D}_n$ is $(d ,\epsilon, \BQP)$-learnable via $\mathcal{A}$, but not $\left( d ,\epsilon, \BPP \right)$-learnable (i.e.\@ learnable by a purely classical algorithm). The task of the learning algorithm $\mathcal{A}$ is, given a target distribution $D\in \mathcal D_n$, to output, with high probability, a \textit{Generator}, $\GEN_{D'}$, for a distribution $D'$, such that $D'$ is close to $D \in \mathcal{D}_n$ with respect to the metric $d$. For the precise definitions of learnability we employ see Supplementary Material Section VII. 

This framework is very similar to that of, and inspired by, probably approximately correct (PAC) learning, which has been well studied in the quantum case\cite{arunachalam_survey_2017}, but it applies more closely to the task of generative modelling. It is known that in certain cases, the use of quantum computers can be beneficial to PAC learning, but not generically\cite{arunachalam_quantum_2019}. Based on this, it is possible that there exist some classes of distributions which cannot be efficiently learned by classical computers ($\BPP$ algorithms), but which could be learned by quantum devices ($\BQP$ algorithms). The motivation for this is exactly rooted in the question of quantum supremacy, and illustrated crudely in \figref{fig:learning_supremacy}(b). 

An initial attempt at QLS is as follows. As mentioned above, if random $\IQP$ circuits could be classically simulated to within a $\TV$ error of $\epsilon = 1/384$\cite{bremner_average-case_2016} in the worst case (with high probability over the choice of circuit), this would imply unlikely consequences for complexity theory. Now, if a generative quantum model was able to achieve a closeness in $\TV$ less than this constant value, perhaps by minimising one of the upper bounds in \eqref{tv_wasserstein_sinkhorn_inequality}, then we could claim this model had achieved something classically intractable. 
For example if we make the following assumptions,
\begin{enumerate}
    \item $\IBM$ could achieve a $\TV < \delta$ to a target $\IQP$ distribution. 
    \item A classical probabilistic algorithm, $\mathsf{C}$, could output a distribution $q$ in polynomial time which was $\gamma$ close in $\TV$ to the $\IBM$, i.e.\@ it could simulate it efficiently.
\end{enumerate}
Then:

\begin{align}
    &\TV(p_{\IQP}, q) = \frac{1}{2}\sum_{\mathbf{x}}|p_{\IQP}(\mathbf{x})-q(\mathbf{x})| \\
    &= \frac{1}{2}\sum_{\mathbf{x}}|p_{\IQP}(\mathbf{x}) - p_{\boldsymbol\theta}(\mathbf{x}) + p_{\boldsymbol\theta}(\mathbf{x})-q(\mathbf{x})|\\
    &\leq \frac{1}{2}\sum_{\mathbf{x}}|p_{\IQP}(\mathbf{x}) -p_{\boldsymbol\theta}(\mathbf{x})| +\frac{1}{2}\sum_{\mathbf{x}}|p_{\boldsymbol\theta}(\mathbf{x})-q(\mathbf{x})|\\
    &\leq \delta+\gamma \equiv \epsilon
\end{align}
where the third line follows from the triangle inequality. Therefore $\mathsf{C}$ could simulate an $\IQP$ distribution also, and we arrive at a contradiction. 

The major open question left by this work is whether QLS is possible at all; can a quantum model outperform all classical ones in generative learning? \red{This idea motivated our search for metrics which upper bound $\TV$, but yet were efficiently computable and therefore could be minimised to efficiently learn distributions to a sufficiently small value of $\TV$. Unfortunately, we can see from the exponential scaling observed in \eqref{tv_wasserstein_sinkhorn_inequality}, which gives the upper bound on $\TV$ by regularised $\OT$, that $\SH$ will not provably achieve this particular task, despite achieving our primary goal of being stronger than the $\MMD$ for generative modelling. We briefly discuss avenues of future research in the Discussion, which could provide alternative routes to QLS.}

\begin{figure*}
    \centering
    \includegraphics[width=1.8\columnwidth, height=0.4\columnwidth]{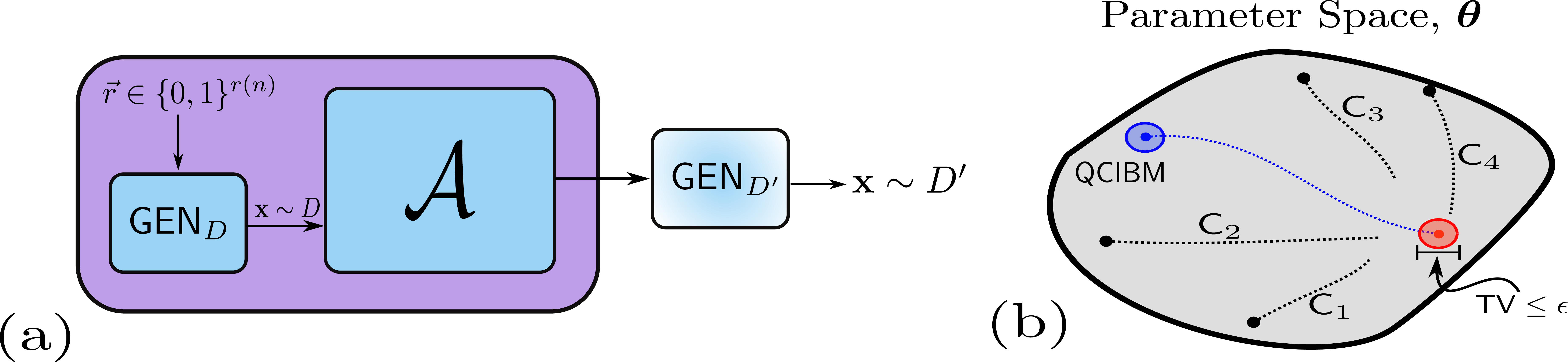}
\caption{(a) Illustration of a learning procedure using a Generator. The algorithm 
$\mathcal A$ is given access to $\GEN_D$, which provides samples, $\mathbf{x}\sim D$, and must output a Generator for a distribution which is close to the original. We allow the target generator to be classical, hence it may take as input a string of random bits, of size polynomial in $n$, $r(n)$, if not able to generate its own randomness. (b) Crude illustration of quantum learning supremacy. No classical algorithm, $\mathsf{C}_i$, should be able to achieve the required closeness in total variation to the target distribution, but the $\IBM$ (or similar) should be able to, for \textit{some} class of target distributions. There should be some path in the parameter space of the $\IBM$, $\boldsymbol\theta$, which achieves this. }
    \label{fig:learning_supremacy}
\end{figure*}

\subsection*{Quantum Compiling}

As a concrete application of such a model outside the scope of classical generative modelling, we can use the $\IBM$ training to perform a type of `weak' quantum circuit compilation. There are potentially other areas which could be studied using these tools, or by applying techniques in generative modelling to other quantum information processing tasks, but this is beyond the scope of this work.

The major objective in this area is to compile a given target unitary, $U$, into one which consists exclusively of operations available to the native hardware of the quantum computer in question. For example, in the case of Rigetti's Aspen QPU, the native gates are $\{R_x(\pm \pi/2), R_z(\theta), CZ\}$ \cite{smith_practical_2016, khatri_quantum-assisted_2019}, and any unitary which a user wishes to implement must be compiled into a unitary $V$ which contains only these ingredients. 

Potential solutions to this problem\cite{jones_quantum_2018, khatri_quantum-assisted_2019} involve approximating the target unitary by assuming that $V$ is a parametric circuit built from the native gates, which can be trained by some optimisation strategy. We adopt a similar view here, but we do not require any extra quantum resources to perform the compilation. With this limitation, we make a tradeoff in that we are not guaranteed to apply \textit{the same} target unitary, only that the output distribution will be close to that produced by the target. Clearly this is a much weaker constraint than the task of \red{direct compilation}, since many unitaries may give rise to the same distribution, but it is much closer to the capabilities of near term devices. To illustrate this application, we train an $\IBM$ to learn the output distribution of a random $\IQP$ circuit when restricted to a $\QAOA$ architecture itself using $\mathcal{L}_{\SH}$ as a cost function. \red{The process is illustrated in \eqref{eqn:ibm_compilation_qaoa_to_IQP}, where we try to determine suitable $\QAOA$ parameters, $\{J^{\QAOA}_{ij}, b^{\QAOA}_k\}$, which reproduce the distribution observed from a set of random $\IQP$ parameters, $\{J^{\IQP}_{ij}, b^{\IQP}_k\}$.}
\begin{multline} \label{eqn:ibm_compilation_qaoa_to_IQP}
    \IBM\left(\left\{J^{\QAOA}_{ij}, b^{\QAOA}_k\right\}, \left\{\Gamma_k = \frac{\pi}{4}\right\}, 0,  0\right) \\
    \overset{Compile}{\rightarrow}   \IBM\left(\left\{J^{\IQP}_{ij}, b^{\IQP}_k\right\}, \right.\\
    \left.\left\{\Gamma_k = \frac{\pi}{2\sqrt{2}}\right\}, 0, \left\{\Sigma_k = \frac{\pi}{2\sqrt{2}}\right\}\right)
\end{multline}
The measurement unitary at the end of the circuit makes this process non-trivial, since this will give rise to significantly different distributions, even given the same parameters in $U_z$. We illustrate this in \figref{fig:autocompilationtwoqubits} using the Rigetti {\fontfamily{cmtt}\selectfont 2q-qvm} and for three qubits in Supplementary Material Section V. We find that even though the learned parameter values are different from the target, the resulting distributions are quite similar, as expected.

\begin{figure*}
    \centering
    \includegraphics[width=2\columnwidth, height=0.45\columnwidth]{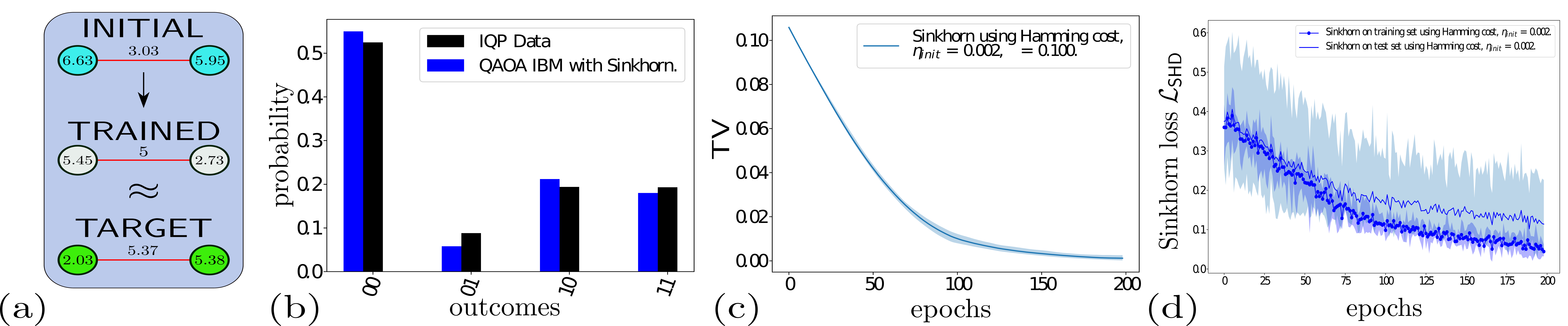}
\caption{Automatic Compilation of $\IQP$ circuit to a $p = 1 \QAOA$ circuit with two qubits using $\mathcal{L}_{\SH}^\epsilon$ with $\epsilon = 0.1$. 500 data samples were used with 400 used for a training set, and 100 used as a test set. $\IBM$ circuit is able to mimic the target distribution well, even though actual parameter values, and circuit families are different. Error bars represent mean, maximum and minimum values achieved over 5 independent training runs on the same data set. (a) Initial [\crule[cyan]{0.2cm}{0.2cm}] and trained [\crule[Lavender]{0.2cm}{0.2cm}] $\QAOA$ circuit parameters for two qubits. Target $\IQP$ circuit parameters [\crule[green]{0.2cm}{0.2cm}]. Parameter values scaled by a factor of 10 for readability. (b) Final learned probabilities of $\IBM$ ($\QAOA$) [\crule[blue]{0.2cm}{0.2cm}] circuit versus `data' probabilities ($\IQP$) [\crule[black]{0.2cm}{0.2cm}]. (c) Total Variation distance and (d) Sinkhorn divergence for 400 training samples, 100 test samples, using a Hamming optimal transport cost.}
    \label{fig:autocompilationtwoqubits}
\end{figure*}

\section*{Discussion}\label{sec:discussion}

Providing provable guarantees of the superior performance of near term quantum computers relative to any classical device for some particular non-trivial application is an important milestone of the field. We have shown one potential route towards this goal by combining complexity-theoretic arguments\cite{bremner_classical_2011, aaronson_computational_2013, boixo_characterizing_2018}, with an application in generative machine learning\cite{gao_efficient_2017, benedetti_generative_2019, liu_differentiable_2018, du_expressive_2018}, and improved training methods of generative models. Specifically we introduced the Ising Born machine, a restricted form of a quantum circuit Born machine. These models utilise the Born rule of quantum mechanics to train a parameterised quantum circuit as a generative machine learning model, in a hybrid manner.

We proved that the model cannot be simulated efficiently by any classical algorithm up to a multiplicative error in the output probabilities, which holds for many circuit families that may be encountered during gradient based training. As such, this type of model is a good candidate for a provable quantum advantage in quantum machine learning using NISQ devices. To formalise this intuition, we defined a notion of quantum learning supremacy to rigorously define what such an advantage would look like, in the context of machine learning.

We adapted novel training methods for generative modelling in two ways. Firstly, by introducing quantum kernels to be evaluated on quantum hardware, and secondly by proposing and adapting new cost functions. In the case of the Sinkhorn divergence, we discussed its sample complexity and used this to define a somewhat optimal cost function through a judicious choice of the regularisation parameter. It is possible to choose this parameter such that the cost is efficiently computable even as the number of qubits grows. We showed numerically that these methods have the ability to outperform previous methods in the random dataset we used as a test case. 

Finally, we demonstrated an application of the model as a heuristic compiler to compile one quantum circuit into another via classical optimisation techniques, which has the advantage of requiring minimal quantum overhead. These techniques could potentially be adapted into methods to benchmark and verify near term quantum devices.

The major question that this work raises is whether or not a provable notion of quantum learning could be achievable for a particular dataset, thereby solidifying a use case for quantum computers in the near term with provable advantage. The best prospect for this is the quantum supremacy distributions we know of (for example $\IQP$), but they are not efficiently testable\cite{hangleiter_sample_2019}. Due to this, they are also likely to not be efficiently learnable either, given the close relationship between distribution testing and learning\cite{goldreich_property_1998}. \red{Furthermore, we can see from the exponential scaling required in \eqref{tv_wasserstein_sinkhorn_inequality} for regularised $\OT$ to upper bound $\TV$, that other techniques are necessary to achieve QLS, since the methods we present here are not suited to this particular task, despite achieving our goal of being stronger than the $\MMD$ for generative modelling.} However, this assumes we have access only to classical samples from the distribution, and the possibility of gaining an advantage using \textit{quantum} samples\cite{schuld_supervised_2018, arunachalam_survey_2017} is unexplored in the context of distribution learning.

\section*{Methods}\label{sec:methods}
In this section, we detail the methods used to train the $\IBM$ to reproduce a given probability distribution. The target distribution is the one given by \eqref{toydatadistribution}, which is used in both refs.\cite{amin_quantum_2018, verdon_quantum_2017} to train versions of the quantum Boltzmann machine:
\begin{equation}
    \pi(\mathbf{y}) \coloneqq \frac{1}{T}\sum\limits_{k=1}^T 
    p^{n - d_H(s_k, \mathbf{y})}
    (1-p)^{d_H(s_k, \mathbf{y})}
\label{toydatadistribution}
\end{equation}

\noindent To generate this data, $T$ 
binary strings of length $n$, written $s_k$ and called `modes',  
are chosen randomly. 
A sample $\mathbf{y}$ is \red{then} produced with a probability which depends on its Hamming distance $d_H(s_k, \mathbf{y})$ to each mode. In all of the above, the Adam \cite{kingma_adam:_2014} optimiser was applied, using the suggested hyperparameters i.e.\@ $\beta_1 = 0.9, \beta_2 = 0.999, \epsilon = 1\times 10^{-8}$, and initial learning rate, $\eta_{\mathsf{init}}$. This was chosen since it was found to be more robust to sampling noise \cite{liu_differentiable_2018}.

In all of the numerical results, we used a $\QAOA$ structure as the underlying circuit in the $\IBM$. Specifically, the parameters in $U_f$ were chosen such that $\forall k, \Gamma_k = \pi/4, \Delta_k = 0, \Sigma_k = 0$. The Ising parameters, $\{J_{ij}, b_k\}$ were initialised randomly. 

For the Stein discrepancy, we used three Nystr\"{o}m eigenvectors to approximate the Spectral score in \figref{fig:MMDvSinkvStein3} for three qubits, and 6 eigenvectors for 4 qubits. In all cases when using the $\MMD$ with a Gaussian kernel, we chose the bandwidth parameters, $\sigma = [0.25, 10, 1000]$\cite{liu_differentiable_2018}.

\subsection*{Data Availability}\label{sec:data_availability}
Data and simulations presented in this work are available from the corresponding author upon request. Code used in this work is available from the corresponding author upon request, or on Github\cite{brian_coyle_briancoyleisingbornmachine_2020}.
\subsection*{Acknowledgements}\label{sec:acknowledgements}

B.C.\@ thanks Andru Gheorghiu for useful discussions and title suggestion. We also thank Jean Feydy, Patric Fulop, Vojtech Havlicek and Jiasen Yang for clarifying pointers. We thank Atul Mantri for comments on the manuscript. This work was supported by the Engineering and Physical Sciences Research Council (grants EP/L01503X/1, EP/N003829/1), EPSRC Centre for Doctoral Training in Pervasive Parallelism at the University of Edinburgh, School of Informatics and Entrapping Machines, (grant FA9550-17-1-0055). We also thank Rigetti Computing for the use of their quantum compute resources, and views expressed in this paper are those of the authors and do not necessarily reflect the views or policies of Rigetti Computing.

\subsection*{Author Contributions}
B.C.\@ devised the theoretical aspects of the work, and wrote the code for the numerical results with the help of D.M.\@; D.M.\@ contributed to the learning supremacy definitions. 
E.K.\@ and V.D.\@ supervised the work. All authors contributed to the manuscript writing. 

\subsection*{Competing Interests}
The authors declare no conflict of interests.


\bibliographystyle{naturemag}

\end{multicols}


\appendix
\renewcommand{\thesection}{\Roman{section}}
\renewcommand{\figurename}{{Supplementary Figure}}

\renewcommand{\thefigure}{\arabic{figure}}

\setcounter{figure}{0}
\newpage
\singlespacing

\appendix
\renewcommand{\thesection}{\Roman{section}}
\renewcommand{\figurename}{{Supplementary Figure}}

\renewcommand{\thefigure}{\arabic{figure}}

\makeatletter
\apptocmd{\thebibliography}{\global\c@NAT@ctr 68\relax}{}{}
\makeatother

\setcounter{figure}{0}
\newpage
\singlespacing
\section*{Supplementary Material}

\section{Integral Probability Metrics as Cost Functions}
\label{supp_matt:cf_and_ipms}

In this section, we provide further details about the cost functions used in this work. As mentioned in the main text, one distribution measure which could be used as a training cost function is the Kullback-Leibler ($\KL$) Divergence. The cost function associated with the $\KL$ divergence (also known as the negative log-likelihood) is given by:
\begin{equation}
    \mathcal{L}_{\KL}(p_{\boldsymbol\theta}, \pi) \coloneqq -\sum\limits_{\mathbf{z}} \pi(\mathbf{z})\log\left(p_{\boldsymbol\theta}(\mathbf{z})\right)  = -\mathbb{E}_{\mathbf{z}\sim\pi}(\log(p_{\boldsymbol\theta}(\mathbf{z}))\label{klcostfunction}
\end{equation}
The gradient of the $\KL$ cost, $\mathcal{L}_{\KL}$, cannot be estimated efficiently for these types of parameterised quantum circuit (PQC) models\citeS{liu_differentiable_2018} since it has the following form:
\begin{equation}
    \frac{\partial\mathcal{L}_{\KL}}{\partial {\boldsymbol\theta}} 
     \sim \sum\limits_{\mathbf{z}}\frac{\pi(\mathbf{z})}{{p_{\boldsymbol\theta}(\mathbf{z})}}\left(p^-_{\boldsymbol\theta}(\mathbf{z})- p^+_{\boldsymbol\theta}(\mathbf{z})\right) \label{klgradient}
\end{equation}
As noted in refs. \citeS{bremner_classical_2011, liu_differentiable_2018}, computing the required probabilities, $p_{\boldsymbol\theta}$, is $\#\Pee$-hard, and so \eqref{klgradient} cannot be computed efficiently.
Furthermore, the $\KL$ Divergence is an example of a so-called $f$-divergence \citeS{sriperumbudur_integral_2009}, and such measures are notoriously hard to estimate in practice. A potential solution to this is to use an alternative family of discrepancies, called \textit{Integral Probability Metrics}\citeS{muller_integral_1997} (IPMs) which are defined by the following equation:

\begin{align}
    \gamma_{\mathcal{F}}(p_{\boldsymbol\theta}, \pi) &  \coloneqq \sup_{\phi \in \mathcal{F}}\left|\int_{\mathcal{M}}\phi dp_{\boldsymbol\theta}-\int_{\mathcal{M}}\phi d\pi\right| \label{ipmequation}\\
    &= \sup_{\phi \in \mathcal{F}}\left(\mathbb{E}_{p_{\boldsymbol\theta}}[\phi]- \mathbb{E}_{\pi}[\phi]\right) 
\end{align}
Where $p_{\boldsymbol\theta}, \pi$ are defined over a measurable space, $\mathcal{M}$. The class of functions which are chosen, $\mathcal{F}$, defines the specific metric.

The $\MMD$ is one of the simplest such example of an IPM, since it is relatively easy to compute and has a low sample complexity, as discussed in the main text. The $\MMD$ was first used in ref.\citeS{gretton_kernel_2012} as a hypothesis test, to determine if two distributions are identical. For a more thorough treatment of the properties of this metric, see refs. \citeS{sriperumbudur_universality_2011, sriperumbudur_hilbert_2010}.  

The $\MMD$ can be extracted from \eqref{ipmequation} by restricting the class of functions, $\mathcal{F}$, to be a unit ball in a Hilbert space, $\mathcal{H}$:
\begin{equation}
    \mathcal{F}_{\MMD}   \coloneqq \{\phi \in \mathcal{H}: ||\phi||_{\mathcal{H}}\leq 1\} \rightarrow \gamma_{\MMD} \label{mmd_metric_supp}
\end{equation}
where $||\cdot||_{\mathcal{H}}$ denotes the norm in $\mathcal{H}$.

The second IPM which is relevant here is the total variation ($\TV$) distance. In fact, the $\TV$ was shown in ref. \citeS{sriperumbudur_integral_2009}, to be the only cost function which is both an IPM, and an $f$-divergence. The $\TV$ can be obtained by taking the function space as follows:
\begin{equation}
    \mathcal{F}_{\TV}   \coloneqq \{\phi: ||\phi||_{\infty} \leq 1\} \label{totalvariationfunctions}
\end{equation}
Where $||\phi||_{\infty} = \sup\{|\phi(\mathbf{x})|, \mathbf{x} \in \mathcal{M}\}$. Also, when working with a discrete (more generally countable) space, $\mathcal{M} = \mathcal{Y}$, the total variation is given by:
\begin{equation}
    \TV(p_{\boldsymbol\theta}, \pi)  \coloneqq  \frac{1}{2}\sum\limits_{\mathbf{x} \in \mathcal{Y}}|p_{\boldsymbol\theta}(\mathbf{x})- \pi(\mathbf{x})| \label{discretetotalvariation}
\end{equation}
This measure is particularly important, since it occurs in the definition of additive error simulation hardness in \appref{supp_matt:ibm_variation_hardness}, and since it is the measure we use as the relative benchmark for the numerical results in the main text, and in \appref{supp_matt:numericalresults}.

Finally, the third useful IPM is the Kantorovich Metric, defined by:
\begin{equation}
    \mathcal{F}_{\W} \coloneqq \{\phi: ||\phi||_{L} \leq 1\} \label{kantorovichmetric}
\end{equation}
where $||\cdot||_{L}$ is the Lipschitz semi-norm, defined by: $||\phi||_{L} := \sup\{|\phi(x) - \phi(y)|/d(x,y) : x\neq{y} \in \mathcal{Y}\}$, where $d$ is a metric on $\mathcal{Y}$. It turns out that, due to the Kantorovich-Rubinstein theorem\citeS{dudley_real_2002}, this is also equivalent to the \textit{Wasserstein Metric} and related to the notion of Optimal Transport, discussed in \appref{supp_matt:sinkhorn}.


\section{Kernel Methods and the \texorpdfstring{$\MMD$}{MMD}}\label{supp_matt:kernelmethods_mmd}
In this section, we elaborate on the kernel methods we use and provide the relevant definitions. Kernel methods are a popular technique, for example in dimensionality reduction, and are beginning to be utilised in quantum machine learning \citeS{havlicek_supervised_2019, schuld_quantum_2019}. A kernel is a symmetric function, $\kappa:\mathcal{Y}\times\mathcal{Y} \rightarrow \mathbb{R}$, which is positive definite, meaning that the matrix it induces, called a \textit{Gram} matrix, is positive semi-definite.

In models like support vector machines (SVM), the idea is to to \emph{embed} the underlying sample space, $\mathcal{Y}$, into a Hilbert space using a non-linear map, $\phi:\mathcal{Y} \rightarrow \mathcal{H}$, called a \emph{feature map}. A kernel is simply defined by the inner product of this feature map at two different points in a Reproducing Kernel Hilbert Space (RKHS). Intuitively, the map should be designed to make the points in the RKHS more easily comparable.  The choice of the kernel function (i.e.\@ the choice of the RKHS) may allow different properties to be compared.  For a comprehensive review of techniques in kernel embeddings, see ref. \citeS{muandet_kernel_2017}. In the following, we assume that the sample space is discrete ($\mathcal{Y} = \mathcal{X}^n$) to be consistent with the main text, but the definitions we present are more generally applicable.

Every feature map has a corresponding kernel, given by the inner product on the Hilbert Space:
\begin{theorem}[Feature Map Kernel]
    Let $\phi: \mathcal{X}^n \rightarrow \mathcal{H}$ be a feature map. The inner product of two inputs mapped to a feature space defines a kernel via:
    \begin{equation}
        \kappa(\mathbf{x},\mathbf{y})  \coloneqq \langle\phi(\mathbf{x}),\phi(\mathbf{y}) \rangle_{\mathcal{H}} \label{featurekernel}
    \end{equation}
\end{theorem}
The full proof that this feature map gives rise to a kernel, a positive semi-definite and symmetric function, can be found in ref.\citeS{schuld_quantum_2019} for quantum feature maps, $\phi :\mathbf{x}\rightarrow \ket{\phi(\mathbf{x})}$.

The Hilbert Space is `\emph{reproducing}' because the inner product of the kernel at a point in the sample space, with a function on the Hilbert space, in some sense \emph{evaluates} or reproduces the function at that point:  $\langle f, \kappa(\mathbf{x}, \cdot)\rangle  = f(\mathbf{x})$ for $f \in \mathcal{H}$, 
$\mathbf{x} \in \mathcal{X}^n$. 

We can also define the \emph{mean embedding}, which is a relevant quantity in the definition of the $\MMD$, and which has also been defined in the quantum setting\citeS{kubler_quantum_2019}:
\begin{equation}
    \mu_{p_{\boldsymbol\theta}}  \coloneqq \mathbb{E}_{p_{\boldsymbol\theta}}(\kappa(\mathbf{x}, \cdot)) =\mathbb{E}_{p_{\boldsymbol\theta}}[\phi(\mathbf{x})] \in \mathcal{H} \label{meanembedding}
\end{equation}
It can be shown that the \textsf{MMD} (\eqref{mmd_metric_supp}) is exactly the difference in mean embeddings (\eqref{meanembedding}) between the two distributions\citeS{borgwardt_integrating_2006, gretton_kernel_2007}:
\begin{equation}
    \gamma_{\MMD}(p_{\boldsymbol\theta}, \pi)   \coloneqq ||\mu_{p_{\boldsymbol\theta}} - \mu_{\pi}||_{\mathcal{H}}\label{mmd_mean_embedding}
\end{equation}
As mentioned above, we wish to use \emph{quantum} kernels, i.e.\@ those arising as a result of a state overlap in a quantum RKHS. This notion was first presented by ref. \citeS{schuld_quantum_2019} where the form of the kernel assumed is the one induced by the inner product on a quantum (i.e.\@ complex) Hilbert space:
\begin{equation}
\kappa(\mathbf{x}, \mathbf{y})  \coloneqq \braket{\phi(\mathbf{x})|\phi(\mathbf{y})} \label{amplitudekernel}
\end{equation}
Given that we are trying to exploit some advantage by using quantum computers, this definition is natural. However, it should be noted that the above definition of a kernel defines it to be a mapping from the sample space to $\mathbb{C}$, i.e.\@ the kernel is the inner product of two wave functions which encode two samples and is therefore a transition amplitude. As mentioned in ref.\citeS{schuld_quantum_2019}, it would be desirable to find kernels which are hard to compute classically, in order to gain some quantum advantage. The first potential candidate for such a situation was derived by ref.\citeS{havlicek_supervised_2019}. As such, to remain consistent with the authors notation\citeS{havlicek_supervised_2019}, the kernel is defined as the real transition \emph{probability} instead. As mentioned in ref.\citeS{schuld_quantum_2019}, it is possible to define a kernel this way by taking the modulus squared of \eqref{amplitudekernel}, and defining the RKHS as follows:

\begin{theorem}[Quantum RKHS]
let $\Phi:\mathcal{X}\rightarrow \mathcal{H}$ be a feature map over an input set $\mathcal{X}$, giving rise to a \emph{real} kernel: $\kappa(\mathbf{x}, \mathbf{y}) = |\braket{\Phi(\mathbf{x})|\Phi(\mathbf{y})}|^2$. The corresponding RKHS is therefore:
\begin{equation}
    \mathcal{H}_\kappa  \coloneqq \{f:\mathcal{X} \rightarrow \mathbb{R}|\ f(\mathbf{x}) = |\braket{w|\Phi(\mathbf{x})}|^2, \forall \mathbf{x} \in \mathcal{X}, w\in \mathcal{H}\} 
\end{equation}
\end{theorem}

In ref.\citeS{schuld_quantum_2019}, the authors investigate kernels which are classically efficiently computable in order to facilitate testing of their classification algorithm, for example they only study so-called Gaussian states, a key ingredient in continuous variable quantum computing, \citeS{lloyd_quantum_1999}, which are known to be classically simulable. However, to gain a potential quantum advantage, we use the kernel proposed in \citeS{havlicek_supervised_2019}, which is defined by a feature map constructed by the following quantum circuit:

\begin{align}
    \Phi: \mathbf{x} &\in \mathcal{X}^n \rightarrow \ket{\Phi(\mathbf{x})} \label{quantumfeaturemap}\\
    \ket{\Phi(\mathbf{x})}  \coloneqq \mathcal{U}_{\Phi(\mathbf{x})}\ket{0}^{\otimes n} &= U_{\Phi(\mathbf{x})}H^{\otimes n}U_{\Phi(\mathbf{x})}H^{\otimes n}\ket{0}^{\otimes n} \label{kernelcircuit}
\end{align}
so the resulting kernel is:
\begin{equation}
    \kappa_Q(\mathbf{x}, \mathbf{y})  \coloneqq |\braket{\Phi(\mathbf{x})|\Phi(\mathbf{y})}|^2 \label{quantumkernel_app}
\end{equation}
A further motivation for why this particular type of circuit is chosen, is its relationship to the Ising model. This can be seen through the choice of the unitary encoding operators, $U_{\Phi(\mathbf{x})}$:
\begin{equation}
    U_{\Phi(\mathbf{x})}  \coloneqq \exp\left(\mathrm{i}\sum\limits_{S\subseteq [n]}\phi_S(\mathbf{x})\prod\limits_{i \in S}Z_i\right) \label{kernelcircuitunitary}
\end{equation}
This is the same form as the diagonal unitary in the $\IBM$, \eqref{diagonalunitary}, with the parameters, ${\boldsymbol\theta}$ replaced by the feature map with sample $\mathbf{x}$, $\phi_S(\mathbf{x})$. However, this is \emph{not} an \textsf{IQP} circuit because of the extra final layer of diagonal gates in \eqref{kernelcircuit}. It is noted in ref.\citeS{havlicek_supervised_2019} that if we only care about computing the overlap between the states to within a multiplicative error (or exactly) it is sufficient to ignore the second layer of diagonal gates, $U_{\Phi(\mathbf{x})}$, in the feature map, \eqref{quantumfeaturemap}, as it will be $\#\Pee$-hard for some sample values. However, to rule out an \emph{additive} error approximation, it is necessary to add the second encoding layer.

It should be noted that this is experimentally favourable to work alongside the original $\IBM$ circuit since the same setup (i.e.\@ layout of entanglement gates, and single qubits rotations) is required for both circuit types, albeit with different parameters. Just like in the Ising scenario, only single and two-qubit operations are required:
\begin{equation}
	U_{\phi_{\{l,m\}}(\mathbf{x})} = \exp\left(\mathrm{i}\phi_{\{l,m\}}	(\mathbf{x})Z_l\otimes Z_m\right) \qquad U_{\phi_{\{k\}}(\mathbf{x})} = \exp\left(\mathrm{i}\phi_{\{k\}}(\mathbf{x})Z_k\right) \label{kernelgates}
\end{equation}
The arguments of the gates in \eqref{kernelgates} used to encode the samples is given by:
\begin{equation}
    \phi_{\{l, m\}}(\mathbf{x})   \coloneqq \left(\frac{\pi}{4} - x_l\right)\left(\frac{\pi}{4} - x_m\right) \qquad
    \phi_{\{k\}}(\mathbf{x})  \coloneqq \frac{\pi}{4}x_k \label{quantumfeaturemapvariableencoding}
\end{equation}
The choice of this kernel is motivated by ref.\citeS{havlicek_supervised_2019}, which conjectures that this overlap, \eqref{quantumkernel_app}, will be classically hard to estimate up to polynomially small error, whereas the kernel can be computed using a quantum device, up to an additive sampling error. Furthermore, the particular encoding function in \eqref{quantumfeaturemapvariableencoding} was found to be a good all-purpose encoding for classification\citeS{suzuki_analyzing_2019}. A rough bound for the number of measurements required to compute the full matrix is also given by ref.\citeS{havlicek_supervised_2019}. The only difference is that in that case, the Gram matrix was computed using only samples from the same source. In our example, this need not be the case. We see that it in general, it requires $|B|$ samples from one source (the $\IBM$ for example), and $|D|$ samples from another (the data). The resulting Gram matrix will then be of dimension $|B| \times |D|$. However, to simplify the expressions, we can assume that we have the same number of samples from each: $|B| = |D| = T$. Therefore, directly from ref.\citeS{havlicek_supervised_2019}, to compute a single kernel entry to precision $\tilde{\epsilon} = \mathcal{O}(R^{-1/2})$ requires $R  = \mathcal{O}(\tilde{\epsilon}^{-2})$ measurement shots. Then an approximation, $\Hat{K}$, of the Gram matrix for the kernel, $K_{\mathbf{x}, \mathbf{y}} = \kappa(\mathbf{x}, \mathbf{y})$, which has $T(T - 1)$  non-trivial entries, that is $\epsilon$-close in operator norm $||K - \Hat{K}|| \leq \epsilon$, can be determined using $R = \mathcal{O}(\epsilon^{-2}T^4)$ measurement shots.

Now that the kernel has been introduced, we can examine how this appears in the cost functions we use, starting with the $\MMD$. To exploit the kernel trick, we do not directly use $\gamma_{\MMD}$ in \eqref{mmd_metric_supp} as a cost, instead its squared version, which is exactly the one used in refs.\citeS{liu_differentiable_2018, du_expressive_2018} for their versions of the Born Machine:
\begin{equation}
    \mathcal{L}_{\MMD}  \coloneqq \gamma_{\MMD}(p_{\boldsymbol\theta}, \pi)^2 = ||\mathbb{E}_{p_{\boldsymbol\theta}}[\phi(\mathbf{x})] - \mathbb{E}_{\pi}[\phi(\mathbf{x})]||_{\mathcal{H}}^2 \label{mmdloss}
\end{equation}
By expanding the inner product in \eqref{mmdloss}, one exactly retrieves the form seen in \eqref{mmdexact}, where $\kappa$ is the $\MMD$ kernel. A requirement for the $\MMD$ to be useful (as a hypothesis test) is that the kernel used to define it must be \textit{characteristic} or universal\citeS{fukumizu_kernel_2007, sriperumbudur_injective_2008}. This is one property which allows it to be a valid metric on the space of probability distributions, with $\gamma_{\MMD}(p_{\boldsymbol\theta}, \pi) = 0 \iff p_{\boldsymbol\theta} \equiv \pi$. This enables the $\MMD$ to be used in hypothesis testing to determine if the null hypotheses ($p_{\boldsymbol\theta} \equiv \pi$) holds or not. The Gaussian kernel (\eqref{gaussiankernel}) is indeed one which is characteristic \citeS{fukumizu_kernel_2007}, and some effort has been made to find conditions under which a kernel is characteristic \citeS{sriperumbudur_universality_2011}.

In the case where the support of the distributions is discrete, as is the case here, the condition on the kernel being characteristic is \textit{strict positive definiteness} \citeS{sriperumbudur_integral_2009}. Strict positive definiteness of the kernel is defined as the resulting Gram matrix being positive definite. A matrix is positive definite $\iff$ $\forall \mathbf{x} \in \mathbb{R}^d/\mathbf{0}, \mathbf{x}^TK\mathbf{x} > 0$. 

The $\MMD$ can be estimated\citeS{sriperumbudur_integral_2009} (in an unbiased way), given i.i.d.\@ samples from two distributions, $\hat{\mathbf{x}} = (\mathbf{x}^1,\dots, \mathbf{x}^N) \sim p_{\boldsymbol\theta}$, $\hat{\mathbf{y}} = (\mathbf{y}^1, \dots, \mathbf{y}^M) \sim \pi$.
This is done by simply replacing the expectation values in \eqref{mmdexact} in the main text with their empirical values.
\begin{equation}
      \tilde{\mathcal{L}}_{\MMD}   \coloneqq \tilde{\gamma}_{\MMD}({p_{\boldsymbol\theta}}_N, {\pi}_M)^2 =\frac{1}{N(N-1)}\sum\limits_{i \neq j}^N \kappa(\mathbf{x}^i, \mathbf{x}^j) + \frac{1}{M(M-1)}\sum\limits_{i \neq j}^M \kappa(\mathbf{y}^i, \mathbf{y}^j) -\frac{2}{MN}\sum\limits_{i ,j}^{M, N} \kappa(\mathbf{x}^i, \mathbf{y}^j)   \label{mmdlossestimator_supp}
\end{equation}
As shown in ref.\citeS{sriperumbudur_integral_2009}, the above \eqref{mmdlossestimator_supp} demonstrates both consistency and lack of bias, and furthermore it converges fast in probability to the true $\MMD$:
\begin{equation}
    |\tilde{\gamma}_{\MMD}({p_{\boldsymbol\theta}}_N, \pi_M)  -\gamma_{\MMD}(p_{\boldsymbol\theta}, \pi)| \leq \mathcal{O}_{p_{\boldsymbol\theta}, \pi}(N^{-1/2} + M^{-1/2})\label{mmdsamplecomplexity}
\end{equation} 
This quadratic convergence rate is highly desirable, since it does not depend on the dimension of the space from which the samples are drawn. 

The derivation of the $\MMD$ loss estimator requires that the kernel be a bounded and measurable function. This is the case for Gaussian kernels, but we must check it holds for the quantum kernel, \eqref{quantumkernel_app}. Happily, this is the case: the overlap between two states is bounded above by $1$, and the sample space consists of binary strings.
\begin{equation}
    \kappa_Q(\mathbf{x}, \mathbf{y}) = |\braket{\Phi(\mathbf{x})|\Phi(\mathbf{y})}|^2 \leq 1 \qquad \forall \mathbf{x}, \mathbf{y} \in \mathcal{X}
\end{equation}
Now, to compute the derivative of this cost function, with respect to the parameters, $\theta_k$, we follow the method of ref.\citeS{liu_differentiable_2018}. This approach applies to all cost functions we employ in this work. In ref.\citeS{liu_differentiable_2018}, the quantum gates with trainable parameters, $\eta$, are of the form $U(\eta) = \exp(-i\frac{\eta}{2}\Sigma)$, where $\Sigma$ satisfies $\Sigma^2 = \mathds{I}$. Any tensor product of Pauli operators satisfies this requirement.

It is known when the gates in the circuit have the above form\citeS{mitarai_quantum_2018, liu_differentiable_2018}, that the gradient of an observable of the circuit, $B$, with respect to a parameter, $\eta$, is given by:
\begin{equation}
    \frac{\partial \langle B \rangle_\eta}{\partial \eta} = \frac{1}{2}\left(\langle B \rangle_{\eta^+} - \langle B \rangle_{\eta^-}\right) \label{liucircuitgradient}
\end{equation} 
where $\eta^{\pm}$ are parameter shifted versions of the original circuits. As noted in ref.\citeS{liu_differentiable_2018}, taking the observable to be a measurement in the computational basis, $B= M_{\mathbf{z}} =  \ket{\mathbf{z}}\bra{\mathbf{z}}$, we arrive at the following form of the gradient of the probabilities:
\begin{equation}
    \frac{\partial p_{\boldsymbol\theta}(\mathbf{z})}{\partial \theta_k} = p_{\theta_k}^{-}(\mathbf{z}) - p_{\theta_k}^{+}(\mathbf{z}) \label{mmdprobabilitygradient_supp}
\end{equation}
The factor of $-1/2$ difference with this gradient from that of ref.\citeS{liu_differentiable_2018}  is due to the slightly different parameterisation we choose, our gates are instead of the form: $\exp(i\eta\Sigma): \{D_1(b_k) =  \exp{\mathrm{i}b_k Z_k}, D_2(J_{ij}) = \exp{\mathrm{i}J_{ij}Z_iZ_j}\}$.
As mentioned above, this gradient requires computing the parameter shifted versions of the original circuit, $p_{\boldsymbol\theta_k}^{\pm}$. This notation indicates that the $\text{k}^{th}$ parameter in the original circuit has been shifted by a factor of $\pm \frac{\pi}{2}$; $p_{{\boldsymbol\theta}_k}^{\pm} = p_{{\boldsymbol\theta}_k \pm \pi/2}$. This formula is actually also valid for the parameter in any unitary which has \textit{at most} two distinct eigenvalues \citeS{schuld_evaluating_2019}.

Using this, we arrive at the $\MMD$ gradient from the main text, \eqref{mmdgradient}, which can be approximated using the following estimator:
\begin{equation}
    \frac{\partial \mathcal{L}_{\MMD}}{\partial \theta_k} \approx \frac{2}{PN}\sum\limits^{P,N}_{p, i}\kappa(\mathbf{a}^p, \mathbf{x}^i) - \frac{2}{QN}\sum\limits_{q,n}^{Q, N}\kappa(\mathbf{b}^q, \mathbf{x}^i)- \frac{2}{PM}\sum\limits_{p,m}^{P, M}\kappa(\mathbf{a}^p, \mathbf{y}^m) + \frac{2}{QM}\sum\limits_{q, m}^{Q, M}\kappa(\mathbf{b}^q, \mathbf{y}^m) \label{mmdgradientestimator}
\end{equation}
where we have $P, Q$ samples, $\hat{\mathbf{a}} = \{\mathbf{a}^1,\dots, \mathbf{a}^P\}, \hat{\mathbf{b}} = \{\mathbf{b}^1,\dots, \mathbf{b}^Q\}$ drawn from the parameter shifted circuits, $ p^-_{\theta_k}(\mathbf{a}), p^+_{\theta_k}(\mathbf{b})$ respectively and $N, M$ samples, $\hat{\mathbf{x}} = \{\mathbf{x}^1,\dots, \mathbf{x}^N\}, \hat{\mathbf{y}} = \{\mathbf{y}^1,\dots, \mathbf{y}^M\}$ drawn from the the original Born machine circuit and the data distribution respectively, $p_{\boldsymbol\theta}(\mathbf{x}), \pi(\mathbf{y})$. The parameter shifted circuits are given in the main text (\eqref{circuit:gradientplusminuscircuit}) and repeated here for completeness:
\begin{equation}
\Qcircuit @C=0.6em @R=1em {
\lstick{\ket{0}^{\otimes n}} & \gate{H^{\otimes n}} & \gate{U_{l:k+1}}&\gate{U_k(\theta_k^{\pm})}&\gate{U_{k-1:1}} & \meter &\cw & \rstick{\mathbf{a}^p/\mathbf{b}^q \in \{0,1\}^n} 
 }
\label{circuit:gradientplusminuscircuit_supp}
\end{equation}
where the notation indicates (for $m\geq l$) $U_{l:m} = U_lU_{l-1}\dots U_{m+1}U_m$, and each gate $k$ carries one of the Ising parameters, $\boldsymbol\alpha_k$ or one of the set of measurement unitaries, $\{\mathbf{\Gamma}_k, \mathbf{\Delta}_k, \mathbf{\Sigma}_k\}$, since these also could be trained.

\section{Stein Discrepancy \& Score Approximations \label{supp_matt:stein_discrepenancy}}
In this section we elaborate on the Stein discrepancy ($\SD$) and the methods to approximate the score function mentioned in the main text. The $\SD$ was originally proposed in ref.\citeS{gorham_measuring_2015}, which also provided as a method to compute it using linear programs. This is necessary as in general the computation involves high dimensional integrals, as is also the case for the $\MMD$. A simpler form was derived by ref.\citeS{liu_kernelized_2016}, which introduced a kernelised version, allowing it to be computed in closed form. 

We begin by discussing the original formulation of the discrepancy in the continuous case, before dealing with its discretisation. Firstly, we have Stein's identity (in the case where the sample space is one-dimensional, $x\in \mathcal{X} \subseteq \mathbb{R}$) given by:
\begin{equation}
    \mathbb{E}_\pi\left[\mathcal{A}_\pi\phi(x)\right] = \mathbb{E}_\pi \left[s_\pi(x)\phi(x)+ \nabla_x\phi(x)\right]  = 0 \label{steinidentity}
\end{equation}
where $s_\pi(x) = \nabla_x\log(\pi(x))$ is the \textit{Stein Score} function of the distribution $\pi$, and $\mathcal{A}_\pi$ is a so-called \textit{Stein operator} of $\pi$. The functions, $\phi$, which obey \eqref{steinidentity}, are said to be in the \textit{Stein class} of the distribution $\pi$. From Stein's identity, one can  define a discrepancy  between the two distributions, $p_{\boldsymbol\theta}, \pi$, by the following optimisation problem, \citeS{yang_goodness--fit_2018}:
\begin{equation}
    D_S(p_{\boldsymbol\theta}|| \pi)  \coloneqq \sup_{\phi \in \mathcal{F}}\left(\mathbb{E}_{p_{\boldsymbol\theta}}[\mathcal{A}_\pi\phi] - \mathbb{E}_{p_{\boldsymbol\theta}}[\mathcal{A}_{p_{\boldsymbol\theta}}\phi]\right)^2 \label{steindiscrepancy_supp}
\end{equation}
If $p\equiv q$, then $D_S(p_{\boldsymbol\theta}|| \pi) = 0$ by \eqref{steinidentity}. Exactly as with the $\MMD$, the power of the discrepancy in \eqref{steindiscrepancy_supp}, will depend on the choice of the function space, $\mathcal{F}$. By choosing it to be a RKHS, a kernelised form which is computable in closed form can be obtained. Also, this form of \eqref{steindiscrepancy_supp} is very reminiscent of that of the integral probability metrics, \eqref{ipmequation}.

However, to make the $\SD$ applicable to this case, we must make a key alteration. Equation (\ref{steindiscrepancy_supp}) is only defined for smooth probability densities, $p_{\boldsymbol\theta}, \pi$, which are supported on \textit{continuous} domains (e.g.\@ $\mathbb{R}$) due to the gradient term, $\nabla_x$, in \eqref{steinidentity}. We must perform a discretisation procedure to deal with this. Fortunately, this has been addressed by ref.\citeS{yang_goodness--fit_2018}, which adapted the kernelised $\SD$ to the discrete domain, by introducing a discrete gradient `shift' operator. Just as above, we assume we are dealing with $n$-dimensional sample vectors, $\mathbf{x}\in \mathcal{X}^n \subseteq \mathbb{R}^n$. First of all, we shall need some definitions\citeS{yang_goodness--fit_2018}:

\begin{definition}[Cyclic Permutation]
For a set $\mathcal{X}$ of finite cardinality, a cyclic permutation $\neg:\mathcal{X} \rightarrow \mathcal{X}$ is a bijective function such that for some ordering $x^{[1]},x^{[2]}, \dots, x^{[|\mathcal{X}|]}$ of the elements in $\mathcal{X}$, $\neg x^{[i]} \mapsto x^{[(i+1)\mod|\mathcal{X}|]}, \forall i = 1,2,\dots, |\mathcal{X}|$
\end{definition}
\begin{definition}[Partial Difference Operator and Difference Score Function]
Given a cyclic permutation  $\neg$ on $\mathcal{X}$, for any vector, $\mathbf{x} = (x_1, \dots, x_n)^T \in \mathcal{X}^n$. For any function $f:\mathcal{X}^n \rightarrow \mathbb{R}$, denote the (partial) difference operator as:
\begin{equation}
    \Delta_{x_i}f(\mathbf{x}) := f(\mathbf{x}) - f(\neg_i\mathbf{x}) \qquad \forall i=1, \dots, d \label{discreteshiftoperator}
\end{equation}
with $\Delta f(\mathbf{x}) = (\Delta_{x_1}f(\mathbf{x}), \dots \Delta_{x_n}f(\mathbf{x}))^T$. Define the (difference) score function for a positive probability mass function, $p(\mathbf{x}) > 0 ~\forall \mathbf{x}$ as:

\begin{align}
    \mathbf{s}_p(\mathbf{x}) &  \coloneqq \frac{\Delta p(\mathbf{x})}{p(\mathbf{x})} \label{discretescorefunction}\\
    (\mathbf{s}_p(\mathbf{x}))_i &= \frac{\Delta_{x_i} p(\mathbf{x})}{p(\mathbf{x})} = 1 -  \frac{p(\neg_i\mathbf{x})}{p(\mathbf{x})}
\end{align}
\end{definition}
Furthermore, ref.\citeS{yang_goodness--fit_2018} also defines the inverse permutation by $\revneg: x^{[i]} \mapsto x^{[(i-1)\mod|\mathcal{X}|]}$, and the \textit{inverse} shift operator by:
\begin{equation}
    \Delta^*_{x_i}f(\mathbf{x})  \coloneqq f(\mathbf{x}) - f(\revneg_i\mathbf{x}) \forall i=1, \dots, n \label{inversediscreteshiftoperator}
\end{equation}
However, this is slightly too general for our purposes, as the sample space for a single qubit is binary; $\mathcal{X} = \{0, 1\}$. In this case, the forward and reverse permutations are identical, so $\Delta = \Delta^*$, as mentioned in the main text.

The discrete versions of the Stein identity, and the kernelised Stein discrepancy are also derived in ref.\citeS{yang_goodness--fit_2018}, and are in an identical form to the continuous case. There is however, one interesting difference. In the continuous case, the class of functions, $\phi \in \mathcal{F}$, which obey Stein's identity (i.e.\@ those which are in the Stein class of the operator $\mathcal{A}_p$) are only those for which $\phi(\mathbf{x})p(\mathbf{x})$ vanishes at the boundary of the set, $\partial \mathcal{Y}$\citeS{liu_kernelized_2016}. Interestingly, this restriction is not necessary in the discrete case. Due to the definition of the discrete shift operator, it turns out that \textit{all} functions $\phi: \mathcal{X}^n \rightarrow \mathbb{R}$ are in the Stein class of the discrete Stein operator. This gives us the freedom to use the quantum kernel, $\kappa_Q$, \eqref{quantumkernel}, in the $\SD$ by making a small adaptation to the proof of Theorem 2 in \citeS{yang_goodness--fit_2018} so that the discrete Stein identity also holds for all \textit{complex} valued functions also. This is necessary due to the complex nature of the feature map in \eqref{quantumfeaturemap}.
\begin{theorem}[Difference Stein's Identity for complex valued functions, adapted from Theorem 2 in ref.\citeS{yang_goodness--fit_2018}]
\label{thm:complexdiscretesteinidentity}
For any function $\phi: \mathcal{X}^n\rightarrow \mathbb{C}$, and a probability mass function $p$ on $\mathcal{X}^n$:
\begin{equation}
    \underset{\mathbf{x}\sim p}{\mathbb{E}}[\mathcal{A}_p \phi(\mathbf{x})]= \mathbb{E}_{\mathbf{x} \sim p}\left[\mathbf{s}_p(\mathbf{x})\phi(\mathbf{x}) - \Delta \phi(\mathbf{x})\right] = 0 
    \label{complexdiscretesteinidentity_supp}
\end{equation}
\end{theorem}
\begin{proof}

Firstly, break the function $\phi$ into real and imaginary parts:

\begin{align}
    \phi(\mathbf{x}) &= a(\mathbf{x}) +\mathrm{i}b(\mathbf{x})\\
     \mathbb{E}[\mathcal{A}_p \phi(\mathbf{x})] &= \sum\limits_{\mathbf{x} \in \mathcal{X}^n}\left[\phi(\mathbf{x})\Delta p(\mathbf{x}) - \Delta^*\phi(\mathbf{x})\right]\\
     &= \sum\limits_{\mathbf{x} \in \mathcal{X}^k}\left[a(\mathbf{x})\Delta p(\mathbf{x}) - p(\mathbf{x})\Delta^*a(\mathbf{x})\right]+ \mathrm{i} \sum\limits_{\mathbf{x} \in \mathcal{X}^n}\left[b(\mathbf{x})\Delta p(\mathbf{x}) - p(\mathbf{x})\Delta^*b(\mathbf{x})\right] \label{finalproofterm}
\end{align}
For each term, $j$, the real parts are given by:

\begin{align}
    \sum\limits_{\mathbf{x} \in \mathcal{X}^n}a(\mathbf{x})\Delta_{x_j} p(\mathbf{x}) = \sum\limits_{\mathbf{x} \in \mathcal{X}^n}a(\mathbf{x})p(\mathbf{x}) - \sum\limits_{\mathbf{x} \in \mathcal{X}^n}a(\mathbf{x}) p(\neg_j\mathbf{x}) \label{proofterm1}\\
    \sum\limits_{\mathbf{x} \in \mathcal{X}^n} p(\mathbf{x})\Delta_{x_j}^*a(\mathbf{x}) = \sum\limits_{\mathbf{x} \in \mathcal{X}^n}p(\mathbf{x})a(\mathbf{x}) - \sum\limits_{\mathbf{x} \in \mathcal{X}^n}p(\mathbf{x}) a(\revneg_j\mathbf{x}) \label{proofterm2}
\end{align}
Since the sum is taken over all the sample space, $\mathcal{X}^n$, and $\neg_i(\revneg_i \mathbf{x}) = \mathbf{x}$, i.e.\@ $\neg$ and $\revneg$ are inverse operations (in fact they are equal in our case), \eqref{proofterm1} is equal to \eqref{proofterm2}, and so the real parts of \eqref{finalproofterm} cancel out. The same result hold for the imaginary parts, completing the proof.

\end{proof}

Furthermore, as in ref.\citeS{yang_goodness--fit_2018}, the functions can also be extended into complex valued \textit{vector} functions, with an analogue of the above result. In this case, for $\Phi: \mathcal{X}^n \rightarrow \mathbb{C}^m$, the discrete Stein Identity is:
\begin{equation}
    \mathbb{E}_{\mathbf{x}}\left[\mathcal{A}_p\Phi(\mathbf{x})\right] = \mathbb{E}_{\mathbf{x}}\left[\mathbf{s}_p(\mathbf{x})\Phi(\mathbf{x})^T - \Delta\Phi(\mathbf{x})\right] = \mathbf{0} \label{vectordiscretesteinidentity}
\end{equation}
where $\Delta\Phi(\mathbf{x})$ is an $n\times m$ matrix: $(\Delta\Phi)_{ij} = \Delta_{x_i}\phi_j(\mathbf{x})$, i.e.\@ shifting the $i^{th}$ element of the $j^{th}$ function value. Now we can reproduce the following Theorem:

\begin{theorem}[Theorem 7 in ref.\citeS{yang_goodness--fit_2018}]\label{thm:yangkerneldiscretestein}
The Discrete Kernelised $\SD$ is given by:
\begin{equation}
    \mathcal{L}_{\SD}(p, q)  \coloneqq D_S(p||q) = \mathbb{E}_{\mathbf{x}, \mathbf{y}\sim p}\left[\kappa_q(\mathbf{x}, \mathbf{y})\right] \label{steindiscrepancycost}
\end{equation}
where $\kappa_q$ is the Stein kernel:
\begin{equation}
    \kappa_q(\mathbf{x}, \mathbf{y})   \coloneqq s_q(\mathbf{x})^T\kappa(\mathbf{x}, \mathbf{y})s_q(\mathbf{y}) -s_q(\mathbf{x})^T\Delta_{\mathbf{y}}^*\kappa(\mathbf{x}, \mathbf{y})  - \Delta_{\mathbf{x}}^*\kappa(\mathbf{x}, \mathbf{y})^Ts_q(\mathbf{y}) + \tr(\Delta_{\mathbf{x}, \mathbf{y}}^*\kappa(\mathbf{x}, \mathbf{y}))
\end{equation}
\end{theorem}
As long as the Gram matrix for the original kernel, $K_{ij} = \kappa(\mathbf{x}^i, \mathbf{y}^j)$ is positive definite, the kernel, $\kappa$, will be \textit{strictly} positive definite. Hence, if a given kernel on a discrete space gives rise to a positive definite Gram matrix, this kernel induces a valid discrepancy measure \citeS{yang_goodness--fit_2018}. Note that this criterion is exactly the same as that which makes the $\MMD$ a valid discrepancy measure in the discrete case\citeS{sriperumbudur_integral_2009}.  

While the $\SD$ is in a similar form to the $\MMD$ (written in terms of expectation values of kernels), there is one key difference. Namely, the score function, $s_q(\mathbf{y})$ in the worst case, requires computing the probabilities of the distribution we are trying to learn, $q(\mathbf{y})$. If we let $p_{\boldsymbol\theta}(\mathbf{x})$ be the output from the $\IBM$, and again, $\pi(\mathbf{x})$ is the data distribution over binary strings, we have the Stein cost function for the $\IBM$ given by:

\begin{align}
    \mathcal{L}_{\SD}(p_{\boldsymbol\theta}, \pi) &  \coloneqq D_S(p_{\boldsymbol\theta}||\pi)  = \mathbb{E}_{\mathbf{x}, \mathbf{y}\sim p_{\boldsymbol\theta}}\left[\kappa_\pi(\mathbf{x}, \mathbf{y})\right]\label{steindiscrepancybornmachine_supp}\\
    \kappa_\pi(\mathbf{x}, \mathbf{y}) &= s_\pi(\mathbf{x})^T\kappa(\mathbf{x}, \mathbf{y})s_\pi(\mathbf{y}) -s_\pi(\mathbf{x})^T\Delta_{\mathbf{y}}^*\kappa(\mathbf{x}, \mathbf{y})  - \Delta_{\mathbf{x}}^*\kappa(\mathbf{x}, \mathbf{y})^Ts_\pi(\mathbf{y}) + \tr(\Delta_{\mathbf{x}, \mathbf{y}}^*\kappa(\mathbf{x}, \mathbf{y})) \label{weightedkernelbornmachine_supp}
\end{align}
To perform gradient descent, we need the gradient of \eqref{weightedkernelbornmachine_supp} with respect to a given parameter, $\theta_k$. The derivation is very similar to the derivation of the $\MMD$ cost function gradient with respect to a parameter $\theta_k$.

\begin{align}
   \frac{\partial \mathcal{L}^{\theta}_{\SD}}{\partial \theta_k} 
    ={}& \sum\limits_{\mathbf{x}, \mathbf{y}} \frac{\partial p_{\boldsymbol\theta}(\mathbf{x})}{\partial \theta_k} \kappa_\pi(\mathbf{x}, \mathbf{y}) p_{\boldsymbol\theta}(\mathbf{y}) +  \sum\limits_{\mathbf{x}, \mathbf{y}} p_{\boldsymbol\theta}(\mathbf{x})\kappa_\pi(\mathbf{x}, \mathbf{y})   \frac{\partial p_{\boldsymbol\theta}(\mathbf{y})}{\partial \theta_k}\\
\begin{split}
    ={}& \sum\limits_{\mathbf{x}, \mathbf{y}}p^-_{\theta_k}(\mathbf{x})  \kappa_\pi(\mathbf{x}, \mathbf{y}) p_{\boldsymbol\theta}(\mathbf{y}) -  \sum\limits_{\mathbf{x}, \mathbf{y}}p^+_{\theta_k}(\mathbf{x})  \kappa_\pi(\mathbf{x}, \mathbf{y}) p_{\boldsymbol\theta}(\mathbf{y}) +\sum\limits_{\mathbf{x}, \mathbf{y}} p_{\boldsymbol\theta}(\mathbf{x})\kappa_\pi(\mathbf{x}, \mathbf{y})   p^-_{\theta_k}(\mathbf{y}) \\
    &  - \sum\limits_{\mathbf{x}, \mathbf{y}} p_{\boldsymbol\theta}(\mathbf{x})\kappa_\pi(\mathbf{x}, \mathbf{y})p^+_{\theta_k}(\mathbf{y})
\end{split}\\
  \frac{\partial \mathcal{L}^{\theta}_{\SD}}{\partial \theta_k} 
    ={}& \underset{\substack{\mathbf{x} \sim p^-_{\theta_k} \\ \mathbf{y}\sim p_{\boldsymbol\theta}}}{\mathbb{E}}[\kappa_\pi(\mathbf{x}, \mathbf{y})] - \underset{\substack{\mathbf{x} \sim p^+_{\theta_k} \\ \mathbf{y}\sim p_{\boldsymbol\theta}}}{\mathbb{E}}[\kappa_\pi(\mathbf{x}, \mathbf{y})] +\underset{\substack{\mathbf{x} \sim p_{\boldsymbol\theta}\\ \mathbf{y}\sim p^-_{\theta_k}}}{\mathbb{E}}[\kappa_\pi(\mathbf{x}, \mathbf{y})] - \underset{\substack{\mathbf{x} \sim p_{\boldsymbol\theta} \\ \mathbf{y}\sim p^+_{\theta_k}}}{\mathbb{E}}[\kappa_\pi(\mathbf{x}, \mathbf{y})] \label{steingradient_supp}
\end{align}
We have used \eqref{mmdprobabilitygradient} for the Ising Born machine gradient, \citeS{mitarai_quantum_2018, liu_differentiable_2018, schuld_evaluating_2019} and that the Stein kernel, $\kappa_{\pi}$, of \eqref{weightedkernelbornmachine} does not depend on the parameter, ${\theta_k}$.

The major difference between \eqref{steingradient_supp}, and the gradient of the $\MMD$, seen in \eqref{mmdgradient}, is the different kernel used, $\kappa_\pi$ vs.\@ simply $\kappa$. This clearly makes the Stein discrepancy more challenging to compute, but it also potentially gives it extra power as a distribution comparison mechanism. This fact is numerically reinforced in \figref{fig:MMDvSinkvStein3}. Also, in contrast to the $\MMD$, the $\SD$ in \eqref{steindiscrepancybornmachine} is asymmetric since the $\pi$ weighted kernel only depends on the distribution $\pi$.

The $\SD$ is computed between the instantaneous model distribution ($p_{\boldsymbol\theta}$) and the distribution to be learned, ($\pi$) by drawing $N = |B|$ samples from the $\IBM$, $\mathbf{x}\sim p_{\boldsymbol\theta}(\mathbf{x})$, and for each pair of samples, $(\mathbf{x}, \mathbf{y})$ the $\pi$-weighted kernel, $\kappa_\pi(\mathbf{x}, \mathbf{y})$ is computed. In doing so for a single pair, we must compute the original kernel, $\kappa$ as a subroutine. If the quantum kernel of \eqref{quantumkernel} is chosen, this requires\citeS{havlicek_supervised_2019} $O\left(\epsilon^{-2}N^4\right)$ measurements of the quantum circuit, as discussed in  \secref{supp_matt:kernelmethods_mmd}. Notice here that we do not need to compute the kernel with respect to the \textit{data samples}, $\mathcal{X}_{\text{Data}}$, since the dependency on $\pi$ in \eqref{weightedkernelbornmachine} is altered with respect to the $\MMD$.

However, adding to the computational burden, we also must compute the following `shifted' versions of the kernel for each pair of samples, $(\mathbf{x}, \mathbf{y})$, in both parameters:
\begin{equation}
    \Delta_\mathbf{x}\kappa(\mathbf{x}, \mathbf{y}) \qquad \Delta_\mathbf{y}\kappa(\mathbf{x}, \mathbf{y}) \qquad  \tr\Delta_{\mathbf{x}\mathbf{y}}\kappa(\mathbf{x}, \mathbf{y})
\end{equation}
which are required in \eqref{weightedkernelbornmachine_supp}. Let $K$ be a polynomial, where it takes $\mathcal{O}(K)$ time to compute the entire original kernel matrix:
\begin{equation}
    K = \mathcal{O}(\epsilon^{-2}N^4) \times T(n)
\end{equation}
where $T(n)$ is a polynomial in the size of the input $n$ (the number of qubits) describing the running time of the quantum circuit required to compute the kernel for \textit{one} pair of samples, $(\mathbf{x}, \mathbf{y})$. We now examine the efficiency of computing the shifted terms in \eqref{steindiscrepancybornmachine_supp} using the quantum kernel of \eqref{quantumkernel_app}.

For example:
\begin{equation}
    \Delta_{\mathbf{x}}\kappa(\mathbf{x}, \mathbf{y}) = (\kappa(\mathbf{x}, \mathbf{y}), \dots, \kappa(\mathbf{x}, \mathbf{y}))^T- (\kappa(\neg_1\mathbf{x}, \mathbf{y}), \dots, \kappa(\neg_n\mathbf{x}, \mathbf{y}))^T \label{xshiftedterm}
\end{equation}
We assume $\kappa(\mathbf{x}, \mathbf{y})$ has been computed for a single pair of samples in time $\mathcal{O}(\epsilon^{-2}N^2 \times T(n))$, and we need to compute $\kappa(\neg_i\mathbf{x}, \mathbf{y})$ for $i = \{1,\dots, n\}$. Therefore, computing the shifted kernel operator in a single parameter takes $\mathcal{O}(\epsilon^{-2}N^2 \times T(n)\times (n+1))$. The same holds for the kernel gradient with respect to the second argument, $\Delta_\mathbf{y}\kappa(\mathbf{x}, \mathbf{y})$.
For $\Delta_{\mathbf{x}, \mathbf{y}}\kappa(\mathbf{x}, \mathbf{y})$, the process is slightly more involved because:

\begin{align}
    \tr\Delta_{\mathbf{x}, \mathbf{y}}\kappa(\mathbf{x}, \mathbf{y}) &= \tr\Delta_{\mathbf{x}}[ \Delta_{\mathbf{y}}\kappa(\mathbf{x}, \mathbf{y})] = \tr\Delta_{\mathbf{x}}[\kappa(\mathbf{x}, \mathbf{y}) - \kappa(\mathbf{x}, \neg\mathbf{y})]  \\
    &= n\kappa(\mathbf{x}, \mathbf{y}) - \sum\limits_{i=1}^n\kappa(\mathbf{x}, \neg_i\mathbf{y}) - \sum\limits_{i=1}^n\kappa(\neg_i\mathbf{x}, \mathbf{y}) + \sum\limits_{i=1}^n\kappa(\neg_i\mathbf{x}, \neg_i\mathbf{y})
\end{align}
Each individual term in the respective sums requires the same complexity, i.e.\@ $\mathcal{O}(\epsilon^{-2}N^2 \times T(n))$ so the term $\tr\Delta_{\mathbf{x}, \mathbf{y}}\kappa(\mathbf{x}, \mathbf{y})$ overall requires $\mathcal{O}(\epsilon^{-2}N^2 \times T(n)\times (3n+1))$. 

Therefore, each term in $\kappa_\pi$ can be computed efficiently using the quantum kernel, \eqref{quantumkernel_app}. That is, with the exception of the score function $s_\pi$ for the data distribution. If we are given oracle access to the probabilities, $\pi(\mathbf{y})$, then there is no issue and $\SD$ will be computable. Unfortunately, in any practical application this will not be the case. To deal with such a case, in \appref{supp_matt:steinscoremethod}, we give two approaches to approximate the Score function via samples from $\pi$. We call these methods the `Identity', and `Spectral' methods for convenience. We only use the Spectral method in training the $\IBM$ in this work, since the former method does not give an immediate out-of-sample method to compute the score, as discussed further in \appref{supp_matt:spectralsteinscore}. Notice that even with the hurdle (difficulty in compute the score), the $\SD$ is still more suitable than the $\KL$ divergence to train these models, since the latter requires computing the \textit{circuit} probabilities, $p_{\boldsymbol\theta}(\mathbf{x})$, which is in general intractable, and so could not be done for \textit{any} dataset. In contrast, the Stein Discrepancy does not require the circuit probabilities, only the data probabilities, which may make it amenable for generative modelling using some datasets.

\subsection{Computing the Stein Score Function} 
\label{supp_matt:steinscoremethod}

Here we address the computability of the score function of \eqref{discretescorefunction}. For every sample, $\mathbf{x}\sim p_{\boldsymbol\theta}$, that we receive from the Born machine we require the score function of that outcome being outputted from the data distribution, $\mathbf{x} \sim \pi$. This involves computing $\pi(\mathbf{x})$, and also $\Delta_{\mathbf{x}}\pi(\mathbf{x})$, i.e.\@ $\pi(\neg_i\mathbf{x}), \forall i \in\{1, \dots, n\}$.

Of course, the Stein discrepancy can be immediately used for any classical problem (classical datasets) for which the probability density is already known, or can be computed efficiently. However, if one is interested in \textit{implicit} models: models which admit a means of sample generation, but not explicit access to the probability density, then this not immediate. In particular, for those distributions which admit some complexity theory result indicating that they cannot be simulated on a classical device efficiently, it will not be possible to efficiently compute the probabilities required in order to compute the score. Implicit models are also present in the classical domain, for example in generative adversarial networks, \citeS{mohamed_learning_2016, diggle_monte_1984}, and it is of great interest to find methods of dealing with them efficiently.

Here we present two approaches\citeS{li_gradient_2018, shi_spectral_2018}, to compute approximations to the score function, and their application in this specific circumstance. In all the following, we assume it is the score of the \text{data}, $\pi$, which we want to compute. Of course, the most obvious approach to computing the score, using $M$ samples alone, would be to simply accumulate the empirical distribution which is observed by the samples, $\hat{\pi}(\mathbf{x}^m) = \frac{1}{M}\sum_{m = 1}^M\mathbb{I}(\mathbf{x} = \mathbf{x}^m)$ and compute the score from this distribution. However, this immediately has a severe drawback. Since the score for a given outcome, $s_{\pi}(\mathbf{x}^m)$, requires \textit{also} computing the probabilities of all shifted samples, $\pi(\neg_i\mathbf{x}^m) ~ \forall i$, if we have not seen any of the outcomes $\neg_i\mathbf{x}^m$ in the observed data, we will not have values for these outcomes in the empirical distribution, and hence we cannot compute the score. This would be a major issue as the number of qubits in our system grows, since we will have exponentially many outcomes, many of which we will not see with $poly(n)$ samples.

\subsubsection{Identity Approximation of Stein Score} \label{app:identitysteinscore}
As a first attempt, we shall try the method of ref.\citeS{li_gradient_2018}. This  involves noticing that the score function appears in Stein identity, and inverting Stein's identity gives a procedure to approximate the score.

Of course, we shall need to use the discrete version of Stein's identity in our case, and rederive the result of ref.\citeS{li_gradient_2018} but there are no major differences. Hence, we refer to this method as the `Identity' method.
If we have generated $M$ samples from the data distribution, we denote the score matrix, $G$, at each of those sample points as follows: 
\begin{equation}
G^{\pi}  \coloneqq \left(\begin{array}{cccc}
    \mathbf{s}^1_\pi(\mathbf{x}^1)  & \mathbf{s}^1_\pi(\mathbf{x}^2)  &\dots &\mathbf{s}^1_\pi(\mathbf{x}^M)   \\
     \mathbf{s}^2_\pi(\mathbf{x}^1)  & \mathbf{s}^2_\pi(\mathbf{x}^2)  &\dots &\mathbf{s}^2_\pi(\mathbf{x}^M)     \\
         \vdots & \vdots &\ddots &\vdots   \\
    \mathbf{s}^n_\pi(\mathbf{x}^1)  & \mathbf{s}^n_\pi(\mathbf{x}^2)  &\dots &\mathbf{s}^n_\pi(\mathbf{x}^M)    \\
\end{array}\right)\label{steinscorematrix}
\qquad G^{\pi}_{i,j} = \mathbf{s}^i_\pi(\mathbf{x}^j) = \frac{\Delta_{x_i^j}\pi(\mathbf{x}^j)}{\pi(\mathbf{x}^j)}
\end{equation}
Each column is the term which corresponds to the score function for the distribution, $\pi$, and that given sample.

Now, to compute, $\hat{G}^{\pi} \approx G^{\pi}$ we can invert the discrete version of Stein's Identity, \eqref{vectordiscretesteinidentityforscoreinAPP}, which is similar to the approach of ref.\citeS{li_gradient_2018} which covers the continuous case:
\begin{equation}
    \underset{\mathbf{x}\sim \pi}{\mathbb{E}}[\mathbf{s}_{\pi}(\mathbf{x})\Phi(\mathbf{x})^T - \Delta \mathbf{f}(\mathbf{x})] = \mathbf{0} \label{vectordiscretesteinidentityforscoreinAPP}
\end{equation}
where $\mathbf{f}$ is a complex vector valued function (for now, we will specify it later).

Rearranging (\ref{vectordiscretesteinidentityforscoreinAPP}) in terms of the score function, and following ref.\citeS{li_gradient_2018}:
\begin{equation}
    \underset{\mathbf{x}\sim \pi}{\mathbb{E}}[\mathbf{s}_\pi\phi(\mathbf{x})^T] = \underset{\mathbf{x}\sim \pi}{\mathbb{E}} [\Delta \mathbf{f}(\mathbf{x})] \\
    \implies \sum\limits_{\mathbf{x}}\pi(\mathbf{x})\mathbf{s}_\pi(\mathbf{x})\mathbf{f}(\mathbf{x})^T = \sum\limits_{\mathbf{x}}\pi(\mathbf{x})\Delta \mathbf{f}(\mathbf{x})
\end{equation}
Approximating with $M$ samples, from the data distribution, $\pi(\mathbf{x})$:
\begin{equation}
    \implies \sum\limits_{i=1}^{M}\mathbf{s}_\pi(\mathbf{x}^i)\mathbf{f}(\mathbf{x}^i)^T \approx \frac{1}{M}\sum\limits_{i=1}^{M}\Delta_{\mathbf{x}^i} \mathbf{f}(\mathbf{x}^i)
\end{equation}
Defining:

\begin{align}
   F &\coloneqq [\mathbf{f}(\mathbf{x}^1), \mathbf{f}(\mathbf{x}^2), \dots, \mathbf{f}(\mathbf{x}^M)]^T,\qquad
   \hat{G}^\pi \coloneqq  [\mathbf{s}^\pi(\mathbf{x}^1), \mathbf{s}^\pi(\mathbf{x}^2), \dots, \mathbf{s}^\pi(\mathbf{x}^M))]^T,\\
    \overline{\Delta_{\mathbf{x}}\mathbf{f}} &= \frac{1}{M}\sum_{i=1}^M \Delta_{\mathbf{x}^i}\mathbf{f}(\mathbf{x}^i), \qquad \Delta_{\mathbf{x}^i}\mathbf{f}(\mathbf{x}^i) \coloneqq   [\Delta_{\mathbf{x}^i}f_1(\mathbf{x}^i),\dots,\Delta_{\mathbf{x}^i}f_l(\mathbf{x}^i)]^T
\end{align}
Now the optimal value for the approximate Stein Matrix, $\hat{G}^\pi$ will be the solution to the following ridge regression problem, and adding a regularisation term, with parameter, $\eta$, to avoid the matrix being non-singular:
\begin{equation}
    \hat{G}^{\pi} = \argmin_{\hat{G}^{\pi} \in \mathbb{R}^{M\times n}}||\overline{\Delta_{\mathbf{x}}\mathbf{f}} - \frac{1}{M}F\hat{G}^{\pi}||_F^2 + \frac{\eta}{M^2}||\hat{G}^{\pi}||_F^2 \label{redregprob_supp}
\end{equation}
Where $||\cdot||_F$ is the Frobenius norm: $||A||_F = \sqrt{\tr\left(A^TA\right)}$.
The analytic solution of this ridge regression problem is well known and can be found by differentiating the above \eqref{redregprob_supp} with respect to $\hat{G}^{\pi}$ and setting to zero:

\begin{align}
    \hat{G}^{\pi} &=  M(K+\eta\mathds{1})^{-1}F^T\overline{\Delta_{\mathbf{x}}\mathbf{f}}\\
    \hat{G}^{\pi} &=  M(K+\eta\mathds{1})^{-1}\langle\Delta, K\rangle \label{approxscore_liderivation}
\end{align}
The method of ref.\citeS{li_gradient_2018} involves implicitly setting the test function to be a feature map in a RKHS, $\mathbf{f} = \Phi$. If this is the case we get $K = F^TF$, and also $\langle \Delta, K\rangle_{ab} = \frac{1}{M}\sum_{i=1}^M \Delta_{x^i_b}\kappa(\mathbf{x}^a, \mathbf{x}^i)$. Unfortunately, there is no motivation given in ref.\citeS{li_gradient_2018} for which choice of feature map should be used to compute \eqref{approxscore_liderivation}. For example, we could use the quantum feature map  of \eqref{quantumkernel_app}, the mixture of Gaussians kernel, \eqref{gaussiankernel}, or the exponentiated Hamming kernel, which is suggested as a sensible kernel to use in (binary) discrete spaces by ref. \citeS{yang_goodness--fit_2018}:

\begin{align}
    \kappa_H(\mathbf{x}, \mathbf{y}) & \coloneqq \exp\left(-H(\mathbf{x}, \mathbf{y})\right)\\
    H(\mathbf{x}, \mathbf{y}) & \coloneqq \frac{1}{n} \sum\limits_{i=1}^n|x_i - y_i| \label{hammingkernel}
\end{align}
Any of these kernels could be used, since the only requirement on the above method is that the feature map obeys the discrete Stein identity, which we have seen is the case for \textit{any} complex vector valued function.

\subsubsection{Spectral Approximation of Stein Score\label{supp_matt:spectralsteinscore}}

While the method used to approximate the score function method which was shown in \appref{app:identitysteinscore} is straightforward, it does not give a method of computing the score accurately at sample points which have \textit{not} been seen in the data distribution, $\pi$. This is a problem if, for instance, we come across a sample from the \IBM, which has not been seen in the data set. Again, this becomes exponentially more likely as the number of qubits grows. If this were to occur during training, a possible solution\citeS{li_gradient_2018} is simply to add that sample to the sample set, and recompute the score function by the Identity method. However, this is expensive, so more streamlined approaches would be desirable. Worse still, this tactic would potentially introduce bias to the data, since there is no guarantee that the given sample from the Born machine, does not have zero probability in the true data, and hence would \textit{never} occur. 

The approach we take is that of ref.\citeS{shi_spectral_2018}, which uses the Nystr{\"o}m method as a subroutine to approximate the score, which is a technique to approximately solve integral equations \citeS{nystrom_uber_1930}. It works by finding eigenfunctions of a given kernel with respect to the target probability mass function, $\pi$. As in the case of ref.\citeS{li_gradient_2018}, the method was defined when $\pi$ is a continuous probability measure, and as such we must make suitable alterations to adapt it to the discrete setting. We refer to this method as the `Spectral' method to compute the score. 

We summarise the parts of ref.\citeS{shi_spectral_2018} which are necessary in the discretisation. For the most part the derivation follows cleanly from ref.\citeS{shi_spectral_2018}, and from \appref{app:identitysteinscore}. Firstly, the eigenfunctions in question are given by the following summation equation:
\begin{equation}
    \sum_{\mathbf{y}}\kappa(\mathbf{x}, \mathbf{y})\psi_j(\mathbf{y})\pi(\mathbf{y}) = \mu \psi_j(\mathbf{x}) \label{nystromsummationequation}
\end{equation}
where $\{\psi_j\}_{j= 1}^N \in \ell^2(\mathcal{X}, \pi)$, and $\ell^2(\mathcal{X}, \pi)$ is the space of all square-summable sequences with respect to $\pi$, over the discrete sample space, $\mathcal{X}$. If the kernel is a quantum one, as in \eqref{quantumkernel_app}, the feature space has a basis, $\{\psi_j = \braket{s_j|\psi}\}_{j= 1}^N \in \ell^2(\mathcal{X}, \pi)$, where $\ket{s_j}$ are for example computational basis states. We also have the constraint that these functions are orthonormal under the discrete $\pi$:
\begin{equation}
    \sum_{\mathbf{x}}\psi_i(\mathbf{x})\psi_j(\mathbf{x}) \pi(\mathbf{x}) = \delta_{ij} \label{nystromorthomal}
\end{equation}
Approximating \eqref{nystromsummationequation} by a Monte-Carlo estimate drawn with $M$ samples, and finding the eigenvalues and eigenvectors of the covariance kernel matrix, $K_{ij} = \kappa(\mathbf{x}^i, \mathbf{y}^j)$, in terms of the approximate ones given by the Monte-Carlo estimate, exactly as in ref.\citeS{shi_spectral_2018}, we get:
\begin{equation}
    \psi_j(\mathbf{x}) \approx \hat{\psi}_j(\mathbf{x}) = \frac{\sqrt{M}}{\lambda_i}\sum\limits_{m = 1}^M u_j(\mathbf{x}^m)\kappa(\mathbf{x}, \mathbf{x}^m) \label{nystromeigenfunctions}
\end{equation}
$\{u_j\}_{j = 1, \dots J}$ are the $J^{th}$ largest eigenvalues of the kernel matrix, $K$, with eigenvalues, $\lambda_j$. The true eigenfunctions are related to these `sampled' versions by: $\psi_j(\mathbf{x}^m) \approx \sqrt{M} u_{jm} \forall m \in\{1, \dots, M\}, \mu_j \approx \lambda_j/M$.

Assuming that the discrete score functions are square summable with respect to $\pi$, i.e.\@ $s^i(\mathbf{x}) \in \ell^2(\mathcal{X}, \pi)$, we can expand the score in terms of the eigenfunctions of the $\ell^2(\mathcal{X}, \pi)$:
\begin{equation}
    s^i(\mathbf{x}) = \sum\limits^{N}_{j=1}\beta_{ij}\psi_j(\mathbf{x}) \label{scoreexpansionineigenbasis}
\end{equation}
Since the eigenfunctions, $\psi_j$ are complex valued, they automatically obey the discrete Stein's identity \eqref{complexdiscretesteinidentity_supp} and we get the same result as ref.\citeS{shi_spectral_2018}:
\begin{equation}
    \beta_{ij} = -\mathbb{E}_{\pi}\Delta_{x_i}\psi_j(\mathbf{x})
\end{equation}
Proceeding\citeS{shi_spectral_2018}, we apply the discrete shift operator, $\Delta_{x_i}$, to both sides of \eqref{nystromsummationequation} to give an approximation for the term, $\hat{\Delta}_{x_i}\psi(\mathbf{x})  \approx \Delta_{x_i}\psi(\mathbf{x})$:
\begin{equation}
    \hat{\Delta}_{x_i}\psi(\mathbf{x}) = \frac{1}{\mu_j M}\sum\limits_{m=1}^M\Delta_{x_i}\kappa(\mathbf{x}, \mathbf{x}^m) \label{approxeigenfunctionshifts}
\end{equation}
It can also be shown in this case that $\hat{\Delta}_{x_i}\psi(\mathbf{x}) \approx \Delta_{x_i}\hat{\psi}(\mathbf{x})$, by comparing \eqref{approxeigenfunctionshifts} with \eqref{nystromeigenfunctions}, and hence we arrive at the estimator for the score function:
\begin{equation}
    \hat{s}^i(\mathbf{x}) = \sum\limits_{j=1}^J \hat{\beta}_{ij}\hat{\psi}_j(\mathbf{x}) \hat{\beta}_{ij} = -\frac{1}{M}\Delta_{x_i}\hat{\Psi}_j(\mathbf{x}^m) \label{spectralestimationfunctions}
\end{equation}
If the sample space is the space of binary strings of length $n$, the number of eigenfunctions, $N$ will be exponentially large, $N = 2^n$, and so the sum in \eqref{scoreexpansionineigenbasis} is truncated to only include the $J^{th}$ largest eigenvalues and corresponding eigenvectors.

\section{Sinkhorn Divergence \& Sample Complexity \label{supp_matt:sinkhorn}}
In this section, we provide further detail about the Sinkhorn divergence used in the main text, and provide proofs of its sample complexity; specifically the mean error and the concentration bounds in \eqref{sinkhorn_expectation_sample_chosen_MAIN} and \eqref{sinkhornborn_samplecomplexity_choosed_MAIN} respectively.

It seems natural to expect that a stronger metric would provide better results than the $\MMD$ for generative modelling. Firstly, we may consider the Kantorovich Metric, \eqref{kantorovichmetric}. The Kantorovich metric happens to be a special case of so-called \textit{optimal transport} ($\OT$) due to the famous  Kantorovich-Rubinstein duality\citeS{dudley_real_2002} which reveals a connection to the Wasserstein metric, as mentioned in the main text. The solution of the `optimal transport' problem gives the optimal way to move, or transport, probability mass from one distribution to another, and hence gives a means of determining the similarity of distributions. 

The optimal transport distance is given by:
\begin{equation}
\OT^c(p, q)  \coloneqq \min\limits_{U \in \mathcal{U}(p, q)}\sum\limits_{(\mathbf{x}, \mathbf{y}) \in \mathcal{X}\times\mathcal{Y}} c(\mathbf{x}, \mathbf{y}) U(\mathbf{x}, \mathbf{y}) \label{otdistance_supp}
\end{equation}
where $p, q$ are the marginal distributions of $U$, i.e.\@\@ $\mathcal{U}(p, q)$ is the space of joint distributions over $\mathcal{X}\times\mathcal{Y}$ such that $\sum_{\mathbf{x}}U(\mathbf{x}, \mathbf{y}) = q(\mathbf{y}), \sum_{\mathbf{y}}U(\mathbf{x}, \mathbf{y}) = p(\mathbf{x})$, in the discrete case. $c(\mathbf{x}, \mathbf{y})$ is the `\textit{cost}' of transporting an individual `point', $\mathbf{x}$, to another point $\mathbf{y}$. It plays a somewhat similar role to the kernel function in the $\MMD$. If we take the optimal transport `cost', to be a metric on the sample space, $\mathcal{X}\times \mathcal{Y}$, i.e.\@ $c(\mathbf{x}, \mathbf{y}) = d(\mathbf{x}, \mathbf{y})$ we get the Wasserstein metric, which turns out to be equivalent to  the Kantorovich metric, revealing the connection to the IPMs discussed in \appref{supp_matt:cf_and_ipms}:
\begin{equation}
W^d(p, q)  \coloneqq \min\limits_{U \in \mathcal{U}(p, q)}\sum\limits_{(\mathbf{x}, \mathbf{y}) \in \mathcal{X}\times \mathcal{Y}} d(\mathbf{x}, \mathbf{y}) U(\mathbf{x}, \mathbf{y}) \label{1wasserstein_supp}
\end{equation}
The Wasserstein metric has become very popular in classical ML literature, for example in the definition of the Wasserstein Generative Adversarial Network\citeS{arjovsky_wasserstein_2017} (WGAN) which utilises it. GANs are alternative generative models, which train \textit{adversarially} via a competition between two neural networks, one \textit{discriminator}, and one \textit{generator}. Quantum models of GANs have also been defined, for example in ref.\citeS{lloyd_quantum_2018, dallaire-demers_quantum_2018}, and have been implemented experimentally \citeS{hu_quantum_2019}. 

However, the Wasserstein metric does suffer from a severe drawback. In $n$ dimensions, its sample complexity scales as $\mathcal{O}(1/M^{1/n})$, while its computational complexity os also high. As a remedy to this, \textit{regularisation} was introduced in order to smooth the problem, and ease computation, which is a standard technique in ML to prevent overfitting.

As noted in ref.\citeS{cuturi_sinkhorn_2013}, one regularised version of optimal transport introduces an entropy term, in the form of the $\KL$ divergence as follows, repeated from \eqref{wassersteinregularised}, with a parameter $\epsilon~>~0$:

\begin{align}
\OT^c_\epsilon(p, q) & \coloneqq  \min\limits_{U \in \mathcal{U}(p, q)}\left(\sum\limits_{(\mathbf{x}, \mathbf{y}) \in \mathcal{X}\times\mathcal{Y}} c(\mathbf{x}, \mathbf{y})U(\mathbf{x}, \mathbf{y}) + \epsilon \KL(U|p \otimes q)\right) \label{wassersteinregularised_supp}\\
\KL(U|p \otimes q) &\equiv \sum\limits_{(\mathbf{x}, \mathbf{y}) \in \mathcal{X}\times\mathcal{Y}} U(\mathbf{x}, \mathbf{y}) \log\left(\frac{U(\mathbf{x}, \mathbf{y})}{(p\otimes q )(\mathbf{x}, \mathbf{y})}\right)
\end{align}
It turns out that the regularised version allows a cost function to be defined, which interpolates between Wasserstein and the $\MMD$. As such, it allows one to take advantage of the small sample complexity, and therefore ease of computability of the $\MMD$, but the desirable properties of the Wasserstein distance.

The relative entropy term serves to determine how distant the coupling distribution, $\pi$, is from a product distribution $p\otimes q$,  and effectively smooths the problem, such that it becomes more efficiently solvable. Based on this, refs.\citeS{genevay_learning_2018, feydy_interpolating_2019} define the \textit{Sinkhorn divergence} ($\SH$) which we can appropriate for the $\IBM$:
\begin{equation}
    \mathcal{L}_{\SH}^\epsilon(p_{\boldsymbol\theta}, \pi)  \coloneqq \OT^c_\epsilon(p_{\boldsymbol\theta}, \pi) - \frac{1}{2} \OT^c_\epsilon(p_{\boldsymbol\theta}, p_{\boldsymbol\theta}) -\frac{1}{2}\OT^c_\epsilon(\pi, \pi) \label{sinkhorndivergence_supp}
\end{equation}
As mentioned in the main text, for extreme values of $\epsilon$, we recover both the $\MMD$ and unregularised optimal transport. The extra terms in \eqref{sinkhorndivergence_supp} relative to \eqref{otdistance_supp} ensure the Sinkhorn divergence is unbiased, since $\OT^c_\epsilon(p, p) \neq 0$ in general.

Now, similarly to the previous two cost functions, we can derive gradients of the Sinkhorn Divergence, with respect to the given parameter, ${\boldsymbol\theta}_k$. According to ref.\citeS{feydy_interpolating_2019}, each term in \eqref{sinkhorndivergence_supp} can be written as follows:
\begin{equation}
    \OT^c_\epsilon(p_{\boldsymbol\theta}, \pi) = \langle p_{\boldsymbol\theta}, f \rangle +  \langle \pi, g \rangle  \\
     =\sum\limits_\mathbf{x} p_{\boldsymbol\theta}(\mathbf{x})f(\mathbf{x}) + \pi(\mathbf{x})g(\mathbf{x})
\end{equation}
$f$ and $g$ are the so-called optimal Sinkhorn potentials, arising from a primal-dual formulation of optimal transport. These are computed using the Sinkhorn algorithm, which gives the divergence its name\citeS{sinkhorn_relationship_1964}. These vectors will be initialised at $f^0(\mathbf{x}) = 0 = g^0(\mathbf{x})$, and iterated in tandem according to \eqref{sinkhorndualvectors_1_supp} and \eqref{sinkhorndualvectors_2_supp} for a fixed number of `Sinkhorn iterations' until convergence. The number of iterations required will depend on the value of $\epsilon$, and the specifics of the problem. Typically, smaller values of epsilon will require more iterations, since this is bringing the problem closer to unregularised optimal transport, which is more challenging to compute. For further discussions on regularised optimal transport and its dual formulation, see refs.\citeS{feydy_interpolating_2019, peyre_computational_2018}. Now, following ref. \citeS{feydy_interpolating_2019}, the $\SH$ can be written on a discrete space as:
\begin{equation}
    \mathcal{L}_{\SH}^\epsilon(p_{\boldsymbol\theta}, \pi) = \sum\limits_\mathbf{x} \left[p_{\boldsymbol\theta}(\mathbf{x})\left(f(\mathbf{x}) - s(\mathbf{x})\right)
    + \pi(\mathbf{x})\left(g(\mathbf{x}) - t(\mathbf{x})\right)\right] \label{sinkhorndualformulation_supp}
\end{equation}
$s, t$ are the `\textit{autocorrelation}' dual potentials, arising from the terms $\OT^c_\epsilon(p_{\boldsymbol\theta}, p_{\boldsymbol\theta})$, $\OT^c_\epsilon(\pi, \pi)$ in \eqref{sinkhorndivergence_supp}.

Following ref.\citeS{feydy_interpolating_2019}, we can see how by discretising the situation based on $N, M$ samples from $p_{\boldsymbol\theta}, \pi$ respectively; $ \hat{\mathbf{x}} = \{\mathbf{x}^1, \dots, \mathbf{x}^N\}\sim p_{\boldsymbol\theta}(\mathbf{x}),  \hat{\mathbf{y}} =\{\mathbf{y}^1, \dots, \mathbf{y}^M\} \sim \pi(\mathbf{y})$. With this, the optimal dual vectors, $f, g$ are given by:
\begin{align}
    f^{l+1}(\mathbf{x}^i) &= -\epsilon \text{LSE}_{k=1}^M\left(\log\left(\pi(\mathbf{\mathbf{y}}^k) + \frac{1}{\epsilon}g^{l}(\mathbf{y}^k) - \frac{1}{\epsilon} C_{ik}(\mathbf{x}^i, \mathbf{y}^k)\right)\right)\label{sinkhorndualvectors_1_supp}\\
    g^{l+1}(\mathbf{y}^j) &= -\epsilon \text{LSE}_{k=1}^N\left(\log\left(p_{\boldsymbol\theta}(\mathbf{x}^k) + \frac{1}{\epsilon}f^{l}(\mathbf{x}^k) - \frac{1}{\epsilon} C_{kj}(\mathbf{x}^k, \mathbf{y}^j)\right)\right) \label{sinkhorndualvectors_2_supp}
\end{align}
$C(\mathbf{x}, \mathbf{y})$ is the so-called optimal transport \textit{cost matrix} derived from the cost function applied to all samples, $C_{ij}(\mathbf{x}^i, \mathbf{y}^j) = c(\mathbf{x}^i, \mathbf{y}^j)$ and $\text{LSE}_{k=1}^N(\mathbf{V}_k) = \log\sum\limits_{k=1}^N\exp(\mathbf{V}_k)$ is a log-sum-exp reduction for a vector $\mathbf{V}$, used to give a smooth approximation to the true dual potentials.

The autocorrelation potential, $s$, is given by:
\begin{equation}
    s(\mathbf{x}^i) = -\epsilon \text{LSE}_{k=1}^N\left(\log\left(p_{\boldsymbol\theta}(\mathbf{\mathbf{x}}^k) + \frac{1}{\epsilon}s(\mathbf{x}^k) - \frac{1}{\epsilon} C(\mathbf{x}^i, \mathbf{x}^k)\right)\right) \label{autocorrelationsinkhornterms_supp}
\end{equation}
$t(\mathbf{y}^i)$ can be derived similarly by replacing $p_{\boldsymbol\theta} \rightarrow \pi$ in \eqref{autocorrelationsinkhornterms_supp} above. However, the autocorrelation dual can be found using a well-conditioned fixed point update \citeS{feydy_interpolating_2019}, and convergence to the optimal potentials can be observed with much fewer Sinkhorn iterations:
\begin{equation}
    s(\mathbf{x}^i) \leftarrow \frac{1}{2}\left[s(\mathbf{x}^i)-\epsilon \text{LSE}_{k=1}^N\left(\log\left(p_{\boldsymbol\theta}(\mathbf{x}^k) + \frac{1}{\epsilon}s(\mathbf{x}^k) - \frac{1}{\epsilon} C(\mathbf{x}^i, \mathbf{x}^k)\right)\right)\right] \label{autocorrelationsinkhorntermsupdate_supp}
\end{equation}

Next, we derive the gradient of $\mathcal{L}^\epsilon_{\SH}$, given in \eqref{sinkhorngradient} in the main text. The derivative with respect to a single probability of the \textit{observed} samples, $p_{\boldsymbol\theta}(\mathbf{x}^i)$, is given by \citeS{feydy_interpolating_2019}:
\begin{equation}
    \frac{\partial \mathcal{L}_{\SH}^\epsilon(p_{\boldsymbol\theta}, \pi)}{\partial p_{\boldsymbol\theta}(\mathbf{x}^i)} = f(\mathbf{x}^i) - s(\mathbf{x}^i)
\end{equation}
However, this only applies to the samples which have been \textit{used} to compute $f, s$ in the first place. If one encounters a sample from $p_{{\boldsymbol\theta}_k^{\pm}}$ (which we shall in the gradient), $\mathbf{x}^s$, which one has not seen in the original samples from $p_{\boldsymbol\theta}$, one has no value for the corresponding vectors at this point: $f(\mathbf{x}^s), s(\mathbf{x}^s)$. Fortunately, as shown in ref.\citeS{feydy_interpolating_2019}, the gradient does extend smoothly to this point (and all points in the sample space) and in general is given by:
\begin{equation}
    \frac{\partial \mathcal{L}_{\SH}^\epsilon(p_{\boldsymbol\theta}, \pi)}{\partial p_{\boldsymbol\theta}(\mathbf{x})} = \varphi(\mathbf{x}) 
\end{equation}
\begin{multline}
    \varphi(\mathbf{x}) = -\epsilon \text{LSE}_{k=1}^M\left(\log\left(\pi(\mathbf{\mathbf{y}}^k) + \frac{1}{\epsilon}g^{0}(\mathbf{y}^k) - \frac{1}{\epsilon} C(\mathbf{x}, \mathbf{y}^k)\right)\right)\\
    + \epsilon \text{LSE}_{k=1}^N\left(\log\left(p_{\boldsymbol\theta}(\mathbf{x}^k) + \frac{1}{\epsilon}s^{0}(\mathbf{x}^k) - \frac{1}{\epsilon} C(\mathbf{x}, \mathbf{x}^k)\right)\right)
\end{multline}
where $g^{(0)}, s^{(0)}$, are the optimal vectors which solve the original optimal transport problem, \eqref{sinkhorndualvectors_2_supp} and \eqref{autocorrelationsinkhornterms_supp} at convergence, given the samples, $\hat{\mathbf{x}}, \hat{\mathbf{y}}$ from $p_{\boldsymbol\theta}, \pi$ respectively.
Given this, the gradient of the Sinkhorn divergence with respect to the parameters, ${\boldsymbol\theta}_k$ is given by:

\begin{align}
    \frac{\partial \mathcal{L}_{\SH}^\epsilon(p_{\boldsymbol\theta}, \pi)}{\partial {\boldsymbol\theta}_k} &= \sum\limits_{\mathbf{x}}\frac{\partial \mathcal{L}_{\SH}^\epsilon(p_{\boldsymbol\theta}, \pi)}{\partial p_{\boldsymbol\theta}(\mathbf{x})}\frac{\partial p_{\boldsymbol\theta}(\mathbf{x})}{\partial {\boldsymbol\theta}_k} \\
    & = \sum\limits_\mathbf{x}\varphi(\mathbf{x})\left(p_{{\boldsymbol\theta}^-_k}(\mathbf{x}) - p_{{\boldsymbol\theta}^+_k}(\mathbf{x})\right)  \\
    & = \underset{\substack{\mathbf{x} \sim p_{{\boldsymbol\theta}^-} }}{\mathbb{E}}[\varphi(\mathbf{x})] 
    -\underset{\substack{\mathbf{x} \sim p_{{\boldsymbol\theta}_k^+}  }}{\mathbb{E}}[\varphi(\mathbf{x})] \label{sinkhorngradient_supp}
\end{align}
 Therefore, one can compute the gradient by drawing samples from the distributions, $\hat{\mathbf{x}} \sim p_{{\boldsymbol\theta}^\pm}$, and computing the vector $\varphi(\mathbf{x})$ , for each sample, $\mathbf{x} \in \hat{\mathbf{x}}$, using the vectors, $g^{(0)}, s^{(0)}$ already computed during the evaluation of $\SH$ at each epoch.

\subsection{Sample Complexity} \label{supp_matt:sinkhorn_sample_complexity}

As mentioned above, the sample complexity of the $\MMD$ scales as $\mathcal{O}(1/\sqrt{M})$, and it is known that the Wasserstein distance scales as $\mathcal{O}(1/M^n)$ \citeS{weed_sharp_2017} for a distribution supported on a subset of $\mathbb{R}^n$. The former indicates that the $\MMD$ can be computed to an accuracy $\epsilon$ with $\mathcal{O}(\epsilon^{-2})$ samples regardless of the underlying space, and the latter means the Wasserstein distance requires $\mathcal{O}(\epsilon^{-n})$ which is exponential in the size of the space. Since we are dealing with binary vectors of length $n$, the support of our distributions (the $\IBM$ and the data distribution) will be $\mathcal{X} = \{0, 1\}^n \subseteq \mathbb{R}^n$, but a general quantum distribution can be supported on this entire space. These two cases correspond to the extreme regularisation values of $\epsilon \rightarrow \infty$ and $\epsilon \rightarrow 0$ respectively. However, the motivation for using the Sinkhorn Divergence is to leverage the favourable sample complexity of the $\MMD$, with the stronger Wasserstein Distance. This $\epsilon$-regularisation hyperparameter allows us to choose a cost function which optimally suits our needs. By the results of ref.\citeS{feydy_interpolating_2019}, we can be assured that no matter which value is chosen, the $\SH$ will be suitable as a cost function to train the $\IBM$, and in fact any generative model. This is because it is zero for any distributions, say $p$ and $q$, which are identical, and strictly positive otherwise:

\begin{align}
    \mathcal{L}_{\SH}^\epsilon(p, q)  = 0 \iff p = q \label{sinkhornzeroiff}\\
    \mathcal{L}_{\SH}^\epsilon(p, q) \geq \mathcal{L}_{\SH}^\epsilon(p, p) = 0 \label{sinkhorngreaterzero_supp}
\end{align}

It also metrizes convergence in law, effectively meaning it can be estimated using $M$ samples, and will converge to its true value in the limit of large samples:
\begin{equation}
    \mathcal{L}_{\SH}^\epsilon(\hat{p}_M, p) \rightarrow 0 \iff \hat{p}_M \rightharpoonup p \label{sinkhornconvergenceinlaw_supp}
 \end{equation}
Now, from ref.\citeS{genevay_sample_2018}, we have the following two results. Note the original theorems technically apply to the regularised $\OT$ cost, rather than $\mathcal{L}_{\SH}$, but the addition of the symmetric terms in \eqref{sinkhorndivergence} will not affect the asymptotic sample complexities since they add only constant overheads. The first is the mean difference between the true Sinkhorn divergence, $\mathcal{L}_{\SH}^\epsilon(P, Q)$ and its estimator derived from the empirical distributions, $\mathcal{L}_{\SH}^\epsilon(\hat{P}_M, \hat{Q}_M)$ is given by:
\begin{theorem}[Theorem 3 from ref.\citeS{genevay_sample_2018}]\label{thm:sinkhornexpectationsamplecomplexity_supp}
Consider the Sinkhorn divergence between two distributions, $p$, and $q$ on two bounded subsets, $\mathcal{X}, \mathcal{Y}$ of $\mathbb{R}^n$, with a $C^{\infty}$, $L-$Lipshitz cost $c$. One has:
\begin{equation}
    \mathbb{E}|\mathcal{L}_{\SH}^\epsilon(p,q) - \mathcal{L}_{\SH}^\epsilon(\hat{p}_M, \hat{q}_M)| = \mathcal{O}\left(\frac{\mathrm{e}^{\frac{\kappa}{\epsilon}}}{\sqrt{M}}\left(1+\frac{1}{\epsilon^{\lfloor n/2\rfloor}}\right)\right)\label{sinkhornexpectationsamplecomplexity_supp}
\end{equation}
where $\kappa = 2L|\mathcal{X}|+||c||_{\infty}$ and constants only depend on $|\mathcal{X}|, |\mathcal{Y}|, c$ and $||c^{l}||_{\infty}$ for $l = 0, \dots, \lfloor n/2\rfloor$.
\end{theorem}
The second is the following concentration result:

\begin{corollary}
[Corollary 1 from ref.\citeS{genevay_sample_2018}]\label{thm:sinkhornconcentrationsamplecomplexity_supp}
With probability at least $1-\delta$,
\begin{equation}
    |\mathcal{L}_{\SH}^\epsilon(p,q) - \mathcal{L}_{\SH}^\epsilon(\hat{p}_M, \hat{q}_M)| \leq 12B\frac{\lambda K}{\sqrt{M}} + 2C\sqrt{\frac{2\log\frac{1}{\delta}}{M}}\label{sinkhornconcentrationsamplecomplexity_supp}
\end{equation}
where $\kappa  \coloneqq 2L|\mathcal{X}|+||c||_{\infty}$, $C  \coloneqq \kappa + \epsilon \mathrm{e}^{\frac{\kappa}{\epsilon}}$, $B\leq 1+\mathrm{e}^{\left(2\frac{L|\mathcal{X}|+||c||_{\infty}}{\epsilon}\right)}, \lambda = \mathcal{O}(1+ \frac{1}{\epsilon^{\lfloor n/2\rfloor}})) \text{ and } K  \coloneqq \max_{\mathbf{x} \in \mathcal{X}}\kappa_S(\mathbf{x}, \mathbf{x})$.
\end{corollary}
$\kappa_S$ is the Matern or the Sobolev kernel, associated to the Sobolev space, $\mathcal{H}^s(\mathbb{R}^n)$, which is a RKHS for $s > n/2$, but we will not go into further detail here.

The more exact expression for \eqref{sinkhornexpectationsamplecomplexity_supp} is given by:
\begin{equation}
     \mathbb{E}|\mathcal{L}_{\SH}^\epsilon - \hat{\mathcal{L}}_{\SH}^\epsilon| \leq 12\frac{B\lambda\sqrt{K}}{\sqrt{M}}
     = \mathcal{O}\left(\frac{1}{\sqrt{M}}\left(1+\mathrm{e}^{\left(2\frac{L|\mathcal{X}|+||c||_{\infty}}{\epsilon}\right)}\right)\left(1+\frac{1}{\epsilon^{\lfloor n/2\rfloor}}\right)\right)
\end{equation}
where we use $\mathcal{L}_{\SH}^\epsilon$, rather than $\OT_\epsilon^c$ as in ref.\cite{genevay_learning_2018}, through the use of the triangle inequality.

Now, for our particular case, we wish to choose the Hamming distance as a metric on the Hamming hypercube (for a fixed dimension, $n$). However, due to the smoothness requirement of the above theorems, $c \in C^\infty$, this would not hold in the discrete case we are dealing with. However, we can take a broader view to simply use the $\ell_1$ distance, and embed the Hamming hypercube in a larger space. This is possible because the Hamming distance, $d_{H}$ is exactly the $\ell_1$ metric but restricted to binary vectors:
\begin{equation}
    ||\mathbf{x} - \mathbf{y}||_1  = \sum\limits_{i=1}^n|x_i-y_i| = d_H(\mathbf{x}, \mathbf{y}) ~ \forall x_i, y_i \in \{0, 1\} \label{samplecomplexityexpanded_supp}
\end{equation}
In this scenario, formally, we are dealing with the general hypercube in $\mathbb{R}^n$, but where the probability masses are strictly concentrated on the vertices of the hypercube. Now, we can compute directly some of the constants in the above, Theorem \ref{thm:sinkhornexpectationsamplecomplexity_supp} and Corollary \ref{thm:sinkhornconcentrationsamplecomplexity_supp}. Taking $\mathcal{X}$ to be the unit hypercube in $\mathbb{R}^n$, and taking the Sinkhorn cost to be the $\ell_1$ cost, which is Lipschitz continuous, we can compute the following:
\begin{equation}
    |\mathcal{X}| = \sup_{\mathbf{x}, \mathbf{y} \in \mathcal{X}}||\mathbf{x} - \mathbf{y}||_1= n,\qquad
    ||c(\mathbf{x}, \mathbf{y})||_\infty = \sup\{|c(\mathbf{x}, \mathbf{y})| : (\mathbf{x}, \mathbf{y}) \in \mathcal{X}\times \mathcal{Y}\} = n\label{sinkhorn_sample_complexity_calcs}
\end{equation}
A rough upper bound for the Lipschitz constant, $L$, can be obtained as follows. For a function, $f:\mathcal{A}\rightarrow\mathcal{B}$, the Lipschitz constant is the smallest value $L$ such that:
\begin{equation}
    d_{\mathcal{B}}(f(\mathbf{x}),f(\mathbf{y})) \leq L d_{\mathcal{A}}(\mathbf{x}, \mathbf{y}) \label{lipschitzcontinuous_supp}
\end{equation}
If we take $d_{\mathcal{A}}$ to be the sum metric on the product space, $\mathcal{X}\times\mathcal{Y}$, and $f$ to be the $\ell_1$ distance, we get:
\begin{equation}
    |(c(\mathbf{x}^1, \mathbf{y}^1) - c(\mathbf{x}^2, \mathbf{y}^2))| \leq L \left[d_{\mathcal{X}}(\mathbf{x}^1, \mathbf{x}^2)+ d_{\mathcal{Y}}(\mathbf{y}^1, \mathbf{y}^2)\right] \label{costfunctionlipschitzcont_supp}
\end{equation}
For two points, $(\mathbf{x}^1, \mathbf{y}^1), (\mathbf{x}^2, \mathbf{y}^2) \in \mathcal{X}\times \mathcal{Y}$, and the cost $c:\mathcal{X}\times\mathcal{Y}\rightarrow \mathbb{R}, c = ||\cdot||_1$.
Now, 
\begin{equation}
    \frac{\left|\sum_i |\mathbf{x}_i^1 - \mathbf{y}_i^1| - \sum_i |\mathbf{x}_i^2 - \mathbf{y}_i^2| \right|}{\left[\sum_i |\mathbf{x}_i^1 - \mathbf{x}_i^2| + \sum_i |\mathbf{y}_i^1 - \mathbf{y}_i^2| \right]} \leq L \label{lipshitzderivation2_supp}
\end{equation}
We want to find an upper bound for the left hand side of \eqref{lipshitzderivation2_supp}, assuming that $\mathbf{x}^1 \neq \mathbf{x}^2$ and $\mathbf{y}^1 \neq \mathbf{y}^2$. In this case, both numerator and denominator are zero, so any non-zero $L$ will satisfy \eqref{costfunctionlipschitzcont_supp}. Now, applying the trick of embedding the Hamming hypercube in the general hypercube, we can assume $\mathbf{x}^{1,2}_i, \mathbf{y}^{1,2}_i \in \{0, 1\}\forall i$. To derive such a bound, we can bound both the numerator and the denominator of the LHS of \eqref{lipshitzderivation2_supp} independently. We find the denominator is as small as possible, when only one element of $\mathbf{x}^{1,2}$ or $\mathbf{y}^{1,2}$ is equal to one, and all the are equal to zero.  The numerator is as large as possible when one of $\mathbf{x}^{1,2}$ or $\mathbf{y}^{1,2}$ is the all-one vector. In this case, the LHS is upper bounded by $n$ (for a fixed $n$), so we can choose $L = n$, which is a constant for a fixed $n$.

So, \eqref{samplecomplexityexpanded_supp} becomes:
\begin{equation}
     \mathbb{E}|\mathcal{L}_{\SH}^\epsilon - \hat{\mathcal{L}}_{\SH}^\epsilon| = \mathcal{O}\left(\frac{1}{\sqrt{M}}\left(1+\mathrm{e}^{\left(2\frac{n^2+n}{\epsilon}\right)}\right)\left(1+\frac{1}{\epsilon^{\lfloor n/2\rfloor}}\right)\right)
\end{equation}
The constants in $\mathcal{O}\left(1+\frac{1}{\epsilon^{\lfloor n/2\rfloor}}\right)$, depend on $|\mathcal{X}|, |\mathcal{Y}|, n, \text{ and } ||c^{(k)}||_\infty$ which are at most linear in $n$. Ignoring the constant terms, we arrive at the scaling for the average error, \eqref{sinkhorn_expectation_sample_chosen_MAIN}, in the main text. Similarly the concentration bound is: 
\begin{align}
    &|\mathcal{L}_{\SH}^\epsilon - \hat{\mathcal{L}}_{\SH}^\epsilon| \leq 12B\frac{\lambda K}{\sqrt{M}} + 2C\sqrt{\frac{2\log\frac{1}{\delta}}{M}}\\
    &\leq \frac{12K}{\sqrt{M}}\left(1+\mathrm{e}^{\left(2\frac{L|\mathcal{X}|+||c||_{\infty}}{\epsilon}\right)}\right)\mathcal{O}\left(1+\frac{1}{\epsilon^{\lfloor n/2\rfloor}}\right)+ 2\kappa\sqrt{\frac{2\log\frac{1}{\delta}}{M}} + 2\epsilon \mathrm{e}^{\frac{\kappa}{\epsilon}}\sqrt{\frac{2\log\frac{1}{\delta}}{M}}\\
    &= \frac{1}{\sqrt{M}}\left[12\left(1+\mathrm{e}^{\left(\frac{\mathcal{O}(n^{2})}{\epsilon}\right)}\right)\mathcal{O}\left(1+\frac{1}{\epsilon^{\lfloor n/2\rfloor}}\right)+\mathcal{O}(n^{2})\sqrt{2\log\frac{1}{\delta}} + 2\epsilon \mathrm{e}^{\frac{\mathcal{O}(n^{2})}{\epsilon}}\sqrt{2\log\frac{1}{\delta}}\right]\\
    &= \mathcal{O}\left(\frac{1}{\sqrt{M}}\left[\left(1+\mathrm{e}^{\left(\frac{n^{2}}{\epsilon}\right)}\right)\left(1+\frac{1}{\epsilon^{\lfloor n/2\rfloor}}\right)\right.+ \left.n^{2}\sqrt{2\log\frac{1}{\delta}}  + \epsilon \mathrm{e}^{\frac{n^{2}}{\epsilon}}\sqrt{2\log\frac{1}{\delta}}\right]\right) \label{sinkhornbornsamplecomplexity_supp}
\end{align}
since $\kappa = 2n^2+n = \mathcal{O}(n^{2})$. 

Clearly, due to the asymptotic behaviour of the Sinkhorn divergence, we would like to choose $\epsilon$ sufficiently large in order to remove as much dependence on the dimension, $n$, as possible. This is because, in our case, the dimension of the space is equivalent to the number of qubits, and hence to derive a favourable sample complexity, we would hope for the dependence on $n$ to be polynomial in the number of qubits. By examining, \eqref{sinkhornbornsamplecomplexity_supp}, we could see that a good choice might be $\epsilon = \mathcal{O}(n^{2})$. In this case, we get:
\begin{equation}
    |\mathcal{L}_{\SH}^{\mathcal{O}(n^{2})} - \hat{\mathcal{L}}_{\SH}^{\mathcal{O}(n^{2})}| 
    = \mathcal{O}\left(\frac{1}{\sqrt{M}}\left[\left(1+\frac{1}{n^{\lfloor n/2\rfloor}}\right)\right.\right.+ \left.\left.n\sqrt{2\log\frac{1}{\delta}}\right]\right) \label{sinkhornbornsamplecomplexitychoosed_supp}
\end{equation}
with probability $1-\delta$. 

It is likely in practice however, that a much smaller value of $\epsilon$ could be chosen, without blowing up the sample complexity. This is evidenced by numerical results in \citeS{genevay_learning_2018, feydy_interpolating_2019, genevay_sample_2018}. Again, by ignoring constant terms, we get the scaling observed in \eqref{sinkhornborn_samplecomplexity_choosed_MAIN} in the main text.

\section{Extra Numerical Results \label{supp_matt:numericalresults}}

In this section, we present extra numerical results to supplement those in the main text. 

Firstly, we will elaborate on our goals. We aim for two properties that our cost functions should exhibit in order to claim they have outperformed the $\MMD$-with-classical-kernel training method:
\begin{itemize}
    \item \textit{Speed of Converge:} The $\MMD$-with-quantum-kernel, and both of the cost functions, $\SD$ and $\SH$, should achieve equal or lower $\TV$ than the $\MMD$ in a \textit{shorter} time period (even accounting for various learning rates).
    \item \textit{Accuracy:} Since the cost functions we employ are in some sense `stronger' than $\MMD$, we would like for them to achieve a \textit{smaller} $\TV$ than is possible with the $\MMD$ in an equal or quicker time. 
\end{itemize}

It should be noted that it would not be possible to compute $\TV$ in general as the number of qubits scales, but it is possible to do so to benchmark the small examples we use here. 

Firstly, we revisit the quantum (\eqref{quantumkernel}) vs.\@ Gaussian  (\eqref{gaussiankernel}) kernels of the main text with the corresponding example for two qubits in \suppfigref{fig:QvGkernel2}. These results corroborate those seen in the main text for four qubits; that $\kappa_Q$ does not provide any obvious advantage in training the Born machine. In fact it actually performs slightly worse than the Guassian kernel on average. This is indicated in \suppfigref{fig:QvGkernel2}(a) and \figref{fig:QvGkernel4}(a). As such, this calls into question the usefulness of quantum kernels first explored in refs. \citeS{havlicek_supervised_2019, schuld_quantum_2019}. This reinforces the key point that classical hardness does not translate directly into learning advantage, and this topic should be explored further.

Again, the data was taken during 5 independent runs in each case, with averages and errors (maximum and minimum values achieved during training) computed over the runs. This behaviour is apparent even for a range of learning rates.

\begin{figure}
    \centering
    \includegraphics[width=0.9\columnwidth]{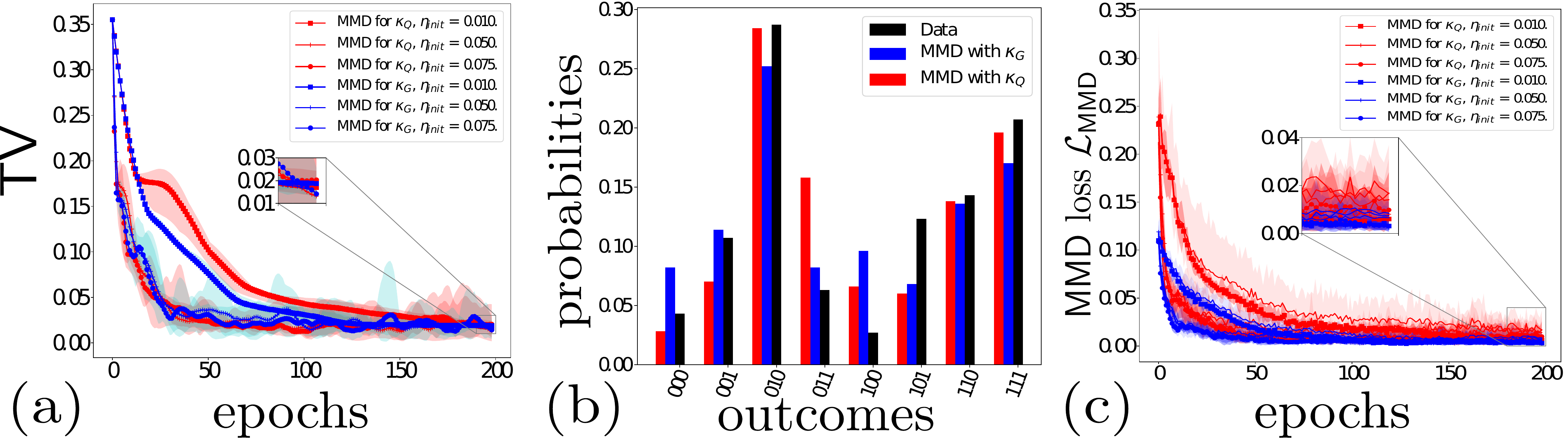}
    \caption{Performance of quantum $\kappa_Q$ [\crule[red]{0.2cm}{0.2cm}] vs. Gaussian kernel,  $\kappa_G$ [\crule[blue]{0.2cm}{0.2cm}]  for 3 qubits. To train, we sample from the $\IBM$ and the data $500$ times and use a minibatch size of $250$. One epoch is one complete update of all parameters according to gradient descent. Error bars represent maximum, minimum and mean values achieved over 5 independent training runs, with the same initial conditions on the same data samples. (a) $\TV$ Difference achieved with both kernel methods during training. Quantum kernel provides no obvious advantage over the classical kernel, as in the 4 qubit case. (b) Final learned probabilities with $\eta_{\mathsf{init}} = 0.1$ using the Adam optimiser. (c) $\MMD$ computed using $400$ samples as training points, $100$ as test points \red{(training on this set observed as thin lines without markers)}, independent of the training data. }
    \label{fig:QvGkernel2}
\end{figure}
\suppfigref{fig:MMDvSinkvStein4} illustrates the differences between using the $\MMD$ cost function and either the Sinkhorn divergence or the Stein discrepancy for four qubits, and supplements \figref{fig:MMDvSinkvStein3} in the main text. We found that both the Stein discrepancy and the Sinkhorn divergence are able to learn with a higher \textit{accuracy} as shown in \suppfigref{fig:MMDvSinkvStein4}(a), reinforcing \figref{fig:MMDvSinkvStein3}(a) in the main text. It should be noted that the results for the Sinkhorn training were highly dependent on the value of the regularisation, as expected. Also, simply because the Sinkhorn achieved better results on one particular data set, does not imply that it would do so over most. However, it does support our theoretical reasoning that the Sinkhorn divergence is better at minimising $\TV$.

For the case of the `Spectral' score method, (\figref{fig:MMDvSinkvStein3}(a) and \suppfigref{fig:MMDvSinkvStein4}(a)) we used many fewer samples than with the previous methods. This is due to the (classical) computational cost of computing the score using this method in our naive implementation. For computing the Spectral score, we used 4 and 6 Nystr{\"o}m eigenvectors for 3 and 4 qubits respectively, for details see \appref{supp_matt:steinscoremethod}. However, the spectral Stein method seemed to perform worse than other methods, and did not agree with the exact Stein discrepancy. As such, further investigation of this point is needed, including perhaps better methods to approximate the score.

One may comment on the fact that we allowed the Stein discrepancy to use the exact probabilities of the data, $\pi$, and this constitutes an unfair advantage against the $\MMD$. In fact, the high number of samples we used ensured that the approximate data distributions that the $\MMD$ received was very close to the exact data, and we found no major improvement for training with the $\MMD$ by allowing oracle data access, i.e.\@ the exact probabilities, $\pi$, reinforcing the fundamental weakness of the $\MMD$ that is discussed above.

\begin{figure}
    \centering
    \includegraphics[width=\columnwidth, height=0.3\columnwidth]{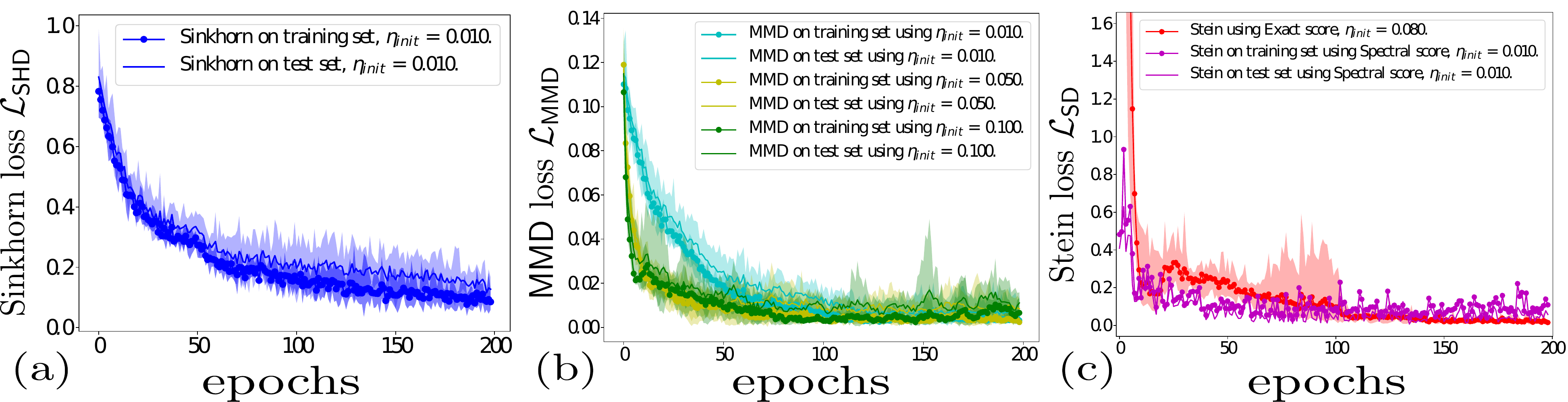}
    \caption{$\MMD$ [\crule[cyan]{0.2cm}{0.2cm}, \crule[yellow]{0.2cm}{0.2cm}, \crule[ForestGreen]{0.2cm}{0.2cm}] vs.\@ Sinkhorn [\crule[blue]{0.2cm}{0.2cm}] and Stein training with Exact score function [\crule[red]{0.2cm}{0.2cm}] and Spectral score method [\crule[magenta]{0.2cm}{0.2cm}] for 3 qubits with fully connected topology, Rigetti {\fontfamily{cmtt}\selectfont 3q-qvm}, \protect\threeqqvm \ \ \ , trained on the data, \eqref{toydatadistribution}. 500 data points are used for training, with 400 used as a training set, and 100 used as a test set for all except the Stein discrepancy using the spectral score, which used only 40 samples. Plots show mean, maximum and minimum values achieved over 5 independent training runs on the same dataset. (a) $\mathcal{L}^{0.08}_{\SH}$   using 500 samples and a batch size of 250. (b) $\mathcal{L}_{\MMD}$ using 500 samples and a batch size of 250, with three different initial learning rates. (c) $\mathcal{L}_{\SD}$ using 500 samples and a batch size of 250 for Exact score, and 40 samples and batch size of 20 for the Spectral score. \red{Training on test set observed as thin lines without markers for all methods, excluding the Exact score method, since this uses exact probabilities.} Corresponds to \figref{fig:MMDvSinkvStein3} in main text. }
    \label{fig:MMDvSinkvStein3_supp}
\end{figure}
\begin{figure}
    \centering
    \includegraphics[width=\columnwidth, height=0.5\columnwidth]{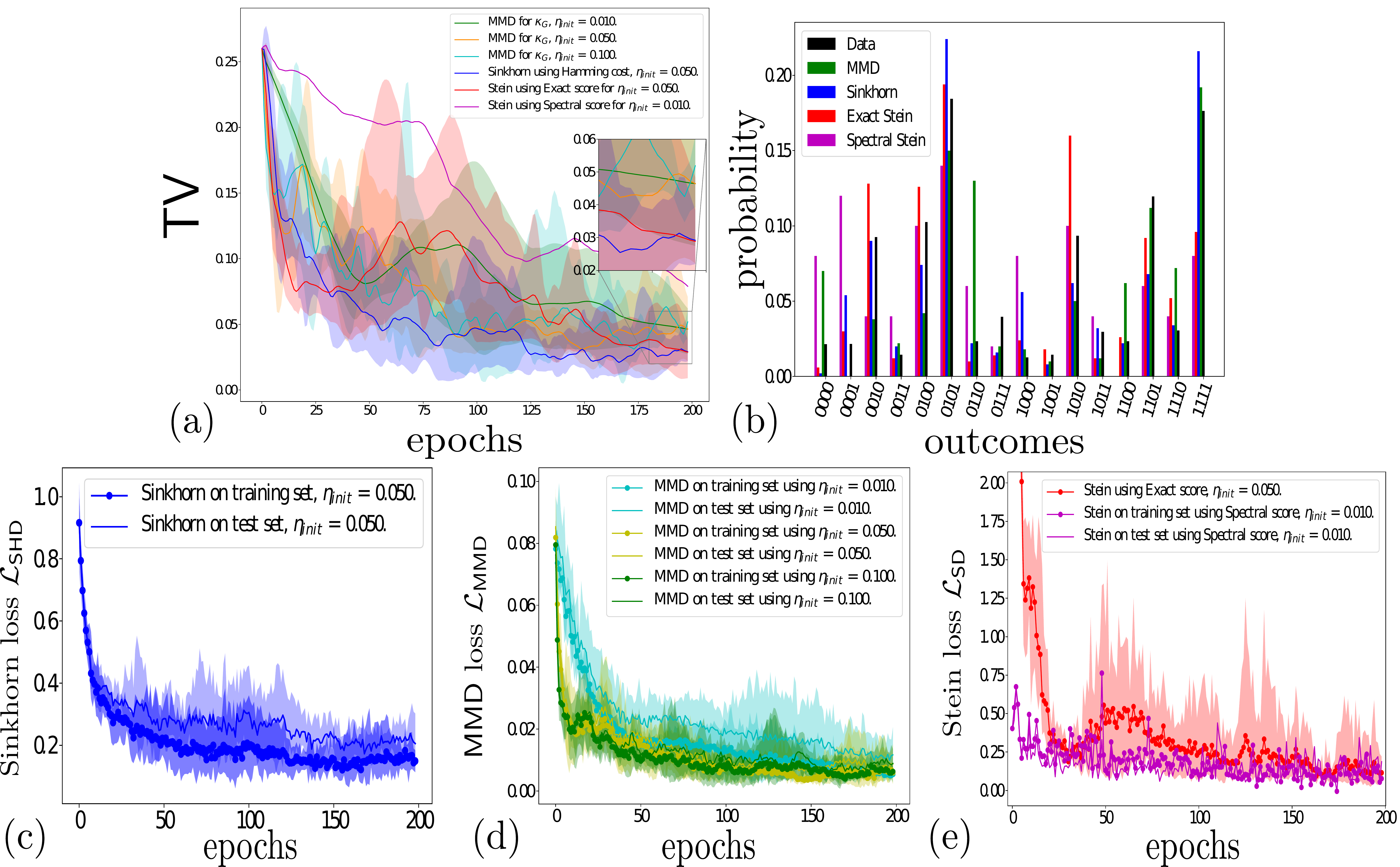}
    \caption{$\MMD$ [\crule[cyan]{0.2cm}{0.2cm}, \crule[yellow]{0.2cm}{0.2cm}, \crule[ForestGreen]{0.2cm}{0.2cm}] vs.\@ Sinkhorn [\crule[blue]{0.2cm}{0.2cm}]  with regularisation parameter $\epsilon = 1$ and Stein training with Spectral score method [\crule[magenta]{0.2cm}{0.2cm}] using 6 eigenvectors for 4 qubits. We use qubit topology in fully connected graph for four qubits, i.e.\@ Rigetti {\fontfamily{cmtt}\selectfont 4q-qvm},  \protect\fourqqvm \ \ \ . Plots show mean, maximum and minimum values over 5 independent training runs on the same dataset. (a) $\TV$ Difference between training methods. Both Sinkhorn divergence and Stein discrepancy can achieve lower $\TV$ values than the $\MMD$. (b) Final learned probabilities of target data, \eqref{toydatadistribution} [\crule[black]{0.2cm}{0.2cm}]. (c) $\mathcal{L}^{1}_{\SH}$ using 500 samples and a batch size of 250. 400 samples used as a training set, 100 for the test set. Trained using Hamming optimal transport cost function (d) $\mathcal{L}_{\MMD}$ using 500 samples and a batch size of 250, with three different initial learning rates. (e) $\mathcal{L}_{\SD}$ with the Exact score and a learning rate of $\eta_{init} = 0.05$ and Spectral score, using 50 Samples and batch size of 25 using a learning rate $\eta_{init}=0.01$. 40 samples used as training data, 10 samples used for test set.}
    \label{fig:MMDvSinkvStein4}
\end{figure}
%


We also performed experiments on the 16 qubit QPU of Rigetti, {\fontfamily{cmtt}\selectfont Aspen}, as seen in \suppfigref{fig:MMDvSinkvStein3_real} and \figref{fig:MMDvSink4_real} in the main text, to determine the performance of the training on real quantum hardware. We used two sublattices for 3 and 4 qubits respectively, the {\fontfamily{cmtt}\selectfont Aspen-4-3Q-A} and {\fontfamily{cmtt}\selectfont Aspen-4-4Q-A}, and their respective {\fontfamily{cmtt}\selectfont qvm} versions. As before we ran 5 independent runs of the training procedure from the same initial condition, and took averages over the run. We restricted to the native connectivity of the chip, as seen in \figref{fig:MMDvSinkvStein3_real}(e), and used the available qubits in the {\fontfamily{cmtt}\selectfont Aspen-4-3Q-A} chip, $(10, 11, 17)$, with no direct connection between qubits $11$ and $17$, as illustrated. Taking this into account, we did not enforce a fully connected topology on any of the real on-chip experiments, since this would have resulted in the compiler implementing SWAP gates, and introducing extra noise. This is one reason why the models trained reasonably well on the hardware, with the two qubit gates in our $\IBM$ being close to the native architecture of the Rigetti chip (where CZ gates are the native entangling links), and so providing an advantage.

\suppfigref{fig:MMDvSinkvStein3_real}(a) compares training using the $\MMD$ with a Gaussian kernel and the Sinkhorn divergence on the QPU, benchmarked relative to $\TV$ as with the previous numerical results. As expected, the QPU noise leads to a higher variance between training runs, but the large number of samples taken $(500)$ still permits the model to learn quickly. This reproduces the behaviour seen in \figref{fig:MMDvSink4_real}(a) in the main text. However, with the smaller amount of qubits, the advantage of Sinkhorn versus $\MMD$ on the real hardware is not observed as it was with 4 qubit.

Interestingly, in \suppfigref{fig:MMDvSinkvStein3_real}(d), the hardware actually trains to lower values of $\MMD$ than the simulator, despite the reverse situation relative to total variation seen in \suppfigref{fig:MMDvSinkvStein3_real}(a) (i.e.\@ the simulator outperforms the hardware). This potentially indicates that the training with the $\MMD$ on noisy hardware could lead to overconfident, and incorrect results. 

\begin{figure}
    \centering
    \includegraphics[width=0.95\columnwidth, height=0.48\columnwidth]{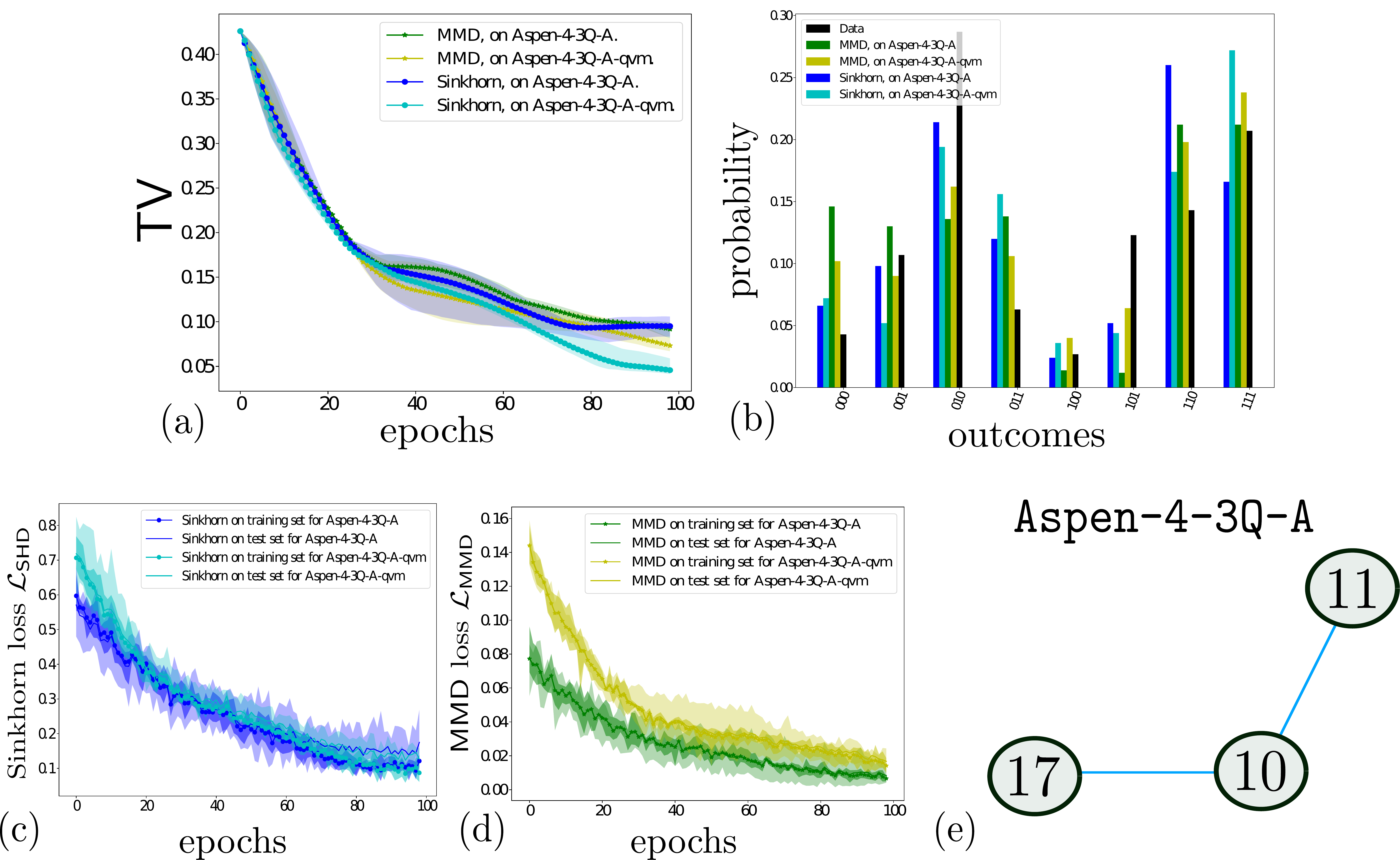}
    \caption{$\MMD$ vs.\@ Sinkhorn for 3 qubits comparing performance on the real QPU ({\fontfamily{cmtt}\selectfont Aspen-4-3Q-A}) vs.\@ simulated behaviour on QVM ({\fontfamily{cmtt}\selectfont Aspen-4-3Q-A-qvm}) using 500 samples and a batch size of 250, with an initial learning rate of $\eta_{init}=0.01$ for Adam. Sinkhorn divergence training outperforms the $\MMD$ both in simulator, relative to total variation, but this advantage does not persist on hardware as with 4 qubit case. (a) $\TV$ Difference between $\MMD$ [\crule[ForestGreen]{0.2cm}{0.2cm}, \crule[yellow]{0.2cm}{0.2cm}], and Sinkhorn [\crule[blue]{0.2cm}{0.2cm}, \crule[cyan]{0.2cm}{0.2cm}] with regularisation parameter $\epsilon = 0.1$ on QVM vs QPU. (b) Final learned probabilities of target data [\crule[black]{0.2cm}{0.2cm}] using $\MMD$ [\crule[ForestGreen]{0.2cm}{0.2cm}, \crule[yellow]{0.2cm}{0.2cm}] LR $\eta_{\mathsf{init}} = 0.2$ and Sinkhorn [\crule[blue]{0.2cm}{0.2cm}, \crule[cyan]{0.2cm}{0.2cm}] with $\epsilon = 0.1, \eta_{\mathsf{init}} = 0.08$. (c) $\mathcal{L}^{0.2}_{\SH}$ for 3 qubits trained on the data \eqref{toydatadistribution} on QVM  [\crule[cyan]{0.2cm}{0.2cm}] vs.\@ QPU  [\crule[blue]{0.2cm}{0.2cm}]. (e) $\mathcal{L}_{\MMD}$ for 3 qubits trained on the data \eqref{toydatadistribution} on QVM  [\crule[yellow]{0.2cm}{0.2cm}] vs.\@ QPU  [\crule[ForestGreen]{0.2cm}{0.2cm}]. (f) Qubit `line' topology in Rigetti {\fontfamily{cmtt}\selectfont Aspen-4-3Q-A} chip, using qubits, $(10, 11, 17)$. }
    \label{fig:MMDvSinkvStein3_real}
\end{figure}
Finally, \suppfigref{fig:autocompilationthreequbits} illustrates the compilation procedure discussed in the main text for three qubits, on the {\fontfamily{cmtt}\selectfont 3q-qvm}. The results are similar to that of the two qubit example; the model is not able to learn \textit{the same} parameters as the target data, but it is able to mimic the output distribution.

\begin{figure}
    \centering
    \includegraphics[width=0.72\columnwidth, height=0.46\columnwidth]{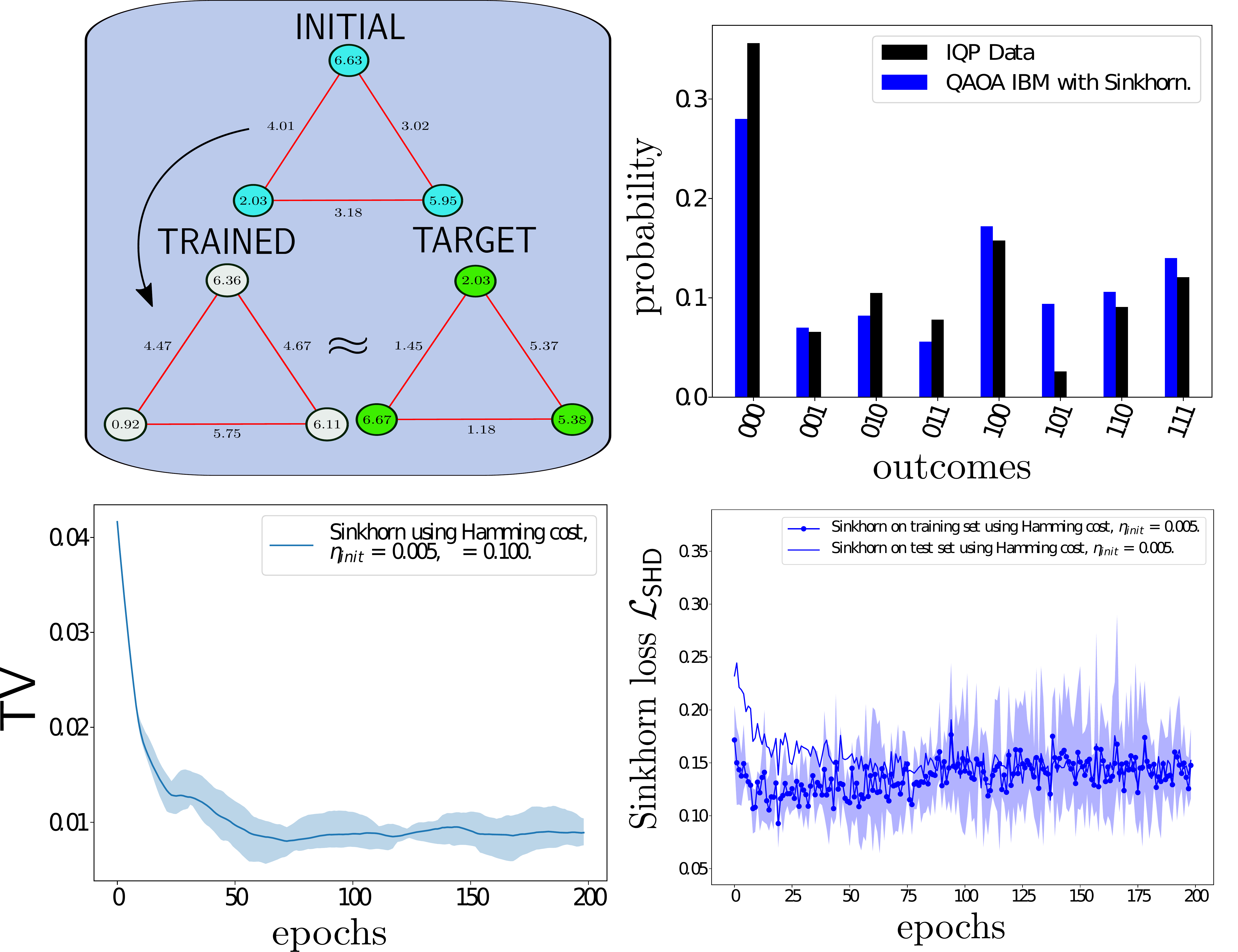}
    \caption{Automatic Compilation of $\IQP$ circuit to a $p = 1 \QAOA$ circuit with three qubits with $\mathcal{L}_{\SH}^\epsilon$ with $\epsilon = 0.1$. A learning rate of $\eta_{init} = 0.005$ was used in the Adam optimiser. $\IBM$ circuit is able to mimic the target distribution well, even though actual achieved parameter values, and circuit families are different. Error bars represent mean, maximum and minimum values achieved over 5 independent training runs on the same data set. (a)  Initial [\crule[cyan]{0.2cm}{0.2cm}] and trained [\crule[Lavender]{0.2cm}{0.2cm}] $\QAOA$ circuit parameters for three qubits. Target $\IQP$ circuit parameters [\crule[green]{0.2cm}{0.2cm}]. (b) Final learned probabilities of $\IBM$ ($\QAOA$) [\crule[blue]{0.2cm}{0.2cm}] circuit versus `data' probabilities ($\IQP$) [\crule[black]{0.2cm}{0.2cm}]. (c) Total Variation Distance during training. (d) Sinkhorn divergence for  400 training samples, 100 test samples, using a Hamming optimal transport cost.}
    \label{fig:autocompilationthreequbits}
\end{figure}

\section{Classical Simulation of Quantum Computations}
\label{supp_matt:hardness}

In this section, we provide relevant definitions relating to quantum computational supremacy, and the notions of simulation error we utilise. We then provide the theorem detailing the classes of parameters for which the $\IBM$ is hard to simulate, which we alluded to in the main text, and prove said theorem. We also show how the conditions of that theorem can be enforced for most of the circuit families encountered during training. Firstly, however, it is necessary to provide relevant definitions.

The central question behind the exploration of quantum supremacy is whether or not it is possible to design a classical algorithm which could, for any output distribution $p(\mathbf{x})$ that can be sampled from by a quantum computer, produce samples from a probability distribution $q(\mathbf{x})$, which is \textit{close} to $p(\mathbf{x})$. This notion of reproducing a quantum distribution can be formalised as classical simulation, of which there are typically two types. For our purposes, the more relevant notion is that of \textit{weak simulation}, which better captures the process of sampling.
\begin{definition}[Strong and Weak Classical Simulation\citeS{bremner_classical_2011, fujii_commuting_2017}]  \label{defn:strong_weak_sim}
    A uniformly generated quantum circuit, $C$, from a family of  circuits, with input size $n$, is \emph{weakly simulatable} if, given a classical description of the circuit, a classical algorithm can produce samples, $\mathbf{x}$, from the output distribution, $p(\mathbf{x})$, in $poly(n)$ time. 
    
    On the other hand, a \emph{strong simulator} of the family would be able to compute the output probabilities, $p(\mathbf{x})$, and also all the marginal distributions over any arbitrary subset of the outputs. Both of these notions apply to with some notion of error, $\epsilon$.
\end{definition}
As mentioned in ref. \citeS{bremner_classical_2011}, strong simulation is a harder task than weak simulation, and it is this weak simulatability which we want to rule out as being classically hard. The suitable notion of error, $\epsilon$, for strong simulation would be the precision to which the probabilities can be computed. 

The specific instances of problems which are classically hard is captured by \text{worst case} and \textit{average case} hardness. Informally, worst case implies there is \textit{at least} one instance of the problem which is hard to simulate. This worst case hardness holds for $\IQP/\QAOA$ circuits, ref. \citeS{bremner_classical_2011, farhi_quantum_2016}, which we will illustrate shortly. A stronger notion is that of average case hardness, which has been proven for Random Circuit Sampling\citeS{bouland_quantum_2018}, and \textsf{BosonSampling} \citeS{aaronson_computational_2013}, but is only conjectured to hold for $\IQP$ circuits for example.

One could ask ``What if we do not care about getting samples from the \textit{exact} distribution, and instead an approximation is good enough?". `Exact' in this case refers to the outcome probabilities of the simulator being identical to those outputted by the quantum device; $q(\mathbf{z}) = p(\mathbf{z})~ \forall \mathbf{z}$ or $\epsilon = 0$. This is a very important and relevant question to ask when discussing quantum supremacy since experimental noise means it could be that even quantum computers cannot produce the exact dynamics that they are supposed to, according to the theory. Worse still, noise typically results in decoherence and the destruction of entanglement and interference in quantum circuits, so in the presence of noise the resulting output distribution could become classically simulatable.

We wish to have strong theoretical guarantees that experiments which claim to demonstrate supremacy, even in the presence of reasonable noise, do in fact behave as expected. Since we are dealing with probability distributions, there are many notions of error one could choose. One of the simplest examples is multiplicative error.

\begin{definition}[Multiplicative Error]
    \label{defnmulterror}
    A circuit family is weakly simulatable within multiplicative (relative) error, if there exists a classical probabilistic algorithm, $Q$, which produces samples, $\mathbf{z}$, according to the distribution, $q(\mathbf{z})$,  in time which is polynomial in the input size, which differs from the ideal quantum distribution, $p(\mathbf{z})$, by a multiplicative constant, $c > 1$:
    \begin{equation}
        \frac{1}{c}p(\mathbf{z}) \leq q(\mathbf{z}) \leq c p(\mathbf{z}) \qquad \forall \mathbf{z} \label{multerror}
    \end{equation}
\end{definition}

As noted in ref.\citeS{fujii_impossibility_2018}, it would be desirable to have a quantum sampler which could achieve the bound of \eqref{multerror}, but this is not believed to be experimentally achievable. It is not believed that a \textit{physical} quantum device could achieve such a multiplicative error bound on its probabilities, relative to its ideal functionality (i.e.\@ replacing $q$ in \eqref{multerror} by the output distribution of a noisy quantum device). That is why much effort has been put into trying to find systems for which supremacy could be provably demonstrated according to the total variation distance error condition, \eqref{variationdistanceerror}, which is easier to achieve on near term quantum devices. 

\begin{definition}[Total Variation ($\TV$) Error]
    \label{defnvardiserror}
    A circuit family is weakly simulable within variation distance  error, $\epsilon$, if there exists a classical probabilistic algorithm, $Q$, which produces samples, $\mathbf{x}$, according to the distribution, $q(\mathbf{x})$, in polynomial time, such that it differs from the ideal quantum distribution, $p(\mathbf{x})$ in total variation distance, $\epsilon$:
    \begin{equation}
       \TV(p, q)  \coloneqq \frac{1}{2}\sum\limits_{\mathbf{x}}|p(\mathbf{x})- q(\mathbf{x})| \leq \epsilon \label{variationdistanceerror} 
    \end{equation}
\end{definition}

Intuitively, multiplicative error sampling is `harder' since it requires stricter relationships between the real and target distributions at every outcome. For example multiplicative error must preserve the probability of those outcomes with zero probability of occurring in the target distribution which is particularly unlikely to happen in noisy devices. Total variation distance error bounds indicate that the distributions can have a bounded difference in probability that is constant over each output, as opposed to multiplicative error where the difference depends on the output. This allows for greater flexibility by requiring the distributions to \emph{only} be similar overall. One could also ask for hardness within other metrics, perhaps those studied here, including the Wasserstein metric for example.

\subsection{Multiplicative Error Hardness of Ising Born Machine Circuits}

Now that we have the necessary definitions, we demonstrate why the core of the $\IBM$ (the underlying circuit) should be hard to sample from efficiently by purely classical means, up to multiplicative error, by pulling together the relevant hardness results from previous works in this area. Specifically, we formalise the discussion in the above text and prove the following theorem:

\begin{theorem}\label{thm:ibmmultsimulationhardness}
    If the output probability distributions generated by uniform families of $\IBM({\boldsymbol\theta})$ circuits could be weakly classically simulated to within multiplicative error $1 \leq c \leq \sqrt{2}$ then $\PBPP = \PP$ (and hence the polynomial hierarchy collapses to the third level), where $\forall k, i, j,$:
    \begin{equation}
        J_{ij}, b_{k} = 
        \begin{cases}
            \frac{(2l+1)\pi}{8d} &\text{ for integers, } d, l\\
            2\nu \pi& \nu \in[0,1) \text{ irrational.}
        \end{cases}\label{hardparametervalues1}
    \end{equation}
    and we consider the following instances:
    \begin{align}
        \Gamma_k = 0, \Delta_k &=  
        \begin{cases}
            \frac{(2l+1)\pi}{8d} &\text{ for integers, } d, l\\
            2\nu \pi& \nu \in[0,1) \text{ irrational.}
        \end{cases}\label{hardparametervalues2}\\
        \Delta_k = 0, \Gamma_k &= 
        \begin{cases}
            \frac{(2l+1)\pi}{4d} &\text{ for integers, } d, l\\
            2\nu \pi& \nu \in[0,1) \text{ irrational.}
        \end{cases}\label{hardparametervalues3}\\
        \Delta_k = 0, \Gamma_k = \Sigma_k &=  
        \begin{cases}
            \frac{(2l+1)\pi}{2\sqrt{2}d} &\text{ for integers, } d, l\\
            2\nu \pi& \nu \in[0,1) \text{ irrational.}
        \end{cases}\label{hardparametervalues4}
\end{align}
\end{theorem}

As discussed in the main text, the following choices of $\mathbf{\Gamma}, \mathbf{\Delta}, \mathbf{\Sigma}$ give the more well known circuit classes:

\begin{align}
    \IBM\left(\{J_{ij}, b_{k}\}, \forall k: \Gamma_k = \frac{\pi}{2\sqrt{2}}, \Delta_k = 0, \Sigma_k =  \frac{\pi}{2\sqrt{2}}\right) &= \IQP(\{J_{ij}, b_{k}\})\\
    \IBM(\{J_{ij}, b_{k}\}, \mathbf{\Gamma} = -\mathbf{\Gamma}, \mathbf{0} , \mathbf{0}) &= \QAOA_{p=1}(\{J_{ij}, b_{k}\}, \mathbf{\Gamma})
\end{align}

\begin{proof}

The proof of \theref{thm:ibmmultsimulationhardness} involves stitching together several results, some well known, some less so, and, to the best of out knowledge, some unknown. The proof will proceed systematically through the instances of the parameters, $\left\{ J_{ij}, b_{k}\right\}, \mathbf{\Gamma}, \mathbf{\Delta}, \mathbf{\Sigma}$, and will show that in each of the cases mentioned in \theref{thm:ibmmultsimulationhardness} the class of circuits generated cannot, in the worst case, be simulated to within multiplicative error efficiently by any classical means. For the remainder of this proof, we use the phrase `simulate' to mean exactly this.

\subsubsection*{\texorpdfstring{$\IQP$}{IQP}}\label{appa:iqphardnessproof}
As stated above, $\IQP$ circuits correspond to the setting of the parameters, $\{J_{ij}, b_{k}\}, \forall k: \Gamma_k = \pi/2\sqrt{2} , \Delta_k = 0, \Sigma_k= \pi/2\sqrt{2} $. This means the unitary $U_z(\theta)$ is applied to the initial superposition state, $\ket{+}^{\otimes n}$, followed finally by a Hadamard applied to all qubits, $H^{\otimes n}$. It is known that $\IQP$ circuits, with the homogeneous parameters $\alpha = J_{ij} = b_{k} =  \pi/8\  \forall i,j,k$ in \eqref{iqp_ising_hamiltonian}, are hard to simulate in the worst case \citeS{bremner_classical_2011}. We will denote $\IQP$ circuits with the two sets of parameters used to define them, $\{J_{ij}, b_k\}$, by $\IQP({J_{ij}}, {b_k})$.

\begin{theorem}[Theorem 2 from ref. \citeS{bremner_classical_2011}]\label{thm:iqphardnessbremner}
If the output probability distributions generated by uniform families of $\IQP(\pi/8, \pi/8)$ circuits could be weakly classically simulated to within multiplicative error $1 \leq c \leq \sqrt{2}$ then $\PBPP$ = $\PP$ (and hence the polynomial hierarchy collapses to the third level).
\end{theorem}

A collapse of $\PH$ to any level is considered unlikely, and in some sense is a generalisation of $\Pee \neq \NP$.
Next, we consider the question of for \textit{which} homogeneous parameters, $\boldsymbol\alpha_l = \alpha = J_{ij} = b_{k}$ result in $\IQP$ circuit families which are hard to simulate. The answer is quite a few of them.

\begin{theorem}[Adapted from Theorem 5 of ref.\citeS{fujii_commuting_2017}]\label{thm:iqphardnessfujii}
If the output probability distributions generated by uniform families of $\IQP(\theta, \theta)$ circuits
could be weakly classically simulated to within multiplicative error $1 \leq c \leq \sqrt{2}$ then $\PBPP = \PP$ (and hence $\PH$ collapses to the third level), where:
\begin{equation}
\theta = \begin{cases}
\frac{(2l+1)\pi}{8d} &\text{ for integers, } d, l\\
 2\nu \pi& \nu \in[0,1) \text{ irrational.}
 \end{cases}\label{hardparametervaluesiqp}
\end{equation}
\end{theorem}

Finally we note that there is no need for the circuit parameters to be homogeneous. We may also allow each single and two qubit gate to have an independent parameter, as long as they are \textit{all} of the form of \eqref{hardparametervaluesiqp}. In \theref{thm:iqphardnessbremner} it was shown that a general computation can be simulated by $\IQP$ circuits generated by a homogeneous parameter value, $\theta = \pi/8$. We will now use the argument of ref.\citeS{fujii_commuting_2017} to show why $\IQP(J_{ij}, b_{k})$ circuits, with almost all parameter angles, can simulate $\IQP(\pi/8, \pi/8)$ efficiently. The result of ref.\citeS{fujii_commuting_2017} will be a subcase of this, which accounts for the case $J_{ij} = b_k \neq \pi/8$.  Specifically, if the original $\IQP$ circuit is generated by gates in the form $D_1(\pi/8) = e^{\mathrm{i}\frac{\pi}{8}Z}, D_2(\pi/8) = e^{\mathrm{i}\frac{\pi}{8}Z \otimes Z}$, general $\IQP$ circuits with gates acting on at most two qubits ($|S_j|\leq 2$) will be generated by gates $D_1(b_{k}) = e^{\mathrm{i}b_{k}Z}, D_2(J_{ij}) = e^{\mathrm{i}J_{ij} Z\otimes Z}$. Therefore, it is necessary to show how each of these gates can simulate each of the former gates with a homogeneous rotation angle, $\pi/8$. To do so we can use the error measure defined as follows\citeS{nielsen_quantum_2011} as the difference between the operations of two arbitrary gates on a quantum state, when maximised over all possible states:
\begin{equation}
    E(U, V)  \coloneqq \max\limits_{\ket{\psi}}||(U-V)\ket{\psi}|| \label{error}
\end{equation}
where $||\cdot||$ is the norm of a vector: $||U\ket{\psi}|| = \sqrt{\bra{\psi}U^\dagger U\ket{\psi}}$. If, by $m$ repetitions, the gate is within $\epsilon$ of the required gate with parameter $\pi/8$, $U(\pi/8+\epsilon)$, the error induced by this extra $\epsilon$ factor will be $\mathcal{O}(\epsilon)$:
\begin{align}
    E\left(U(\pi/8+\epsilon), U(\pi/8)\right) &= |1-e^{i\frac{\epsilon}{2}}| =|1-(1-i\epsilon/2-\epsilon^2/8+\mathcal{O}(\epsilon^3))|\\
    &=|i\epsilon/2+\epsilon^2/8+\mathcal{O}(\epsilon^3))| = \sqrt{\left(\epsilon^2/8\right)^2 + \left(\epsilon/2\right)^2 +\mathcal{O}(\epsilon^3)}\\
    &= \sqrt{\epsilon^4/64 + \epsilon^2/4^2 +\mathcal{O}(\epsilon^3)} = \epsilon/2\sqrt{1+\epsilon^2/16 +\mathcal{O}(\epsilon^3)} = \mathcal{O}(\epsilon)
\end{align}
Specifically, with the two-qubit gate for example:
\begin{equation}
    D_2(J_{ij})^m = D_2(mJ_{ij}) = D_2\left(\frac{\pi}{8}+\epsilon\right) \label{epsilonclose}
\end{equation}
so the gate with parameter $\theta_2$ can be made $\epsilon$ close to the required $CZ \sim D_2(\frac{\pi}{8})$ with respect to the error, as noted by ref. \citeS{fujii_commuting_2017, nielsen_quantum_2011}. The same thing holds for the single qubit gates with angle $\theta_1$ approximating the $Z$ or $T$ gates for example.

The irrationality of the parameter values makes \eqref{epsilonclose} true since $2\pi\nu m\mod 2\pi$ is distributed uniformly. See ref.\citeS{boykin_universal_1999} for a proof that any arbitrary phase can be approximated to accuracy, $\epsilon$, with $\mathsf{poly}(1/\epsilon)$ repetitions of a phase which is an irrational multiple of $2\pi$. This shows that it is possible to achieve any gate: $e^{\mathrm{i}2\nu\pi \hat{n}\dot\sigma}$ using $m$ repetitions of $(e^{\mathrm{i}2\nu\pi \hat{n}\dot\sigma})^n$ if $\nu$ is irrational. In this construction, each individual gate, $j$, will have an error, $\epsilon_j$, in contrast to ref.\citeS{fujii_commuting_2017} in which all of the errors would be the same (or a constant multiple ($\epsilon = \epsilon_j \forall k$). However, as long as this parameter, $\epsilon_j$,  for each gate is lower than the threshold for fault-tolerant quantum computation, as noted in ref.\citeS{fujii_commuting_2017}, then we can reliably simulate universal quantum computation, and hence $\PIQP( J_{ij}, b_{k}) = \PIQP(\pi/8, \pi/8)$.

\subsubsection*{\texorpdfstring{$p = 1 \QAOA$}{pQAOA} \label{appa:hardnessqaoaproof}}

Now, we can prove the analogous statement to that of $\IQP$ with this setting of the parameters, in an identical way.
\begin{enumerate}
    \item The case $\forall k: \Gamma_k =  -\pi/4$ is covered by ref.\citeS{farhi_quantum_2016} where it is shown how $\QAOA(\pi/8, \pi/8, \pi/4)$ equipped with postselection is equivalent to $\BQP$ with postselection. The proof involves the design of a gadget similar to the $\IQP$ gadget in the $\IQP$ proof of hardness, and results in a similar collapse of the $\PH$ as in Theorem (\ref{thm:iqphardnessbremner}).
    \item Extending the parameters, $J_{ij} = b_{k} = \pi/8$ in the diagonal unitary $U_z(\theta)$ to general parameters, $J_{ij}, b_k$ follows identically to the argument for $\IQP$ in the previous section. More specifically, any diagonal gate with parameter, $\pi/8$, can be simulated by a gate with any parameter of the form of \eqref{hardparametervalues1}, applied a constant number of times to get within a fixed error of the desired $\pi/8$ gate.
    \item Finally, if we wish to simulate the behaviour of any gate $\tilde{H} = e^{-\mathrm{i}\pi/4 X_k}$, by any general gate $U_f(\Gamma_k) =  e^{-\mathrm{i}\Gamma_k X_k}$ on qubit $k$, we just need to apply the latter gate a constant, $m = \mathcal{O}(1/\epsilon)$, number of times, to get within $\epsilon$, of the gate $\tilde{H}$, exactly as in \eqref{epsilonclose}. Now, we will have $n$ such gates, each requiring the application of a constant number of repetitions, $m_k = \mathcal{O}(1/\epsilon_k)$, so the total number of gates that would have to be applied is $\prod\limits_{k=1}^n m_k$.
    Overall, we will acquire a polynomial overhead in the simulation of the circuit generated by the fixed parameters, $\forall i,j,k: J_{ij} = \pi/8,b_k =  \pi/8,\Gamma_k =  -\pi/4$. Hence we can achieve universal quantum computation with postselection in exactly the same way, and hence $\PQAOA(\{J_{ij}, b_{k}\} , \mathbf{\Gamma}) = \PBQP$.
\end{enumerate}

\subsubsection*{Remaining Cases}\label{ssec:remaininghardnesscases}
The above two sections cover many of the angles for the final measurement unitary to result in a hard circuit class. We can immediately extend the argument to cover many more angles. 
The following case is a generalisation of $\IQP$:
\begin{equation}
        \Delta_k = 0, \Gamma_k = \Sigma_k =  
        \begin{cases}
            \frac{(2l+1)\pi}{2\sqrt{2}d} &\text{ for integers, } d, l\\
            2\nu \pi& \nu \in[0,1) \text{ irrational.}
        \end{cases}\label{hardparametervaluesgenhadamardapp}
\end{equation}
where the final angle is a rotation in the `Hadamard' axis, i.e.\@\@ a rotation around the axis $1/\sqrt{2}(X+Z)$, but by a more general angle than $\pi$.

Secondly, we consider when the final measurement unitary is a rotation around the Pauli-$Y$ axis, i.e.\@ the $\IBM$ has the following parameters:
\begin{equation}
        \Sigma_k = 0, \Gamma_k = 0, \Delta_k =  
        \begin{cases}
            \frac{(2l+1)\pi}{4d} &\text{ for integers, } d, l\\
            2\nu \pi& \nu \in[0,1) \text{ irrational.}
        \end{cases}\label{hardparametervaluesYbasis}
\end{equation}
where the final measurement unitary is
\begin{equation}
U_f(\Delta_k) = R_y(-2\Delta_{k}) =  e^{\mathrm{i}\Delta_{k} Y_k} \label{yrotation}
\end{equation}
We can see how this family is also hard to simulate by relating it to $\IQP$ circuits as follows. If there exists no classical randomised algorithm to produce samples, $\mathbf{z}$, in polynomial time according to the $\IQP$ distribution:

\begin{equation}
p_{\IQP}(\mathbf{z}) = |\bra{\mathbf{z}}H^{\otimes n}U_z(\boldsymbol\alpha) H^{\otimes n}\ket{0}^{\otimes n}|^2 \label{iqpdist2}
\end{equation}
Let $p^*$ be the output distribution produced by these types of circuits with a final Pauli-$Y$ rotation:
\begin{equation}
p^*(\mathbf{z}') = |\bra{\mathbf{z}'}\bigotimes\limits_{k=1}^n e^{\mathrm{i}\frac{\Delta_k}{2} Y_k}U_z(\boldsymbol\alpha) H^{\otimes n}\ket{0}^{\otimes n}|^2 \label{iqpydistribution}
\end{equation}
This is due to the relationship: $H  =\frac{1}{\sqrt{\mathrm{i}}}XY^{\frac{1}{2}} = \frac{1}{\sqrt{\mathrm{i}}}XR_y(\frac{\pi}{2})$. Therefore, choosing $\Delta_k = \pi/4$, we get that the two distributions in \eqref{iqpdist2}, \eqref{iqpydistribution} and are related as follows:

\begin{align}
p_{\IQP}(\mathbf{z}) &= |\mathrm{i}^{-n}\bra{\mathbf{z}}X^{\otimes n}\sqrt{Y}^{\otimes n}U_z(\boldsymbol\alpha) H^{\otimes n}\ket{0}^{\otimes n}|^2\\
&= |\bra{\mathbf{z}'}\bigotimes\limits_{k=1}^ne^{\mathrm{i}\frac{\pi}{4} Y_k}U_z(\boldsymbol\alpha) H^{\otimes n}\ket{0}^{\otimes n}|^2  = p^*(\mathbf{z}')
\end{align}
so $\mathbf{z}'$ is simply $\mathbf{z}$ with every bit flipped. Clearly, if one could produce samples, $\mathbf{z}$, by some classical means efficiently, then the same algorithm with one extra step can be used to produce the samples, $\mathbf{z}'$, and hence the choice of parameters $\IBM(\boldsymbol\alpha, \mathbf{0}, \{\Delta_k = \pi/4\},\mathbf{0})$ is also classically hard to sample from, where $\boldsymbol\alpha = \{ J_{ij}, b_{k}\}$ have the same constraints as in $\IQP$ or $\QAOA$. The extension from $\Delta_k = \pi/4$ to almost any general $\Delta_k$, according to \eqref{hardparametervaluesYbasis}, follows in an identical fashion from the above two cases so these circuit classes with postselection are also equal to $\PBPP$ and if they could be simulated, once again the $\PH$ collapses to the third level.

\end{proof}

\subsection{Total Variation Distance Hardness of Ising Born Machine Circuits \label{supp_matt:ibm_variation_hardness}}

We can improve the hardness of the model by incorporating a stronger result about the circuit class, $\IQP$. The $\IBM$ model must be initialised to some setting of the circuit parameters to begin the training process. One such possible initialisation is to randomly assign $J_{ij}, b_{k}$ to some subset with uniform probability. The random initialisation is typical in machine learning, but it is not the only way one could initialise the parameters. In this case however, we will do so in order to use the $\IQP$ result of ref.\citeS{bremner_average-case_2016}.

Firstly, define the \textit{partition function},  $\mathcal{Z}$, associated with the Ising Hamiltonian, \eqref{iqp_ising_hamiltonian_supp}:

\begin{align}
 \mathcal{H}  \coloneqq \mathrm{i}\sum\limits_{i<j} J_{ij}z_iz_j + \mathrm{i}\sum\limits_{k=1}^n b_k z_k \label{iqp_ising_hamiltonian_supp}\\
    \mathcal{Z}  \coloneqq \sum\limits_{z\in\{\pm 1\}^n}e^{  \mathrm{i}\sum\limits_{i<j} J_{ij}z_iz_j + \mathrm{i}\sum\limits_{k=1}^n b_k z_k} \label{iqp_isingpartition_fun_supp}
\end{align}
Now, an output amplitude of an $\IQP$ circuit can be written as this partition function:
\begin{equation}
    2^n\bra{\mathbf{z}}U_z(\boldsymbol\alpha)\ket{\mathbf{z}} = \mathcal{Z}
\end{equation} 
The hardness of $\IQP$ circuits is related to computing this partition function in a similar way that in which $\BosonSampling$ is related to computing the permanent of a matrix\citeS{aaronson_computational_2013}.
Now, we can immediately use the following theorem:

\begin{theorem}[From ref.\citeS{bremner_average-case_2016}] 
\label{thm:iqphardnessbremneradditive}
    Assume it is $\#\Pee$-hard to approximate $|\mathcal{Z}|^2$ up to a relative error  $1/4+o(1)$ for $1/24$ fraction of instances over the choice of the weights and biases, $J_{ij}, b_k$. If it is possible to classically sample from the output probability distribution of any $\IQP$ circuit in polynomial time, up to an additive error of $1/384$ in total variation distance, then there is a $\BPP^{\NP}$ algorithm to solve any problem in $\Pee^{\#\Pee}$, and hence the polynomial hierarchy collapses to the third level)
\end{theorem}

This holds if the parameters are chosen uniformly at random from $\{J_{ij}, b_k\} \in\{0,\frac{\pi}{8}, \dots, \frac{7\pi}{8}\}$, which corresponds to randomly choosing a circuit from the $\IQP$ circuit class according to some measure over the unitary group. 

We refer to ref.\citeS{bremner_average-case_2016} for the proof of the above but outline it here. This proof relies on a conjecture (included in Theorem (\ref{thm:iqphardnessbremneradditive})), which claims that \textit{on average} (i.e.\@\@ over a $\frac{1}{24}$ fraction of instances) the Ising partition function is hard to compute up to a multiplicative error. Again, this holds only in the worst case, and converts an average case conjecture in multiplicative error into a worst case result in total variation error. 

If we were to choose the setting of parameters such that they correspond to an $\IQP$ circuit, and initialise $J_{ij}, b_k$ to the values above, then we shall start with an $\IBM$ in a regime which is also hard to simulate classically up to a variation distance error in the worst case. A random initialisation of parameterised quantum circuits has been shown to potentially lead to `barren plateaus'\citeS{mcclean_barren_2018}, in which the gradient with respect to some quantum circuits becomes exponentially small, and so one would need exponential resources to estimate it. This indeed could be an issue for training such $\IBM$ circuits as the number of qubits scales, but this question is currently under investigation, with some potential solutions found for circuit initialisation \citeS{grant_initialization_2019, verdon_learning_2019}. However, in this work, we assume a random (but fixed) initialisation for simplicity as we are solely interested in the performance of various cost functions.

Collecting the results required in Theorem (\ref{thm:ibmmultsimulationhardness}) is necessary, since we want to make the strongest arguments for the intractability of classically simulating the $\IBM$. If we were not careful, the parameter updates could lead us into a regime which was classically simulatable.
Now, the initialisation of the parameters, $\theta^{(0)}$, will lead to a circuit class which is hard to sample from classically in the worst case, up to variation distance error, if $U_f$ is chosen such that the $\IBM$ is exactly an $\IQP$ circuit, otherwise it will only be (provably) hard up to multiplicative error. These samples can be used to compute the cost function, $\mathcal{L}_B$. 

Computing the gradients, $\frac{\partial \mathcal{L}}{\partial \theta_k}$, for all choices of cost function, $B = \{\MMD, \SD, \SH\}$ requires running parameter shifted circuits, $p^-_{\theta_k}, p^+_{\theta_k}$. If the original circuit is hard to sample from, these shifted circuits will \textit{also} be hard to sample from. This is due to the fact that they only differ by a single parameter shift, which can be reabsorbed into the original parameters in such a way that the resulting circuit also resides in the hard circuit classes. As an example, using an $\IQP$ circuit, if the parameter to be updated is $J_{ij}$, then that parameter should have been initialised according to Theorem (\ref{thm:iqphardnessfujii}):
\begin{equation}
    J_{ij} = \begin{cases}
        \frac{(2l+1)\pi}{8d} &\text{ for integers, } d, l\\
        2\nu \pi& \nu \in[0,1) \text{ irrational.}
    \end{cases} \label{hardanglesend}
\end{equation}
Therefore:
\begin{equation}
    J_{ij}\pm \frac{\pi}{2} = \begin{cases}
        \frac{(2(l \pm 4d)+1)\pi}{8d} &\text{ for integers, } d, l\\
        2\left(\frac{2\nu \pm 1}{4}\right) \pi& \nu \in[0,1) \text{ irrational.}
    \end{cases}
\end{equation}
A relabelling of $l\pm 4d \rightarrow k, \frac{1}{4}(2\nu \pm 1) \rightarrow \mu$, where $\mu$ is irrational if $\nu$ is, and $k$ is still an integer, gives that the parameter shifted circuits should have exactly the same structure as (\ref{hardanglesend}), and hence should be just as hard. 

The above argument illustrates how the circuits required to compute the gradient for a given parameter, $J_{ij}$, should be as hard to sample from as the original circuit. However, the training algorithm requires we update the parameters of the $\IBM$ circuit, over a number of epochs, in order to provide better fits to the data. As detailed above, this is done using the following update rule:
\begin{equation}
    \theta_l^{(d+1)} \leftarrow \theta_l^{(d)} - \eta \frac{\partial \mathcal{L}_B}{\partial \theta^{(d)}_l} \label{updaterule}
\end{equation}
We can ensure that the system remains in a class which has worst case hardness to simulate as long as we start with an initial configuration of the parameters that demonstrates supremacy, and update them in such a way that this remains the case. For example, if we choose the initial parameters, $\theta^{(0)} = 2\pi\nu$, where $\nu$ is irrational, then an update will be of the form $\theta^{(d+1)} = \theta^{(d)} + \mu$. If we choose $\mu = 2\pi\alpha$, then the new parameters at epoch, $d+1$, will be $\theta^{(d+1)} = 2\pi\beta, \beta = \nu+\alpha$, where $\alpha$ will be the learning rate multiplied by the gradient, $\eta \frac{\partial \mathcal{L}}{\partial \theta^{(d)}_l}$. Since $\nu$ was irrational originally, $\beta$ will also be irrational, and therefore the new configuration should be hard to simulate classically. Although, clearly we cannot allow updates like $\alpha = -\nu (+ \delta)$, where $\delta$ is rational since this will result in the parameter going to $0 (\delta)$ and could make the model classically simulatable. As an example of this, choose the following circuit parameters for the $\IBM(\{J_{ij}, b_k\}, -\mathbf{\Gamma}_k, \Delta_k = 0, \Sigma_k = 0)$, which is a $p=1 \QAOA$ circuit, and then set $\Gamma_k = 0$. In this case, we create a uniform superposition from the initial Hadamards, $1/\sqrt{2^n}\sum_\mathbf{z}\ket{\mathbf{z}}$, and then apply gates diagonal in the computational basis. This will only add phases to each contribution $e^{i\theta}\ket{z_j}$. If we have no final `measurement' gate (i.e.\@ the parameter is zero), this is equivalent to measuring the state in the computational basis at the end. Since the phases do not have the ability to interfere with each other, they will have no effect on the measurement probabilities, and the final distribution will be simply the uniform distribution over all $n$-bit strings, $\mathbf{z}$. This can clearly be classically simulated by a sequence of coin flips. Another example is if we were to allow the entangling parameters, $J_{ij} \rightarrow 0$. This would lead to a circuit composed only of single qubit gates, which again is not classically hard.

There is, of course, a caveat to this argument. The hardness results in Theorem (\ref{thm:ibmmultsimulationhardness}) are only valid in the \textit{worst case}. This means that there exists \textit{some} instance of the problem generated by the set of parameters which cannot be classically simulated efficiently. More specifically, while with each instance of `hard' parameter values, we may end up in a parameter landscape which admits a worst case hardness result, there is obviously no guarantee that the particular instance \textit{implemented} in the $\IBM$ is in fact \textit{the} hard one, and we do not claim to have found such a thing. Furthermore, there is also no guarantee that we go from the `hard' circuit generated from one set of parameters, to the hard circuit using the next set of parameters (given by the gradient update).

Further adding to this hardness argument is the required computation of the intermediate classically-hard kernel (\ref{quantumkernel}), should that be the one which is chosen. This computation must be done \textit{for all} pairs of samples required in the loss, $\mathcal{L}_B$ and its gradient, and as such, increases the conjectured hardness of classically simulating the training algorithm. This is due to the argument of ref.\citeS{havlicek_supervised_2019}, which conjectures that this kernel should be hard to compute for any classical algorithm given only a classical description of the states, up to additive error.

\section{Supremacy of Quantum Learning}
\label{supp_mat:superioritydefinitions}

Here we provide, to the best of our knowledge, the first formalisation of what we call `\textit{quantum learning supremacy}' (QLS), specifically for distribution learning. We model our definitions around those provided in ref.\citeS{kearns_learnability_1994}, which pertain to the theory of classical distribution learnability. 

As discussed in the main text, intuitively a generative quantum machine learning algorithm can be said to have demonstrated QLS, if it is possible for it to efficiently learn a representation of a distribution for which there does not exist a classical learning algorithm achieving the same end. More specifically, the quantum device has the ability to produce samples according to a distribution that is close in total variation to some target, using a polynomial number of samples from the target. However, there should be no classical algorithm which could achieve this. 

We now formalise this intuition. First we must understand the inputs and outputs to learning algorithm. The inputs are samples, either classical vectors, or quantum states encoding a superposition of such bitstring states, i.e.\@ \textit{qsamples} \citeS{schuld_supervised_2018}. A generator can be interpreted as a routine that simulates sampling from the distribution. As in ref.\citeS{kearns_learnability_1994}, we will assume only discrete distribution classes, $\mathcal{D}_n$, over binary vectors of length $n$.

\begin{definition}[Generator \citeS{kearns_learnability_1994}]\label{defn:generator_supp}
    A class of distributions, $\mathcal{D}_n$ has efficient Generators, $\GEN_{D}$, if for every distribution $D \in \mathcal{D}_n$, $\GEN_{D}$ produces samples in $\{0, 1\}^n$ according to the exact distribution $D$, using polynomial resources.
    The generator may take a string of uniformly random bits, of size polynomial in $n$, $r(n)$, as input.
\end{definition}

The reader will notice that this definition allows, for example, for the Generator to be either a classical circuit, or a quantum circuit, with polynomially many gates. Further, in the definition of a classical Generator \citeS{kearns_learnability_1994} a string of uniformly random bits is taken as input, and then transformed into the randomness of $D$. However, a quantum Generator would be able to produce its own randomness and so no such input is necessary. In this case the algorithm could ignore the input string $r(n)$. 

While we are predominately interested in efficient learning with a Generator, one can also define a similar \textit{Evaluator}:
\begin{definition}[Evaluator \citeS{kearns_learnability_1994}]\label{defn:evaluator_supp}
     A class of distributions, $\mathcal{D}_n$ has efficient Evaluators, $\EVAL_{D}$, if for every distribution $D \in \mathcal{D}_n$, $\EVAL_{D}$ produces the weight of an input $\mathbf{y}$ in $\{0, 1\}^n$ under the exact distribution $D$, i.e.\@ the probability of $\mathbf{y}$ according to $D$. The Evaluator is \emph{efficient} if it uses polynomial resources.
\end{definition}

The distinction between $\EVAL$ and $\GEN$ is important and interesting in this case since the output probabilities of even $\IQP$ circuits are $\#\Pee$-Hard to compute\citeS{bremner_classical_2011} and also hard to sample from by classical means, yet the distributions they produce can be sampled from efficiently by a quantum computer. This draws parallels to examples in ref.\citeS{kearns_learnability_1994} where certain classes of distributions are shown not to be learnable efficiently with an Evaluator, but they \textit{are} learnable with a Generator. We also wish to highlight connections to the definitions of strong and weak simulators of quantum circuits in \defref{defn:strong_weak_sim} to reinforce the similarity between computational supremacy and learning. An Evaluator for a quantum circuit would be a strong simulator of it, and a Generator would be a weak simulator. However, we keep these definitions separate in order to connect the hardness and learnability ideas explicitly.

For our purposes, the following definitions of learnable will be used. In contrast to ref.\citeS{kearns_learnability_1994}, which was concerned with defining a `good' generator to be one which achieves closeness relative to the Kullback-Leibler ($\KL$) divergence, we wish to expand this to general cost functions, $d$. This is due to the range of cost functions we have access to and our wish to connect to the quantum circuit hardness results mentioned above, which typically strive for closeness in $\TV$.

\begin{definition}[$\left( d , \epsilon \right)$-Generator]
    For a cost function, $d$, let $D \in \mathcal{D}_n$. Let $\GEN_{D'}$ be a Generator for a distribution $D'$. We say $\GEN$ is a $\left( d , \epsilon \right)$-Generator for $D$ if $d(D, D') \leq \epsilon$.
\end{definition}

A similar notion of an $\epsilon$-good Evaluator could be defined.

\begin{definition}[$\left( d , \epsilon, \C \right)$-Learnable]
    \label{def:weak_learnable_supp}
    For a metric $d$, $\epsilon > 0$, and complexity class $\C$, a class of distributions $\mathcal{D}_n$ is called $\left( d , \epsilon, C \right)$-learnable (with a Generator) if there exists an algorithm $\mathcal{A} \in \C$, called a learning algorithm for $\mathcal{D}_n$, which given $0 < \delta < 1$ as input, and given access to $\GEN_{D}$ for any distribution $D \in \mathcal{D}_n$, outputs $\GEN_{D'}$, a $\left( d , \epsilon \right)$-Generator for $D$, with high probability:
    \begin{equation}
        Pr \left[ d \left( D, D' \right) \leq \epsilon \right] \geq 1 - \delta
    \end{equation}
    $\mathcal{A}$ should run in time $\mathsf{poly}(1/\epsilon, 1/\delta, n)$.
\end{definition}

In \defref{def:weak_learnable_supp}, $\epsilon$ may, for example, be a function of the inputs to the learning algorithm. We may also wish to require a learnability definition which holds for all $\epsilon > 0$. This definition would, however, be too strong for our purposes. In order to claim quantum learning supremacy of a learning algorithm, we only need to achieve closeness up to a \textit{fixed} $\TV$ distance, as mentioned in the main text, and discussed in detail in \appref{supp_matt:hardness}. Finally, we define what it would mean for a quantum algorithm to be superior to any classical algorithm for the problem of distribution learning:

\begin{definition}[Quantum learning supremacy (QLS)]    
    \label{defn:superiority_supp}
    An algorithm $\mathcal{A} \in \BQP$ is said to have demonstrated the supremacy of quantum learning over classical learning if there exists a class of distributions $\mathcal{D}_n$ for which there exists $d ,\epsilon$ such that $\mathcal{D}_n$ is $\left( d ,\epsilon, \BQP \right)$-Learnable, but $\mathcal{D}_n$ is not $\left( d ,\epsilon, \BPP \right)$-Learnable.
\end{definition}

As mentioned above, a typical choice for $d$ would be $d = \TV$, but one could imagine weaker definitions by using weaker cost functions. One may also be more restrictive and look for a demonstration of learning superiority by a class which was efficiently $\IQP$-Learnable, but not $\BPP$-Learnable. This case may be more challenging to prove theoretically, but may be more amenable for the near term, precisely the original motivation for quantum supremacy, and, indeed, implies \defref{defn:superiority_supp}.

\bibliographystyleS{naturemag}


\begin{thebibliography}{10}
\expandafter\ifx\csname url\endcsname\relax
  \def\url#1{\texttt{#1}}\fi
\expandafter\ifx\csname urlprefix\endcsname\relax\def\urlprefix{URL }\fi
\providecommand{\bibinfo}[2]{#2}
\providecommand{\eprint}[2][]{\url{#2}}

\bibitem{preskill_quantum_2018}
\bibinfo{author}{Preskill, J.}
\newblock \bibinfo{title}{Quantum {Computing} in the {NISQ} era and beyond}.
\newblock \emph{\bibinfo{journal}{Quantum}} \textbf{\bibinfo{volume}{2}},
  \bibinfo{pages}{79} (\bibinfo{year}{2018}).
\newblock \urlprefix\url{https://doi.org/10.22331/q-2018-08-06-79}.

\bibitem{shor_polynomial-time_1997}
\bibinfo{author}{Shor, P.}
\newblock \bibinfo{title}{Polynomial-{Time} {Algorithms} for {Prime}
  {Factorization} and {Discrete} {Logarithms} on a {Quantum} {Computer}}.
\newblock \emph{\bibinfo{journal}{SIAM Journal on Computing}}
  \textbf{\bibinfo{volume}{26}}, \bibinfo{pages}{1484--1509}
  (\bibinfo{year}{1997}).
\newblock \urlprefix\url{https://doi.org/10.1137/S0097539795293172}.

\bibitem{harrow_quantum_2009}
\bibinfo{author}{Harrow, A.~W.}, \bibinfo{author}{Hassidim, A.} \&
  \bibinfo{author}{Lloyd, S.}
\newblock \bibinfo{title}{Quantum {Algorithm} for {Linear} {Systems} of
  {Equations}}.
\newblock \emph{\bibinfo{journal}{Phys. Rev. Lett.}}
  \textbf{\bibinfo{volume}{103}}, \bibinfo{pages}{150502}
  (\bibinfo{year}{2009}).
\newblock
  \urlprefix\url{https://link.aps.org/doi/10.1103/PhysRevLett.103.150502}.

\bibitem{bremner_classical_2011}
\bibinfo{author}{Bremner, M.~J.}, \bibinfo{author}{Jozsa, R.} \&
  \bibinfo{author}{Shepherd, D.~J.}
\newblock \bibinfo{title}{Classical simulation of commuting quantum
  computations implies collapse of the polynomial hierarchy}.
\newblock \emph{\bibinfo{journal}{Proceedings of the Royal Society of London A:
  Mathematical, Physical and Engineering Sciences}}
  \textbf{\bibinfo{volume}{467}}, \bibinfo{pages}{459--472}
  (\bibinfo{year}{2011}).
\newblock
  \urlprefix\url{http://rspa.royalsocietypublishing.org/content/467/2126/459}.

\bibitem{gao_quantum_2017}
\bibinfo{author}{Gao, X.}, \bibinfo{author}{Wang, S.-T.} \&
  \bibinfo{author}{Duan, L.-M.}
\newblock \bibinfo{title}{Quantum {Supremacy} for {Simulating} a
  {Translation}-{Invariant} {Ising} {Spin} {Model}}.
\newblock \emph{\bibinfo{journal}{Phys. Rev. Lett.}}
  \textbf{\bibinfo{volume}{118}}, \bibinfo{pages}{040502}
  (\bibinfo{year}{2017}).
\newblock
  \urlprefix\url{https://link.aps.org/doi/10.1103/PhysRevLett.118.040502}.

\bibitem{bremner_average-case_2016}
\bibinfo{author}{Bremner, M.~J.}, \bibinfo{author}{Montanaro, A.} \&
  \bibinfo{author}{Shepherd, D.~J.}
\newblock \bibinfo{title}{Average-{Case} {Complexity} {Versus} {Approximate}
  {Simulation} of {Commuting} {Quantum} {Computations}}.
\newblock \emph{\bibinfo{journal}{Phys. Rev. Lett.}}
  \textbf{\bibinfo{volume}{117}}, \bibinfo{pages}{080501}
  (\bibinfo{year}{2016}).
\newblock
  \urlprefix\url{https://link.aps.org/doi/10.1103/PhysRevLett.117.080501}.

\bibitem{aaronson_computational_2013}
\bibinfo{author}{Aaronson, S.} \& \bibinfo{author}{Arkhipov, A.}
\newblock \bibinfo{title}{The {Computational} {Complexity} of {Linear}
  {Optics}}.
\newblock \emph{\bibinfo{journal}{Theory of Computing}}
  \textbf{\bibinfo{volume}{9}}, \bibinfo{pages}{143--252}
  (\bibinfo{year}{2013}).
\newblock \urlprefix\url{http://www.theoryofcomputing.org/articles/v009a004}.

\bibitem{farhi_quantum_2016}
\bibinfo{author}{Farhi, E.} \& \bibinfo{author}{Harrow, A.~W.}
\newblock \bibinfo{title}{Quantum {Supremacy} through the {Quantum}
  {Approximate} {Optimization} {Algorithm}}.
\newblock \emph{\bibinfo{journal}{arXiv:1602.07674 [quant-ph]}}
  (\bibinfo{year}{2016}).
\newblock \urlprefix\url{http://arxiv.org/abs/1602.07674}.
\newblock \bibinfo{note}{ArXiv: 1602.07674}.

\bibitem{boixo_characterizing_2018}
\bibinfo{author}{Boixo, S.} \emph{et~al.}
\newblock \bibinfo{title}{Characterizing quantum supremacy in near-term
  devices}.
\newblock \emph{\bibinfo{journal}{Nature Physics}}
  \textbf{\bibinfo{volume}{14}}, \bibinfo{pages}{595--600}
  (\bibinfo{year}{2018}).
\newblock \urlprefix\url{https://doi.org/10.1038/s41567-018-0124-x}.

\bibitem{arute_quantum_2019}
\bibinfo{author}{Arute, F.} \emph{et~al.}
\newblock \bibinfo{title}{Quantum supremacy using a programmable
  superconducting processor}.
\newblock \emph{\bibinfo{journal}{Nature}} \textbf{\bibinfo{volume}{574}},
  \bibinfo{pages}{505--510} (\bibinfo{year}{2019}).
\newblock \urlprefix\url{https://www.nature.com/articles/s41586-019-1666-5}.

\bibitem{maron_automatic_1961}
\bibinfo{author}{Maron, M.~E.}
\newblock \bibinfo{title}{Automatic {Indexing}: {An} {Experimental} {Inquiry}}.
\newblock \emph{\bibinfo{journal}{J. ACM}} \textbf{\bibinfo{volume}{8}},
  \bibinfo{pages}{404--417} (\bibinfo{year}{1961}).
\newblock \urlprefix\url{http://doi.acm.org/10.1145/321075.321084}.

\bibitem{goodfellow_generative_2014}
\bibinfo{author}{Goodfellow, I.~J.} \emph{et~al.}
\newblock \bibinfo{title}{Generative {Adversarial} {Networks}}.
\newblock \emph{\bibinfo{journal}{arXiv:1406.2661 [cs, stat]}}
  (\bibinfo{year}{2014}).
\newblock \urlprefix\url{http://arxiv.org/abs/1406.2661}.
\newblock \bibinfo{note}{ArXiv: 1406.2661}.

\bibitem{cheng_information_2017}
\bibinfo{author}{Cheng, S.}, \bibinfo{author}{Chen, J.} \&
  \bibinfo{author}{Wang, L.}
\newblock \bibinfo{title}{Information {Perspective} to {Probabilistic}
  {Modeling}: {Boltzmann} {Machines} versus {Born} {Machines}}.
\newblock \emph{\bibinfo{journal}{arXiv:1712.04144 [cond-mat, physics:physics,
  physics:quant-ph, stat]}}  (\bibinfo{year}{2017}).
\newblock \urlprefix\url{http://arxiv.org/abs/1712.04144}.
\newblock \bibinfo{note}{ArXiv: 1712.04144}.

\bibitem{liu_differentiable_2018}
\bibinfo{author}{Liu, J.-G.} \& \bibinfo{author}{Wang, L.}
\newblock \bibinfo{title}{Differentiable learning of quantum circuit {Born}
  machines}.
\newblock \emph{\bibinfo{journal}{Phys. Rev. A}} \textbf{\bibinfo{volume}{98}},
  \bibinfo{pages}{062324} (\bibinfo{year}{2018}).
\newblock \urlprefix\url{https://link.aps.org/doi/10.1103/PhysRevA.98.062324}.

\bibitem{benedetti_generative_2019}
\bibinfo{author}{Benedetti, M.} \emph{et~al.}
\newblock \bibinfo{title}{A generative modeling approach for benchmarking and
  training shallow quantum circuits}.
\newblock \emph{\bibinfo{journal}{npj Quantum Inf}}
  \textbf{\bibinfo{volume}{5}}, \bibinfo{pages}{1--9} (\bibinfo{year}{2019}).
\newblock \urlprefix\url{https://www.nature.com/articles/s41534-019-0157-8}.

\bibitem{du_expressive_2018}
\bibinfo{author}{Du, Y.}, \bibinfo{author}{Hsieh, M.-H.}, \bibinfo{author}{Liu,
  T.} \& \bibinfo{author}{Tao, D.}
\newblock \bibinfo{title}{The {Expressive} {Power} of {Parameterized} {Quantum}
  {Circuits}}.
\newblock \emph{\bibinfo{journal}{arXiv:1810.11922 [quant-ph]}}
  (\bibinfo{year}{2018}).
\newblock \urlprefix\url{http://arxiv.org/abs/1810.11922}.
\newblock \bibinfo{note}{ArXiv: 1810.11922}.

\bibitem{zeng_learning_2019}
\bibinfo{author}{Zeng, J.}, \bibinfo{author}{Wu, Y.}, \bibinfo{author}{Liu,
  J.-G.}, \bibinfo{author}{Wang, L.} \& \bibinfo{author}{Hu, J.}
\newblock \bibinfo{title}{Learning and inference on generative adversarial
  quantum circuits}.
\newblock \emph{\bibinfo{journal}{Phys. Rev. A}} \textbf{\bibinfo{volume}{99}},
  \bibinfo{pages}{052306} (\bibinfo{year}{2019}).
\newblock \urlprefix\url{https://link.aps.org/doi/10.1103/PhysRevA.99.052306}.

\bibitem{romero_variational_2019}
\bibinfo{author}{Romero, J.} \& \bibinfo{author}{Aspuru-Guzik, A.}
\newblock \bibinfo{title}{Variational quantum generators: {Generative}
  adversarial quantum machine learning for continuous distributions}.
\newblock \emph{\bibinfo{journal}{arXiv:1901.00848 [quant-ph]}}
  (\bibinfo{year}{2019}).
\newblock \urlprefix\url{http://arxiv.org/abs/1901.00848}.
\newblock \bibinfo{note}{ArXiv: 1901.00848}.

\bibitem{benedetti_parameterized_2019}
\bibinfo{author}{Benedetti, M.}, \bibinfo{author}{Lloyd, E.},
  \bibinfo{author}{Sack, S.} \& \bibinfo{author}{Fiorentini, M.}
\newblock \bibinfo{title}{Parameterized quantum circuits as machine learning
  models}.
\newblock \emph{\bibinfo{journal}{Quantum Sci. Technol.}}
  \textbf{\bibinfo{volume}{4}}, \bibinfo{pages}{043001} (\bibinfo{year}{2019}).
\newblock \urlprefix\url{https://doi.org/10.1088%2F2058-9565%2Fab4eb5}.

\bibitem{tang_quantum-inspired_2018}
\bibinfo{author}{Tang, E.}
\newblock \bibinfo{title}{Quantum-inspired classical algorithms for principal
  component analysis and supervised clustering}.
\newblock \emph{\bibinfo{journal}{arXiv:1811.00414 [quant-ph]}}
  (\bibinfo{year}{2018}).
\newblock \urlprefix\url{http://arxiv.org/abs/1811.00414}.
\newblock \bibinfo{note}{ArXiv: 1811.00414}.

\bibitem{tang_quantum-inspired_2018-1}
\bibinfo{author}{Tang, E.}
\newblock \bibinfo{title}{A quantum-inspired classical algorithm for
  recommendation systems}.
\newblock \emph{\bibinfo{journal}{arXiv:1807.04271 [quant-ph]}}
  (\bibinfo{year}{2018}).
\newblock \urlprefix\url{http://arxiv.org/abs/1807.04271}.
\newblock \bibinfo{note}{ArXiv: 1807.04271}.

\bibitem{andoni_solving_2018}
\bibinfo{author}{Andoni, A.}, \bibinfo{author}{Krauthgamer, R.} \&
  \bibinfo{author}{Pogrow, Y.}
\newblock \bibinfo{title}{On {Solving} {Linear} {Systems} in {Sublinear}
  {Time}}.
\newblock \emph{\bibinfo{journal}{arXiv:1809.02995 [cs]}}
  (\bibinfo{year}{2018}).
\newblock \urlprefix\url{http://arxiv.org/abs/1809.02995}.
\newblock \bibinfo{note}{ArXiv: 1809.02995}.

\bibitem{chia_quantum-inspired_2018}
\bibinfo{author}{Chia, N.-H.}, \bibinfo{author}{Lin, H.-H.} \&
  \bibinfo{author}{Wang, C.}
\newblock \bibinfo{title}{Quantum-inspired sublinear classical algorithms for
  solving low-rank linear systems}.
\newblock \emph{\bibinfo{journal}{arXiv:1811.04852 [quant-ph]}}
  (\bibinfo{year}{2018}).
\newblock \urlprefix\url{http://arxiv.org/abs/1811.04852}.
\newblock \bibinfo{note}{ArXiv: 1811.04852}.

\bibitem{gilyen_quantum-inspired_2018}
\bibinfo{author}{Gily{\'e}n, A.}, \bibinfo{author}{Lloyd, S.} \&
  \bibinfo{author}{Tang, E.}
\newblock \bibinfo{title}{Quantum-inspired low-rank stochastic regression with
  logarithmic dependence on the dimension}.
\newblock \emph{\bibinfo{journal}{arXiv:1811.04909 [quant-ph]}}
  (\bibinfo{year}{2018}).
\newblock \urlprefix\url{http://arxiv.org/abs/1811.04909}.
\newblock \bibinfo{note}{ArXiv: 1811.04909}.

\bibitem{kearns_learnability_1994}
\bibinfo{author}{Kearns, M.} \emph{et~al.}
\newblock \bibinfo{title}{On the {Learnability} of {Discrete} {Distributions}}.
\newblock In \emph{\bibinfo{booktitle}{Proceedings of the {Twenty}-sixth
  {Annual} {ACM} {Symposium} on {Theory} of {Computing}}}, {STOC} '94,
  \bibinfo{pages}{273--282} (\bibinfo{publisher}{ACM}, \bibinfo{address}{New
  York, NY, USA}, \bibinfo{year}{1994}).
\newblock \urlprefix\url{http://doi.acm.org/10.1145/195058.195155}.
\newblock \bibinfo{note}{Event-place: Montreal, Quebec, Canada}.

\bibitem{shepherd_temporally_2009}
\bibinfo{author}{Shepherd, D.} \& \bibinfo{author}{Bremner, M.~J.}
\newblock \bibinfo{title}{Temporally unstructured quantum computation}.
\newblock \emph{\bibinfo{journal}{PROC R SOC A}}  (\bibinfo{year}{2009}).
\newblock
  \urlprefix\url{http://rspa.royalsocietypublishing.org/content/early/2009/02/18/rspa.2008.0443.abstract}.

\bibitem{farhi_quantum_2014}
\bibinfo{author}{Farhi, E.}, \bibinfo{author}{Goldstone, J.} \&
  \bibinfo{author}{Gutmann, S.}
\newblock \bibinfo{title}{A {Quantum} {Approximate} {Optimization}
  {Algorithm}}.
\newblock \emph{\bibinfo{journal}{arXiv:1411.4028 [quant-ph]}}
  (\bibinfo{year}{2014}).
\newblock \urlprefix\url{http://arxiv.org/abs/1411.4028}.
\newblock \bibinfo{note}{ArXiv: 1411.4028}.

\bibitem{farhi_quantum_2000}
\bibinfo{author}{Farhi, E.}, \bibinfo{author}{Goldstone, J.},
  \bibinfo{author}{Gutmann, S.} \& \bibinfo{author}{Sipser, M.}
\newblock \bibinfo{title}{Quantum {Computation} by {Adiabatic} {Evolution}}.
\newblock \emph{\bibinfo{journal}{arXiv:quant-ph/0001106}}
  (\bibinfo{year}{2000}).
\newblock \urlprefix\url{http://arxiv.org/abs/quant-ph/0001106}.
\newblock \bibinfo{note}{ArXiv: quant-ph/0001106}.

\bibitem{bremner_achieving_2017}
\bibinfo{author}{Bremner, M.~J.}, \bibinfo{author}{Montanaro, A.} \&
  \bibinfo{author}{Shepherd, D.~J.}
\newblock \bibinfo{title}{Achieving quantum supremacy with sparse and noisy
  commuting quantum computations}.
\newblock \emph{\bibinfo{journal}{Quantum}} \textbf{\bibinfo{volume}{1}},
  \bibinfo{pages}{8} (\bibinfo{year}{2017}).
\newblock \urlprefix\url{https://doi.org/10.22331/q-2017-04-25-8}.

\bibitem{fujii_commuting_2017}
\bibinfo{author}{Fujii, K.} \& \bibinfo{author}{Morimae, T.}
\newblock \bibinfo{title}{Commuting quantum circuits and complexity of {Ising}
  partition functions}.
\newblock \emph{\bibinfo{journal}{New Journal of Physics}}
  \textbf{\bibinfo{volume}{19}}, \bibinfo{pages}{033003}
  (\bibinfo{year}{2017}).
\newblock
  \urlprefix\url{http://stacks.iop.org/1367-2630/19/i=3/a=033003?key=crossref.cefbe34cf11242886552ceea447a4526}.

\bibitem{leyton-ortega_robust_2019}
\bibinfo{author}{Leyton-Ortega, V.}, \bibinfo{author}{Perdomo-Ortiz, A.} \&
  \bibinfo{author}{Perdomo, O.}
\newblock \bibinfo{title}{Robust {Implementation} of {Generative} {Modeling}
  with {Parametrized} {Quantum} {Circuits}}.
\newblock \emph{\bibinfo{journal}{arXiv:1901.08047 [quant-ph]}}
  (\bibinfo{year}{2019}).
\newblock \urlprefix\url{http://arxiv.org/abs/1901.08047}.
\newblock \bibinfo{note}{ArXiv: 1901.08047}.

\bibitem{hamilton_generative_2019}
\bibinfo{author}{Hamilton, K.~E.}, \bibinfo{author}{Dumitrescu, E.~F.} \&
  \bibinfo{author}{Pooser, R.~C.}
\newblock \bibinfo{title}{Generative model benchmarks for superconducting
  qubits}.
\newblock \emph{\bibinfo{journal}{Phys. Rev. A}} \textbf{\bibinfo{volume}{99}},
  \bibinfo{pages}{062323} (\bibinfo{year}{2019}).
\newblock \urlprefix\url{https://link.aps.org/doi/10.1103/PhysRevA.99.062323}.

\bibitem{lloyd_quantum_2018}
\bibinfo{author}{Lloyd, S.} \& \bibinfo{author}{Weedbrook, C.}
\newblock \bibinfo{title}{Quantum {Generative} {Adversarial} {Learning}}.
\newblock \emph{\bibinfo{journal}{Phys. Rev. Lett.}}
  \textbf{\bibinfo{volume}{121}}, \bibinfo{pages}{040502}
  (\bibinfo{year}{2018}).
\newblock
  \urlprefix\url{https://link.aps.org/doi/10.1103/PhysRevLett.121.040502}.

\bibitem{dallaire-demers_quantum_2018}
\bibinfo{author}{Dallaire-Demers, P.-L.} \& \bibinfo{author}{Killoran, N.}
\newblock \bibinfo{title}{Quantum generative adversarial networks}.
\newblock \emph{\bibinfo{journal}{Phys. Rev. A}} \textbf{\bibinfo{volume}{98}},
  \bibinfo{pages}{012324} (\bibinfo{year}{2018}).
\newblock \urlprefix\url{https://link.aps.org/doi/10.1103/PhysRevA.98.012324}.

\bibitem{borgwardt_integrating_2006}
\bibinfo{author}{Borgwardt, K.~M.} \emph{et~al.}
\newblock \bibinfo{title}{Integrating structured biological data by {Kernel}
  {Maximum} {Mean} {Discrepancy}}.
\newblock \emph{\bibinfo{journal}{Bioinformatics}}
  \textbf{\bibinfo{volume}{22}}, \bibinfo{pages}{e49--e57}
  (\bibinfo{year}{2006}).
\newblock \urlprefix\url{http://dx.doi.org/10.1093/bioinformatics/btl242}.

\bibitem{gretton_kernel_2007}
\bibinfo{author}{Gretton, A.}, \bibinfo{author}{Borgwardt, K.~M.},
  \bibinfo{author}{Rasch, M.}, \bibinfo{author}{Sch{\"o}lkopf, B.} \&
  \bibinfo{author}{Smola, A.~J.}
\newblock \bibinfo{title}{A {Kernel} {Method} for the
  {Two}-{Sample}-{Problem}}.
\newblock In \bibinfo{editor}{Sch{\"o}lkopf, B.}, \bibinfo{editor}{Platt,
  J.~C.} \& \bibinfo{editor}{Hoffman, T.} (eds.)
  \emph{\bibinfo{booktitle}{Advances in {Neural} {Information} {Processing}
  {Systems} 19}}, \bibinfo{pages}{513--520} (\bibinfo{publisher}{MIT Press},
  \bibinfo{year}{2007}).
\newblock
  \urlprefix\url{http://papers.nips.cc/paper/3110-a-kernel-method-for-the-two-sample-problem.pdf}.

\bibitem{havlicek_supervised_2019}
\bibinfo{author}{Havl{\'i}{\v c}ek, V.} \emph{et~al.}
\newblock \bibinfo{title}{Supervised learning with quantum-enhanced feature
  spaces}.
\newblock \emph{\bibinfo{journal}{Nature}} \textbf{\bibinfo{volume}{567}},
  \bibinfo{pages}{209--212} (\bibinfo{year}{2019}).
\newblock \urlprefix\url{https://www.nature.com/articles/s41586-019-0980-2}.

\bibitem{schuld_supervised_2018}
\bibinfo{author}{Schuld, M.} \& \bibinfo{author}{Petruccione, F.}
\newblock \emph{\bibinfo{title}{Supervised {Learning} with {Quantum}
  {Computers}}}.
\newblock Quantum {Science} and {Technology} (\bibinfo{publisher}{Springer
  International Publishing}, \bibinfo{year}{2018}).
\newblock \urlprefix\url{https://www.springer.com/us/book/9783319964232}.

\bibitem{mitarai_quantum_2018}
\bibinfo{author}{Mitarai, K.}, \bibinfo{author}{Negoro, M.},
  \bibinfo{author}{Kitagawa, M.} \& \bibinfo{author}{Fujii, K.}
\newblock \bibinfo{title}{Quantum circuit learning}.
\newblock \emph{\bibinfo{journal}{Phys. Rev. A}} \textbf{\bibinfo{volume}{98}},
  \bibinfo{pages}{032309} (\bibinfo{year}{2018}).
\newblock \urlprefix\url{https://link.aps.org/doi/10.1103/PhysRevA.98.032309}.

\bibitem{schuld_evaluating_2019}
\bibinfo{author}{Schuld, M.}, \bibinfo{author}{Bergholm, V.},
  \bibinfo{author}{Gogolin, C.}, \bibinfo{author}{Izaac, J.} \&
  \bibinfo{author}{Killoran, N.}
\newblock \bibinfo{title}{Evaluating analytic gradients on quantum hardware}.
\newblock \emph{\bibinfo{journal}{Phys. Rev. A}} \textbf{\bibinfo{volume}{99}},
  \bibinfo{pages}{032331} (\bibinfo{year}{2019}).
\newblock \urlprefix\url{https://link.aps.org/doi/10.1103/PhysRevA.99.032331}.

\bibitem{liu_kernelized_2016}
\bibinfo{author}{Liu, Q.}, \bibinfo{author}{Lee, J.~D.} \&
  \bibinfo{author}{Jordan, M.}
\newblock \bibinfo{title}{A {Kernelized} {Stein} {Discrepancy} for
  {Goodness}-of-fit {Tests}}.
\newblock In \emph{\bibinfo{booktitle}{Proceedings of the 33rd {International}
  {Conference} on {International} {Conference} on {Machine} {Learning} -
  {Volume} 48}}, {ICML}'16, \bibinfo{pages}{276--284}
  (\bibinfo{publisher}{JMLR.org}, \bibinfo{address}{New York, NY, USA},
  \bibinfo{year}{2016}).
\newblock \urlprefix\url{http://dl.acm.org/citation.cfm?id=3045390.3045421}.

\bibitem{stein_bound_1972}
\bibinfo{author}{Stein, C.}
\newblock \bibinfo{title}{A bound for the error in the normal approximation to
  the distribution of a sum of dependent random variables}.
\newblock In \emph{\bibinfo{booktitle}{Proceedings of the {Sixth} {Berkeley}
  {Symposium} on {Mathematical} {Statistics} and {Probability}, {Volume} 2:
  {Probability} {Theory}}}, \bibinfo{pages}{583--602}
  (\bibinfo{publisher}{University of California Press},
  \bibinfo{address}{Berkeley, Calif.}, \bibinfo{year}{1972}).
\newblock \urlprefix\url{https://projecteuclid.org/euclid.bsmsp/1200514239}.

\bibitem{yang_goodness--fit_2018}
\bibinfo{author}{Yang, J.}, \bibinfo{author}{Liu, Q.}, \bibinfo{author}{Rao,
  V.} \& \bibinfo{author}{Neville, J.}
\newblock \bibinfo{title}{Goodness-of-{Fit} {Testing} for {Discrete}
  {Distributions} via {Stein} {Discrepancy}}.
\newblock In \bibinfo{editor}{Dy, J.} \& \bibinfo{editor}{Krause, A.} (eds.)
  \emph{\bibinfo{booktitle}{Proceedings of the 35th {International}
  {Conference} on {Machine} {Learning}}}, vol.~\bibinfo{volume}{80} of
  \emph{\bibinfo{series}{Proceedings of {Machine} {Learning} {Research}}},
  \bibinfo{pages}{5561--5570} (\bibinfo{publisher}{PMLR},
  \bibinfo{address}{Stockholmsm{\"a}ssan, Stockholm Sweden},
  \bibinfo{year}{2018}).
\newblock \urlprefix\url{http://proceedings.mlr.press/v80/yang18c.html}.

\bibitem{gorham_measuring_2015}
\bibinfo{author}{Gorham, J.} \& \bibinfo{author}{Mackey, L.}
\newblock \bibinfo{title}{Measuring {Sample} {Quality} with {Stein}'s
  {Method}}.
\newblock In \bibinfo{editor}{Cortes, C.}, \bibinfo{editor}{Lawrence, N.~D.},
  \bibinfo{editor}{Lee, D.~D.}, \bibinfo{editor}{Sugiyama, M.} \&
  \bibinfo{editor}{Garnett, R.} (eds.) \emph{\bibinfo{booktitle}{Advances in
  {Neural} {Information} {Processing} {Systems} 28}}, \bibinfo{pages}{226--234}
  (\bibinfo{publisher}{Curran Associates, Inc.}, \bibinfo{year}{2015}).
\newblock
  \urlprefix\url{http://papers.nips.cc/paper/5768-measuring-sample-quality-with-steins-method.pdf}.

\bibitem{li_gradient_2018}
\bibinfo{author}{Li, Y.} \& \bibinfo{author}{Turner, R.~E.}
\newblock \bibinfo{title}{Gradient {Estimators} for {Implicit} {Models}}.
\newblock In \emph{\bibinfo{booktitle}{International {Conference} on {Learning}
  {Representations}}} (\bibinfo{year}{2018}).
\newblock \urlprefix\url{https://openreview.net/forum?id=SJi9WOeRb}.

\bibitem{shi_spectral_2018}
\bibinfo{author}{Shi, J.}, \bibinfo{author}{Sun, S.} \& \bibinfo{author}{Zhu,
  J.}
\newblock \bibinfo{title}{A {Spectral} {Approach} to {Gradient} {Estimation}
  for {Implicit} {Distributions}}.
\newblock \emph{\bibinfo{journal}{arXiv:1806.02925 [cs, stat]}}
  (\bibinfo{year}{2018}).
\newblock \urlprefix\url{http://arxiv.org/abs/1806.02925}.
\newblock \bibinfo{note}{ArXiv: 1806.02925}.

\bibitem{nystrom_uber_1930}
\bibinfo{author}{Nystr{\"o}m, E.~J.}
\newblock \bibinfo{title}{{\"U}ber {Die} {Praktische} {Aufl{\"o}sung} von
  {Integralgleichungen} mit {Anwendungen} auf {Randwertaufgaben}}.
\newblock \emph{\bibinfo{journal}{Acta Math.}} \textbf{\bibinfo{volume}{54}},
  \bibinfo{pages}{185--204} (\bibinfo{year}{1930}).
\newblock \urlprefix\url{https://doi.org/10.1007/BF02547521}.

\bibitem{ramdas_wasserstein_2015}
\bibinfo{author}{Ramdas, A.}, \bibinfo{author}{Garcia, N.} \&
  \bibinfo{author}{Cuturi, M.}
\newblock \bibinfo{title}{On {Wasserstein} {Two} {Sample} {Testing} and
  {Related} {Families} of {Nonparametric} {Tests}}.
\newblock \emph{\bibinfo{journal}{arXiv:1509.02237 [math, stat]}}
  (\bibinfo{year}{2015}).
\newblock \urlprefix\url{http://arxiv.org/abs/1509.02237}.
\newblock \bibinfo{note}{ArXiv: 1509.02237}.

\bibitem{genevay_learning_2018}
\bibinfo{author}{Genevay, A.}, \bibinfo{author}{Peyre, G.} \&
  \bibinfo{author}{Cuturi, M.}
\newblock \bibinfo{title}{Learning {Generative} {Models} with {Sinkhorn}
  {Divergences}}.
\newblock In \bibinfo{editor}{Storkey, A.} \& \bibinfo{editor}{Perez-Cruz, F.}
  (eds.) \emph{\bibinfo{booktitle}{Proceedings of the {Twenty}-{First}
  {International} {Conference} on {Artificial} {Intelligence} and
  {Statistics}}}, vol.~\bibinfo{volume}{84} of
  \emph{\bibinfo{series}{Proceedings of {Machine} {Learning} {Research}}},
  \bibinfo{pages}{1608--1617} (\bibinfo{publisher}{PMLR},
  \bibinfo{address}{Playa Blanca, Lanzarote, Canary Islands},
  \bibinfo{year}{2018}).
\newblock \urlprefix\url{http://proceedings.mlr.press/v84/genevay18a.html}.

\bibitem{feydy_interpolating_2019}
\bibinfo{author}{Feydy, J.} \emph{et~al.}
\newblock \bibinfo{title}{Interpolating between {Optimal} {Transport} and {MMD}
  using {Sinkhorn} {Divergences}}.
\newblock In \bibinfo{editor}{Chaudhuri, K.} \& \bibinfo{editor}{Sugiyama, M.}
  (eds.) \emph{\bibinfo{booktitle}{Proceedings of {Machine} {Learning}
  {Research}}}, vol.~\bibinfo{volume}{89} of \emph{\bibinfo{series}{Proceedings
  of {Machine} {Learning} {Research}}}, \bibinfo{pages}{2681--2690}
  (\bibinfo{publisher}{PMLR}, \bibinfo{year}{2019}).
\newblock \urlprefix\url{http://proceedings.mlr.press/v89/feydy19a.html}.

\bibitem{villani_optimal_2009}
\bibinfo{author}{Villani, C.}
\newblock \emph{\bibinfo{title}{Optimal {Transport}: {Old} and {New}}}.
\newblock Grundlehren der mathematischen {Wissenschaften}
  (\bibinfo{publisher}{Springer-Verlag}, \bibinfo{address}{Berlin Heidelberg},
  \bibinfo{year}{2009}).
\newblock \urlprefix\url{//www.springer.com/gb/book/9783540710493}.

\bibitem{arjovsky_wasserstein_2017}
\bibinfo{author}{Arjovsky, M.}, \bibinfo{author}{Chintala, S.} \&
  \bibinfo{author}{Bottou, L.}
\newblock \bibinfo{title}{Wasserstein {GAN}}.
\newblock \emph{\bibinfo{journal}{arXiv:1701.07875 [cs, stat]}}
  (\bibinfo{year}{2017}).
\newblock \urlprefix\url{http://arxiv.org/abs/1701.07875}.
\newblock \bibinfo{note}{ArXiv: 1701.07875}.

\bibitem{dudley_speed_1969}
\bibinfo{author}{Dudley, R.~M.}
\newblock \bibinfo{title}{The {Speed} of {Mean} {Glivenko}-{Cantelli}
  {Convergence}}.
\newblock \emph{\bibinfo{journal}{Ann. Math. Statist.}}
  \textbf{\bibinfo{volume}{40}}, \bibinfo{pages}{40--50}
  (\bibinfo{year}{1969}).
\newblock \urlprefix\url{https://doi.org/10.1214/aoms/1177697802}.

\bibitem{genevay_sample_2018}
\bibinfo{author}{Genevay, A.}, \bibinfo{author}{Chizat, L.},
  \bibinfo{author}{Bach, F.}, \bibinfo{author}{Cuturi, M.} \&
  \bibinfo{author}{Peyr{\'e}, G.}
\newblock \bibinfo{title}{Sample {Complexity} of {Sinkhorn} divergences}.
\newblock \emph{\bibinfo{journal}{arXiv:1810.02733 [math, stat]}}
  (\bibinfo{year}{2018}).
\newblock \urlprefix\url{http://arxiv.org/abs/1810.02733}.
\newblock \bibinfo{note}{ArXiv: 1810.02733}.

\bibitem{sriperumbudur_integral_2009}
\bibinfo{author}{Sriperumbudur, B.~K.}, \bibinfo{author}{Fukumizu, K.},
  \bibinfo{author}{Gretton, A.}, \bibinfo{author}{Sch{\"o}lkopf, B.} \&
  \bibinfo{author}{Lanckriet, G. R.~G.}
\newblock \bibinfo{title}{On integral probability metrics, phi-divergences and
  binary classification}.
\newblock \emph{\bibinfo{journal}{arXiv:0901.2698 [cs, math]}}
  (\bibinfo{year}{2009}).
\newblock \urlprefix\url{http://arxiv.org/abs/0901.2698}.

\bibitem{gibbs_choosing_2002}
\bibinfo{author}{Gibbs, A.~L.} \& \bibinfo{author}{Su, F.~E.}
\newblock \bibinfo{title}{On {Choosing} and {Bounding} {Probability}
  {Metrics}}.
\newblock \emph{\bibinfo{journal}{International Statistical Review}}
  \textbf{\bibinfo{volume}{70}}, \bibinfo{pages}{419--435}
  (\bibinfo{year}{2002}).
\newblock
  \urlprefix\url{https://onlinelibrary.wiley.com/doi/abs/10.1111/j.1751-5823.2002.tb00178.x}.

\bibitem{smith_practical_2016}
\bibinfo{author}{Smith, R.~S.}, \bibinfo{author}{Curtis, M.~J.} \&
  \bibinfo{author}{Zeng, W.~J.}
\newblock \bibinfo{title}{A {Practical} {Quantum} {Instruction} {Set}
  {Architecture}}.
\newblock \emph{\bibinfo{journal}{arXiv:1608.03355 [quant-ph]}}
  (\bibinfo{year}{2016}).
\newblock \urlprefix\url{http://arxiv.org/abs/1608.03355}.
\newblock \bibinfo{note}{ArXiv: 1608.03355}.

\bibitem{arunachalam_survey_2017}
\bibinfo{author}{Arunachalam, S.} \& \bibinfo{author}{de~Wolf, R.}
\newblock \bibinfo{title}{A {Survey} of {Quantum} {Learning} {Theory}}.
\newblock \emph{\bibinfo{journal}{arXiv:1701.06806 [quant-ph]}}
  (\bibinfo{year}{2017}).
\newblock \urlprefix\url{http://arxiv.org/abs/1701.06806}.
\newblock \bibinfo{note}{ArXiv: 1701.06806}.

\bibitem{arunachalam_quantum_2019}
\bibinfo{author}{Arunachalam, S.}, \bibinfo{author}{Grilo, A.~B.} \&
  \bibinfo{author}{Sundaram, A.}
\newblock \bibinfo{title}{Quantum hardness of learning shallow classical
  circuits}.
\newblock \emph{\bibinfo{journal}{arXiv:1903.02840 [quant-ph]}}
  (\bibinfo{year}{2019}).
\newblock \urlprefix\url{http://arxiv.org/abs/1903.02840}.
\newblock \bibinfo{note}{ArXiv: 1903.02840}.

\bibitem{khatri_quantum-assisted_2019}
\bibinfo{author}{Khatri, S.} \emph{et~al.}
\newblock \bibinfo{title}{Quantum-assisted quantum compiling}.
\newblock \emph{\bibinfo{journal}{Quantum}} \textbf{\bibinfo{volume}{3}},
  \bibinfo{pages}{140} (\bibinfo{year}{2019}).
\newblock \urlprefix\url{https://doi.org/10.22331/q-2019-05-13-140}.

\bibitem{jones_quantum_2018}
\bibinfo{author}{Jones, T.} \& \bibinfo{author}{Benjamin, S.~C.}
\newblock \bibinfo{title}{Quantum compilation and circuit optimisation via
  energy dissipation}.
\newblock \emph{\bibinfo{journal}{arXiv:1811.03147 [quant-ph]}}
  (\bibinfo{year}{2018}).
\newblock \urlprefix\url{http://arxiv.org/abs/1811.03147}.
\newblock \bibinfo{note}{ArXiv: 1811.03147}.

\bibitem{gao_efficient_2017}
\bibinfo{author}{Gao, X.}, \bibinfo{author}{Zhang, Z.} \&
  \bibinfo{author}{Duan, L.}
\newblock \bibinfo{title}{An efficient quantum algorithm for generative machine
  learning}.
\newblock \emph{\bibinfo{journal}{arXiv:1711.02038 [quant-ph, stat]}}
  (\bibinfo{year}{2017}).
\newblock \urlprefix\url{http://arxiv.org/abs/1711.02038}.
\newblock \bibinfo{note}{ArXiv: 1711.02038}.

\bibitem{hangleiter_sample_2019}
\bibinfo{author}{Hangleiter, D.}, \bibinfo{author}{Kliesch, M.},
  \bibinfo{author}{Eisert, J.} \& \bibinfo{author}{Gogolin, C.}
\newblock \bibinfo{title}{Sample {Complexity} of {Device}-{Independently}
  {Certified} {\textquotedblleft}{Quantum} {Supremacy}{\textquotedblright}}.
\newblock \emph{\bibinfo{journal}{Phys. Rev. Lett.}}
  \textbf{\bibinfo{volume}{122}}, \bibinfo{pages}{210502}
  (\bibinfo{year}{2019}).
\newblock
  \urlprefix\url{https://link.aps.org/doi/10.1103/PhysRevLett.122.210502}.

\bibitem{goldreich_property_1998}
\bibinfo{author}{Goldreich, O.}, \bibinfo{author}{Goldwasser, S.} \&
  \bibinfo{author}{Ron, D.}
\newblock \bibinfo{title}{Property {Testing} and {Its} {Connection} to
  {Learning} and {Approximation}}.
\newblock \emph{\bibinfo{journal}{J. ACM}} \textbf{\bibinfo{volume}{45}},
  \bibinfo{pages}{653--750} (\bibinfo{year}{1998}).
\newblock \urlprefix\url{http://doi.acm.org/10.1145/285055.285060}.

\bibitem{amin_quantum_2018}
\bibinfo{author}{Amin, M.~H.}, \bibinfo{author}{Andriyash, E.},
  \bibinfo{author}{Rolfe, J.}, \bibinfo{author}{Kulchytskyy, B.} \&
  \bibinfo{author}{Melko, R.}
\newblock \bibinfo{title}{Quantum {Boltzmann} {Machine}}.
\newblock \emph{\bibinfo{journal}{Phys. Rev. X}} \textbf{\bibinfo{volume}{8}},
  \bibinfo{pages}{021050} (\bibinfo{year}{2018}).
\newblock \urlprefix\url{https://link.aps.org/doi/10.1103/PhysRevX.8.021050}.

\bibitem{verdon_quantum_2017}
\bibinfo{author}{Verdon, G.}, \bibinfo{author}{Broughton, M.} \&
  \bibinfo{author}{Biamonte, J.}
\newblock \bibinfo{title}{A quantum algorithm to train neural networks using
  low-depth circuits}.
\newblock \emph{\bibinfo{journal}{arXiv:1712.05304 [cond-mat,
  physics:quant-ph]}}  (\bibinfo{year}{2017}).
\newblock \urlprefix\url{http://arxiv.org/abs/1712.05304}.
\newblock \bibinfo{note}{ArXiv: 1712.05304}.

\bibitem{kingma_adam:_2014}
\bibinfo{author}{Kingma, D.~P.} \& \bibinfo{author}{Ba, J.}
\newblock \bibinfo{title}{Adam: {A} {Method} for {Stochastic} {Optimization}}.
\newblock \emph{\bibinfo{journal}{arXiv:1412.6980 [cs]}}
  (\bibinfo{year}{2014}).
\newblock \urlprefix\url{http://arxiv.org/abs/1412.6980}.
\newblock \bibinfo{note}{ArXiv: 1412.6980}.


\bibitem{brian_coyle_briancoyleisingbornmachine_2020}
\bibinfo{author}{Coyle, B.}
\newblock \bibinfo{title}{{BrianCoyle}/{IsingBornMachine}: {Ising} {Born}
  {Machine}} (\bibinfo{year}{2020}).
\newblock \urlprefix\url{https://zenodo.org/record/3779865#.XqvfknVKhrk}.

\end{thebibliography}

\begin{thebibliography}{10}
\expandafter\ifx\csname url\endcsname\relax
  \def\url#1{\texttt{#1}}\fi
\expandafter\ifx\csname urlprefix\endcsname\relax\def\urlprefix{URL }\fi
\providecommand{\bibinfo}[2]{#2}
\providecommand{\eprint}[2][]{\url{#2}}

\bibitem{muller_integral_1997}
\bibinfo{author}{Müller, A.}
\newblock \bibinfo{title}{Integral {Probability} {Metrics} and {Their}
  {Generating} {Classes} of {Functions}}.
\newblock \emph{\bibinfo{journal}{Advances in Applied Probability}}
  \textbf{\bibinfo{volume}{29}}, \bibinfo{pages}{429--443}
  (\bibinfo{year}{1997}).
\newblock \urlprefix\url{http://www.jstor.org/stable/1428011}.

\bibitem{gretton_kernel_2012}
\bibinfo{author}{Gretton, A.}, \bibinfo{author}{Borgwardt, K.~M.},
  \bibinfo{author}{Rasch, M.~J.}, \bibinfo{author}{Schölkopf, B.} \&
  \bibinfo{author}{Smola, A.}
\newblock \bibinfo{title}{A {Kernel} {Two}-{Sample} {Test}}.
\newblock \emph{\bibinfo{journal}{Journal of Machine Learning Research}}
  \textbf{\bibinfo{volume}{13}}, \bibinfo{pages}{723−773}
  (\bibinfo{year}{2012}).
\newblock \urlprefix\url{http://jmlr.csail.mit.edu/papers/v13/gretton12a.html}.

\bibitem{sriperumbudur_universality_2011}
\bibinfo{author}{Sriperumbudur, B.~K.}, \bibinfo{author}{Fukumizu, K.} \&
  \bibinfo{author}{Lanckriet, G. R.~G.}
\newblock \bibinfo{title}{Universality, {Characteristic} {Kernels} and {RKHS}
  {Embedding} of {Measures}}.
\newblock \emph{\bibinfo{journal}{Journal of Machine Learning Research}}
  \textbf{\bibinfo{volume}{12}}, \bibinfo{pages}{2389--2410}
  (\bibinfo{year}{2011}).
\newblock \urlprefix\url{http://www.jmlr.org/papers/v12/sriperumbudur11a.html}.

\bibitem{sriperumbudur_hilbert_2010}
\bibinfo{author}{Sriperumbudur, B.~K.}, \bibinfo{author}{Gretton, A.},
  \bibinfo{author}{Fukumizu, K.}, \bibinfo{author}{Schölkopf, B.} \&
  \bibinfo{author}{Lanckriet, G. R.~G.}
\newblock \bibinfo{title}{Hilbert {Space} {Embeddings} and {Metrics} on
  {Probability} {Measures}}.
\newblock \emph{\bibinfo{journal}{Journal of Machine Learning Research}}
  \textbf{\bibinfo{volume}{11}}, \bibinfo{pages}{1517--1561}
  (\bibinfo{year}{2010}).
\newblock \urlprefix\url{http://www.jmlr.org/papers/v11/sriperumbudur10a.html}.

\bibitem{dudley_real_2002}
\bibinfo{author}{Dudley, R.~M.}
\newblock \emph{\bibinfo{title}{Real {Analysis} and {Probability}}}.
\newblock Cambridge {Studies} in {Advanced} {Mathematics}
  (\bibinfo{publisher}{Cambridge University Press}, \bibinfo{year}{2002}),
  \bibinfo{edition}{2} edn.

\bibitem{schuld_quantum_2019}
\bibinfo{author}{Schuld, M.} \& \bibinfo{author}{Killoran, N.}
\newblock \bibinfo{title}{Quantum {Machine} {Learning} in {Feature} {Hilbert}
  {Spaces}}.
\newblock \emph{\bibinfo{journal}{Phys. Rev. Lett.}}
  \textbf{\bibinfo{volume}{122}}, \bibinfo{pages}{040504}
  (\bibinfo{year}{2019}).
\newblock
  \urlprefix\url{https://link.aps.org/doi/10.1103/PhysRevLett.122.040504}.

\bibitem{muandet_kernel_2017}
\bibinfo{author}{Muandet, K.}, \bibinfo{author}{Fukumizu, K.},
  \bibinfo{author}{Sriperumbudur, B.} \& \bibinfo{author}{Schölkopf, B.}
\newblock \bibinfo{title}{Kernel {Mean} {Embedding} of {Distributions}: {A}
  {Review} and {Beyond}}.
\newblock \emph{\bibinfo{journal}{Foundations and Trends® in Machine
  Learning}} \textbf{\bibinfo{volume}{10}}, \bibinfo{pages}{1--141}
  (\bibinfo{year}{2017}).
\newblock \urlprefix\url{http://arxiv.org/abs/1605.09522}.
\newblock \bibinfo{note}{ArXiv: 1605.09522}.

\bibitem{kubler_quantum_2019}
\bibinfo{author}{Kübler, J.~M.}, \bibinfo{author}{Muandet, K.} \&
  \bibinfo{author}{Schölkopf, B.}
\newblock \bibinfo{title}{Quantum mean embedding of probability distributions}.
\newblock \emph{\bibinfo{journal}{Phys. Rev. Research}}
  \textbf{\bibinfo{volume}{1}}, \bibinfo{pages}{033159} (\bibinfo{year}{2019}).
\newblock
  \urlprefix\url{https://link.aps.org/doi/10.1103/PhysRevResearch.1.033159}.

\bibitem{lloyd_quantum_1999}
\bibinfo{author}{Lloyd, S.} \& \bibinfo{author}{Braunstein, S.~L.}
\newblock \bibinfo{title}{Quantum {Computation} over {Continuous} {Variables}}.
\newblock \emph{\bibinfo{journal}{Physical Review Letters}}
  \textbf{\bibinfo{volume}{82}}, \bibinfo{pages}{1784--1787}
  (\bibinfo{year}{1999}).
\newblock \urlprefix\url{https://link.aps.org/doi/10.1103/PhysRevLett.82.1784}.

\bibitem{suzuki_analyzing_2019}
\bibinfo{author}{Suzuki, Y.} \emph{et~al.}
\newblock \bibinfo{title}{Analyzing feature space via {Pauli} decomposition for
  quantum classifier}.
\newblock \emph{\bibinfo{journal}{arXiv:1906.10467 [quant-ph]}}
  (\bibinfo{year}{2019}).
\newblock \urlprefix\url{http://arxiv.org/abs/1906.10467}.
\newblock \bibinfo{note}{ArXiv: 1906.10467}.

\bibitem{fukumizu_kernel_2007}
\bibinfo{author}{Fukumizu, K.}, \bibinfo{author}{Gretton, A.},
  \bibinfo{author}{Sun, X.} \& \bibinfo{author}{Schölkopf, B.}
\newblock \bibinfo{title}{Kernel {Measures} of {Conditional} {Dependence}}.
\newblock In \emph{\bibinfo{booktitle}{{NIPS}}} (\bibinfo{year}{2007}).

\bibitem{sriperumbudur_injective_2008}
\bibinfo{author}{Sriperumbudur, B.~K.}, \bibinfo{author}{Gretton, A.},
  \bibinfo{author}{Fukumizu, K.}, \bibinfo{author}{Lanckriet, G. R.~G.} \&
  \bibinfo{author}{Schölkopf, B.}
\newblock \bibinfo{title}{Injective {Hilbert} {Space} {Embeddings} of
  {Probability} {Measures}}.
\newblock In \emph{\bibinfo{booktitle}{{COLT}}} (\bibinfo{year}{2008}).

\bibitem{mohamed_learning_2016}
\bibinfo{author}{Mohamed, S.} \& \bibinfo{author}{Lakshminarayanan, B.}
\newblock \bibinfo{title}{Learning in {Implicit} {Generative} {Models}}.
\newblock \emph{\bibinfo{journal}{arXiv:1610.03483 [cs, stat]}}
  (\bibinfo{year}{2016}).
\newblock \urlprefix\url{http://arxiv.org/abs/1610.03483}.
\newblock \bibinfo{note}{ArXiv: 1610.03483}.

\bibitem{diggle_monte_1984}
\bibinfo{author}{Diggle, P.~J.} \& \bibinfo{author}{Gratton, R.~J.}
\newblock \bibinfo{title}{Monte {Carlo} {Methods} of {Inference} for {Implicit}
  {Statistical} {Models}}.
\newblock \emph{\bibinfo{journal}{Journal of the Royal Statistical Society.
  Series B (Methodological)}} \textbf{\bibinfo{volume}{46}},
  \bibinfo{pages}{193--227} (\bibinfo{year}{1984}).
\newblock \urlprefix\url{http://www.jstor.org/stable/2345504}.

\bibitem{hu_quantum_2019}
\bibinfo{author}{Hu, L.} \emph{et~al.}
\newblock \bibinfo{title}{Quantum generative adversarial learning in a
  superconducting quantum circuit}.
\newblock \emph{\bibinfo{journal}{Science Advances}}
  \textbf{\bibinfo{volume}{5}}, \bibinfo{pages}{eaav2761}
  (\bibinfo{year}{2019}).
\newblock
  \urlprefix\url{http://advances.sciencemag.org/content/5/1/eaav2761.abstract}.

\bibitem{cuturi_sinkhorn_2013}
\bibinfo{author}{Cuturi, M.}
\newblock \bibinfo{title}{Sinkhorn {Distances}: {Lightspeed} {Computation} of
  {Optimal} {Transportation} {Distances}}.
\newblock \emph{\bibinfo{journal}{arXiv:1306.0895 [stat]}}
  (\bibinfo{year}{2013}).
\newblock \urlprefix\url{http://arxiv.org/abs/1306.0895}.
\newblock \bibinfo{note}{ArXiv: 1306.0895}.

\bibitem{sinkhorn_relationship_1964}
\bibinfo{author}{Sinkhorn, R.}
\newblock \bibinfo{title}{A {Relationship} {Between} {Arbitrary} {Positive}
  {Matrices} and {Doubly} {Stochastic} {Matrices}}.
\newblock \emph{\bibinfo{journal}{The Annals of Mathematical Statistics}}
  \textbf{\bibinfo{volume}{35}}, \bibinfo{pages}{876--879}
  (\bibinfo{year}{1964}).
\newblock \urlprefix\url{https://projecteuclid.org/euclid.aoms/1177703591}.

\bibitem{peyre_computational_2018}
\bibinfo{author}{Peyré, G.} \& \bibinfo{author}{Cuturi, M.}
\newblock \bibinfo{title}{Computational {Optimal} {Transport}}.
\newblock \emph{\bibinfo{journal}{arXiv:1803.00567 [stat]}}
  (\bibinfo{year}{2018}).
\newblock \urlprefix\url{http://arxiv.org/abs/1803.00567}.
\newblock \bibinfo{note}{ArXiv: 1803.00567}.

\bibitem{weed_sharp_2017}
\bibinfo{author}{Weed, J.} \& \bibinfo{author}{Bach, F.}
\newblock \bibinfo{title}{Sharp asymptotic and finite-sample rates of
  convergence of empirical measures in {Wasserstein} distance}.
\newblock \emph{\bibinfo{journal}{arXiv:1707.00087 [math, stat]}}
  (\bibinfo{year}{2017}).
\newblock \urlprefix\url{http://arxiv.org/abs/1707.00087}.
\newblock \bibinfo{note}{ArXiv: 1707.00087}.

\bibitem{bouland_quantum_2018}
\bibinfo{author}{Bouland, A.}, \bibinfo{author}{Fefferman, B.},
  \bibinfo{author}{Nirkhe, C.} \& \bibinfo{author}{Vazirani, U.}
\newblock \bibinfo{title}{"{Quantum} {Supremacy}" and the {Complexity} of
  {Random} {Circuit} {Sampling}}.
\newblock In \bibinfo{editor}{Blum, A.} (ed.) \emph{\bibinfo{booktitle}{10th
  {Innovations} in {Theoretical} {Computer} {Science} {Conference} ({ITCS}
  2019)}}, vol. \bibinfo{volume}{124} of \emph{\bibinfo{series}{Leibniz
  {International} {Proceedings} in {Informatics} ({LIPIcs})}},
  \bibinfo{pages}{15:1--15:2} (\bibinfo{publisher}{Schloss
  Dagstuhl–Leibniz-Zentrum fuer Informatik}, \bibinfo{address}{Dagstuhl,
  Germany}, \bibinfo{year}{2018}).
\newblock \urlprefix\url{http://drops.dagstuhl.de/opus/volltexte/2018/10108}.

\bibitem{fujii_impossibility_2018}
\bibinfo{author}{Fujii, K.} \emph{et~al.}
\newblock \bibinfo{title}{Impossibility of {Classically} {Simulating}
  {One}-{Clean}-{Qubit} {Model} with {Multiplicative} {Error}}.
\newblock \emph{\bibinfo{journal}{Phys. Rev. Lett.}}
  \textbf{\bibinfo{volume}{120}}, \bibinfo{pages}{200502}
  (\bibinfo{year}{2018}).
\newblock
  \urlprefix\url{https://link.aps.org/doi/10.1103/PhysRevLett.120.200502}.

\bibitem{nielsen_quantum_2011}
\bibinfo{author}{Nielsen, M.~A.} \& \bibinfo{author}{Chuang, I.~L.}
\newblock \emph{\bibinfo{title}{Quantum {Computation} and {Quantum}
  {Information}: 10th {Anniversary} {Edition}}} (\bibinfo{publisher}{Cambridge
  University Press}, \bibinfo{address}{New York, NY, USA},
  \bibinfo{year}{2011}), \bibinfo{edition}{10th} edn.

\bibitem{boykin_universal_1999}
\bibinfo{author}{Boykin, P.~O.}, \bibinfo{author}{Mor, T.},
  \bibinfo{author}{Pulver, M.}, \bibinfo{author}{Roychowdhury, V.} \&
  \bibinfo{author}{Vatan, F.}
\newblock \bibinfo{title}{On {Universal} and {Fault}-{Tolerant} {Quantum}
  {Computing}}.
\newblock \emph{\bibinfo{journal}{arXiv:quant-ph/9906054}}
  (\bibinfo{year}{1999}).
\newblock \urlprefix\url{http://arxiv.org/abs/quant-ph/9906054}.
\newblock \bibinfo{note}{ArXiv: quant-ph/9906054}.

\bibitem{mcclean_barren_2018}
\bibinfo{author}{McClean, J.~R.}, \bibinfo{author}{Boixo, S.},
  \bibinfo{author}{Smelyanskiy, V.~N.}, \bibinfo{author}{Babbush, R.} \&
  \bibinfo{author}{Neven, H.}
\newblock \bibinfo{title}{Barren plateaus in quantum neural network training
  landscapes}.
\newblock \emph{\bibinfo{journal}{Nature Communications}}
  \textbf{\bibinfo{volume}{9}}, \bibinfo{pages}{1--6} (\bibinfo{year}{2018}).
\newblock \urlprefix\url{https://www.nature.com/articles/s41467-018-07090-4}.

\bibitem{grant_initialization_2019}
\bibinfo{author}{Grant, E.}, \bibinfo{author}{Wossnig, L.},
  \bibinfo{author}{Ostaszewski, M.} \& \bibinfo{author}{Benedetti, M.}
\newblock \bibinfo{title}{An initialization strategy for addressing barren
  plateaus in parametrized quantum circuits}.
\newblock \emph{\bibinfo{journal}{Quantum}} \textbf{\bibinfo{volume}{3}},
  \bibinfo{pages}{214} (\bibinfo{year}{2019}).
\newblock \urlprefix\url{https://doi.org/10.22331/q-2019-12-09-214}.

\bibitem{verdon_learning_2019}
\bibinfo{author}{Verdon, G.} \emph{et~al.}
\newblock \bibinfo{title}{Learning to learn with quantum neural networks via
  classical neural networks}.
\newblock \emph{\bibinfo{journal}{arXiv:1907.05415 [quant-ph]}}
  (\bibinfo{year}{2019}).
\newblock \urlprefix\url{http://arxiv.org/abs/1907.05415}.
\newblock \bibinfo{note}{ArXiv: 1907.05415}.

\end{thebibliography}
\end{document}